\documentclass[graybox]{svmult}

\usepackage{type1cm}        
\usepackage{makeidx}         
\usepackage{graphicx}        
\usepackage{multicol}        
\usepackage[bottom]{footmisc}

\usepackage{newtxtext}       %
\usepackage{newtxmath}       
\usepackage{constants}
\newconstantfamily{eps}{symbol=\varepsilon}
\usepackage{bm}
\usepackage{gensymb}
\renewcommand{\vec}[1]{\bm{#1}}

\makeindex             

\begin{document}

\title*{Grain Growth and the Effect of Different Time Scales}
\author{Katayun Barmak, Anastasia Dunca, Yekaterina Epshteyn, Chun
  Liu and Masashi Mizuno}
\authorrunning{K. Barmak, A. Dunca, Y. Epshteyn, C. Liu and M. Mizuno} 
\institute{ Katayun Barmak
\at Columbia University, New York, NY 10027, USA,
\email{kb2612@columbia.edu} \and Anastasia Dunca\at West High School,
Salt Lake City, UT 84103, USA,
  \email{anastasia.d960@slcstudents.org}
\and Yekaterina Epshteyn \at University of Utah, Salt Lake City, UT
84112, USA,
\email{epshteyn@math.utah.edu} \and Chun Liu \at Illinois Institute of
Technology,  Chicago, IL 60616, USA,
\email{cliu124@iit.edu} \and Masashi Mizuno \at Nihon University,
Tokyo 101-8308 JAPAN, \email{mizuno@math.cst.nihon-u.ac.jp} }
%
%
\maketitle

\abstract*{}
{\bf Abstract:} Many technologically useful materials are polycrystals
composed of a myriad of small monocrystalline grains
separated by grain boundaries. Dynamics of grain boundaries play a crucial
role in determining the grain structure and defining the materials properties across multiple
scales. In this work, we consider two models for the motion of grain
boundaries with the dynamic lattice misorientations and  the triple junctions drag, and we conduct extensive numerical study of the
models, as well as present relevant experimental results of grain
growth in thin films.




\section{Introduction}

Many technologically useful materials are polycrystals composed of a myriad of small monocrystalline grains
separated by grain boundaries, see Figures~\ref{fig:lattice1}-\ref{fig:lattice2}.  Dynamics of grain boundaries play a crucial role in determining the grain structure and defining materials properties across multiple
scales. Experimental and computational studies give useful insight
into the geometric features and the crystallography of the grain
boundary network in polycrystalline
microstructures.
\par In this work, we consider two models for the motion of grain
boundaries in a planar network with dynamic lattice misorientations
and with drag of triple junctions. A
classical model for the motion of  grain
boundaries in polycrystalline materials is growth by curvature, as a local
evolution law for the grain boundaries due to Mullins and Herring~\cite{doi:10.1007-978-3-642-59938-5_2,
doi:10.1063-1.1722511,doi:10.1063-1.1722742}, and see work on mean curvature flow, e.g., \cite{MR1100211, MR2024995,MR1100206,MR2815949,MR3556529}.  In
addition, to have a well-posed model for the evolution of the grain
boundary network, one has to impose a separate condition at the triple
junctions where three grain boundaries meet \cite{MR1833000}.
A conventional choice is the
Herring condition which is the natural boundary condition at the
triple points for the
grain boundary network at equilibrium, \cite{MR0485012,MR1240580,MR1833000,MR3612327}, and reference
therein. There are several studies
about grain boundary motion by mean curvature with the Herring
condition at the triple junctions, see for instance
\cite{MR1833000,MR2075985,
  DK:BEEEKT,DK:gbphysrev, MR2772123, MR3729587, BobKohn,
  barmak_grain_2013, MR3316603, MR2561958, MR2272185,
  doi:10.1023-A:1008781611991}.
\par A standard assumption in the theory and simulations of grain
growth is to address only the evolution of the grain boundaries/interfaces themselves and not the
dynamics of the triple junctions. However, recent experimental work
indicates that the motion of the triple junctions together with the anisotropy of
the grain interfaces can have a significant effect on the resulting grain
growth
\cite{barmak_grain_2013}, see work on molecular dynamics simulation
\cite{doi:10.1023-A:1008781611991,doi:10.1016-S1359-6454(01)00446-3}, a recent work on
dynamics of line
defects \cite{Zhang2017PhysRevLett, Zhang2018JMPS,Thomas2019PNAS} and
a relevant work on numerical analysis of a vertex model \cite{doi:10.1137/140999232}.
The current work is a continuation of our previous work
\cite{Katya-Chun-Mzn1, Katya-Chun-Mzn2}, where we proposed a new model for the
evolution of planar grain boundaries, which takes into account dynamic
lattice misorientations (evolving anisotropy of grain boundaries or
``grains rotations'') and
the mobility of the triple
junctions. In \cite{Katya-Chun-Mzn1, Katya-Chun-Mzn2}, using the
energetic variational approach, we derived a system of geometric differential equations to describe
the motion of such grain boundaries, and we established a local well-posedness result, as well as large
time asymptotic behavior for the model. In addition, in
\cite{Katya-Chun-Mzn2}, similar to our
previous work on Grain Boundary Character Distribution, e.g. \cite{DK:gbphysrev,MR3729587} we
conducted some numerical experiments for the 2D grain boundary network in order to 
illustrate the effect of time scales, e.g. of the mobility of triple junctions and of the
dynamics of misorientations on how the grain boundary system decays energy and coarsens with
time (note, in \cite{Katya-Chun-Mzn2}, we studied numerically only the
model with curved grain boundaries). 
Our current goal  is to conduct extensive numerical studies of two
models, a model with curved grain boundaries and a model without
curvature/''vertex model'' of planar grain
boundaries network with the dynamic lattice misorientations and with
the drag of
triple junctions \cite{Katya-Chun-Mzn1, Katya-Chun-Mzn2} and to
further understand the effect of relaxation time scales, e.g.  of the
curvature of grain boundaries, mobility of triple junctions and dynamics of misorientations on how the grain boundary system decays energy and coarsens with
time. We also present and discuss relevant experimental results of
grain growth in thin films.
\par The paper is organized as follows. In
Sections~\ref{sec:1}-\ref{sec:2}, we discuss and review important details and
properties of the two models for grain boundary motion. In
Sections~\ref{sec:3a}, we present and discuss relevant experimental findings
of grain growth in thin films, and in Section~\ref{sec:3b} we conduct extensive numerical studies of the
grain growth models.

\section{Review of the Models with Single Triple Junction}
\label{sec:1}
In this paper we use recently developed models for the evolution of the planar grain
boundary network with dynamic lattice misorientations and
triple junction drag \cite{Katya-Chun-Mzn1, Katya-Chun-Mzn2} to study
the effect of time scales of curvature of grain boundaries, dynamics
of the triple junctions and dynamics of the misorientations on grain growth. Thus, in this section for the reader's convenience,
we
first review the models which were originally
developed in \cite{Katya-Chun-Mzn1, Katya-Chun-Mzn2}.
\begin{figure}
\centerline{\includegraphics[width=2.0in, angle=0]{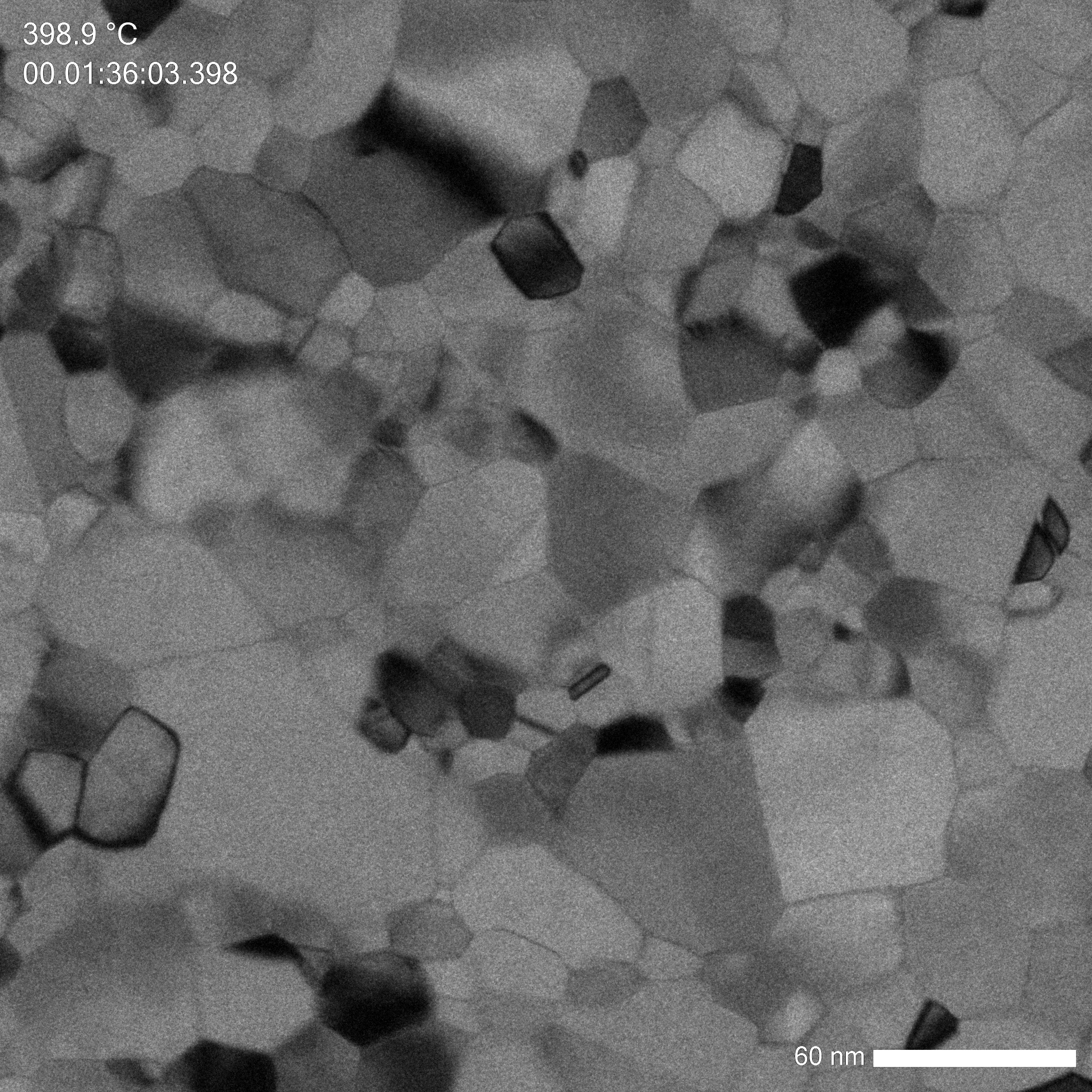}}
\caption{Experimental microstructure: drift-corrected bright-field image of a 50 nm-thick Pt film from an instance of the in-situ grain growth experiment in the transmission electron microscope.}
\label{fig:lattice1}
\end{figure}
\begin{figure}
\vspace{-2.0cm}
\centerline{\includegraphics[width=2.7in,
  angle=0]{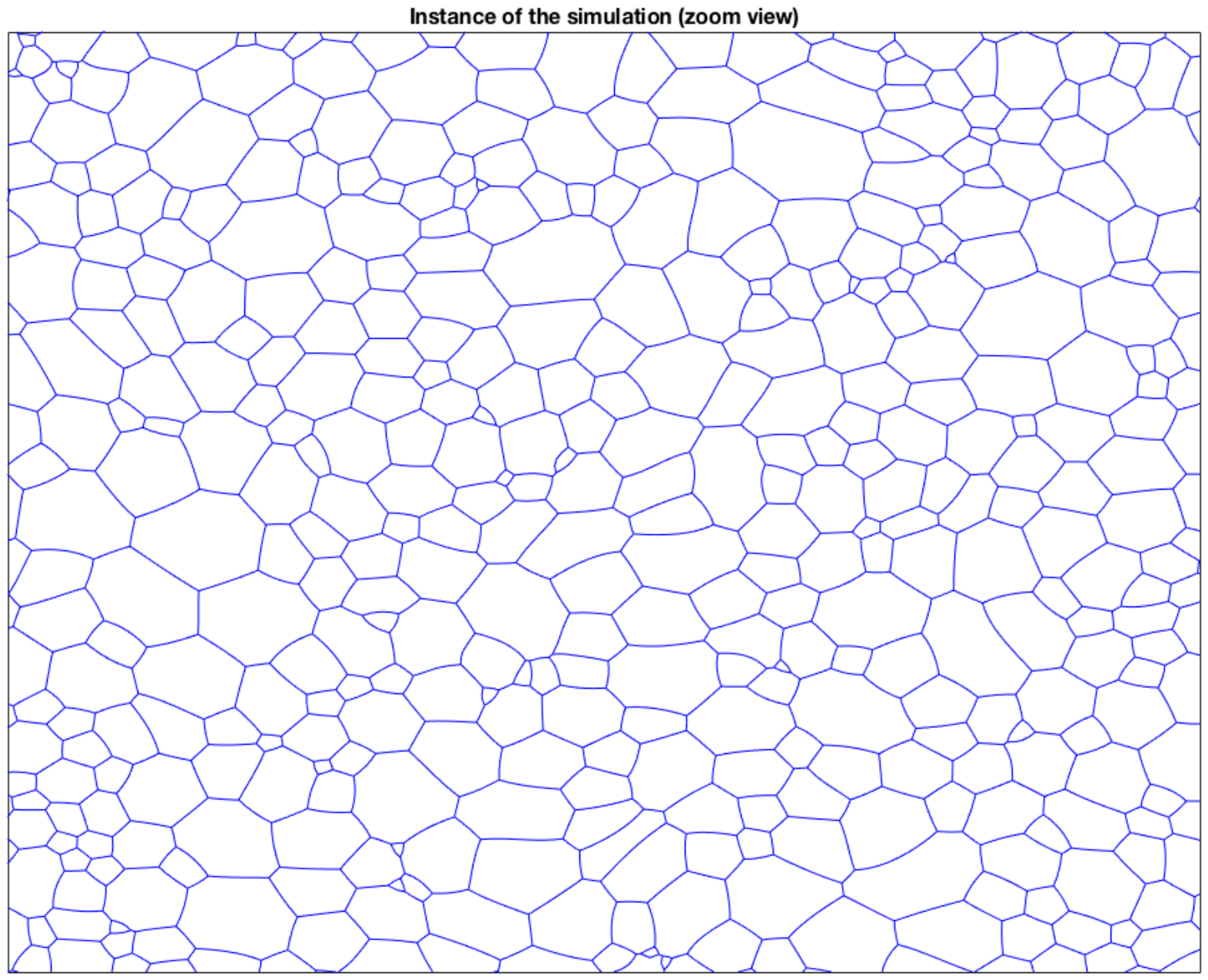}\includegraphics[width=2.7in,
  angle=0]{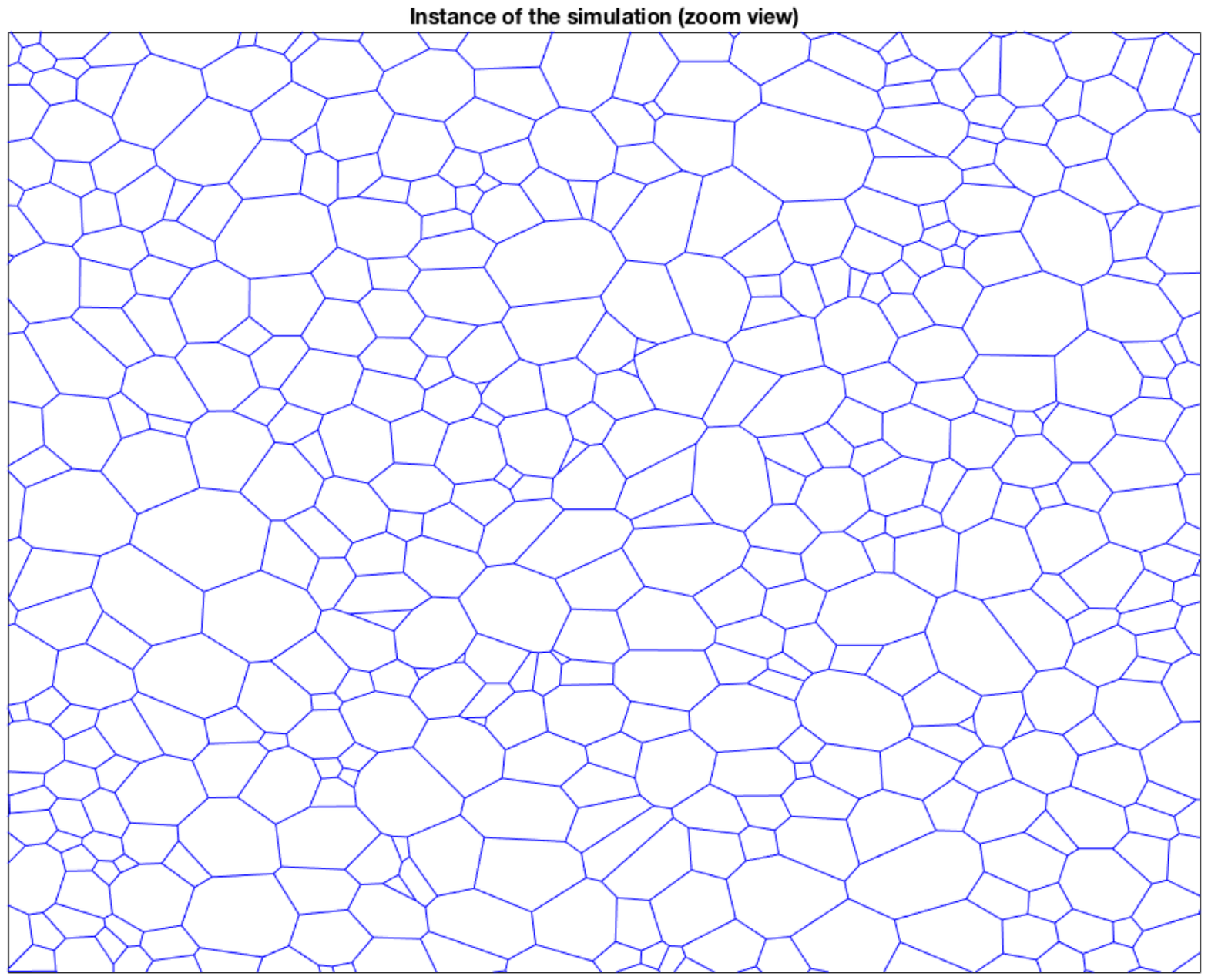}}
\vspace{-2.0cm}
\caption{Left figure - microstructure from simulation, model with
  curvature and finite mobility of the triple junctions
  \eqref{eq:2.20}: example of a time instance during the simulated
  evolution of a cellular network (zoom view). 
Right figure - microstructure from simulation, model without
  curvature \eqref{eq:1.1}: example of a time  instance
  during the simulated evolution of a cellular network (zoom view).}
\label{fig:lattice2}
\end{figure}

 Let us first recall the system for a single triple junction which was
 derived in
 \cite{Katya-Chun-Mzn1}. The total grain boundary energy for such model is
 \begin{equation}
  \label{eq:2.3}
   \sum_{j=1}^3
   \sigma(\Delta^{(j)}\alpha)|\Gamma_t^{(j)}|.
 \end{equation}
 Here, $\sigma:\mathbb{R}\rightarrow\mathbb{R}$ is a given surface tension,
 $\alpha^{(j)}=\alpha^{(j)}(t):[0,\infty)\rightarrow\mathbb{R}$ is
 time-dependent orientations of the grains,
 $\theta=\Delta^{(j)}\alpha:=\alpha^{(j-1)}-\alpha^{(j)}$ is a lattice
 misorientation of the grain boundary $\Gamma^{(j)}_t$ (difference in the orientation between two neighboring grains that share the grain boundary), and
 $|\Gamma_t^{(j)}|$ is the length of $\Gamma_t^{(j)}$.  As a result of
 applying the maximal dissipation principle, in \cite{Katya-Chun-Mzn1},
 the following model was derived,
\begin{equation}
 \label{eq:2.20}
  \left\{
  \begin{aligned}
   v_n^{(j)}
   &=
   \mu
   \sigma(\Delta^{(j)}\alpha)
   \kappa^{(j)},\quad\text{on}\ \Gamma_t^{(j)},\ t>0,\quad j=1,2,3, \\
   \frac{d\alpha^{(j)}}{dt}
   &=
   -\gamma
   \Bigl(
   \sigma_\theta(\Delta^{(j+1)}\alpha)
   |\Gamma_t^{(j+1)}|
   -
   \sigma_\theta(\Delta^{(j)}\alpha)
   |\Gamma_t^{(j)}|
   \Bigr)
   ,\quad
   j=1,2,3,
   \\
   \frac{d\vec{a}}{dt}(t)
   &=
   \eta\sum_{k=1}^3
   \sigma(\Delta^{(k)}\alpha)
   \frac{\vec{b}^{(k)}(0,t)}{|\vec{b}^{(k)}(0,t)|},
   \quad t>0, \\
   \Gamma_t^{(j)}
   &:
   \vec{\xi}^{(j)}(s,t),\quad
   0\leq s\leq 1,\quad
   t>0,\quad
   j=1,2,3, \\
   \vec{a}(t)
   &=
   \vec{\xi}^{(1)}(0,t)
   =
   \vec{\xi}^{(2)}(0,t)
   =
   \vec{\xi}^{(3)}(0,t),
   \quad
   \text{and}
   \quad
   \vec{\xi}^{(j)}(1,t)=\vec{x}^{(j)},\quad
   j=1,2,3.
  \end{aligned}
  \right.
\end{equation}
In \eqref{eq:2.20}, $v_n^{(j)}$, $\kappa^{(j)}$, and
$\vec{b}^{(j)}=\vec{\xi}_s^{(j)}$ are a normal velocity, a curvature and
a tangent vector of the grain boundary $\Gamma_t^{(j)}$, respectively. Note that $s$
is not an arc length parameter of $\Gamma_t^{(j)}$, namely, $\vec{b}^{(j)}$
is \emph{not} necessarily a unit tangent
vector. The vector $\vec{a}=\vec{a}(t):[0,\infty)\rightarrow\mathbb{R}^2$ defines a position of
the triple junction (triple junctions are where three grain boundaries meet), $\vec{x}^{(j)}$ is a position of the end point of
the grain boundary. The three independent relaxation time scales
$\mu,\gamma,\eta>0$  (curvature, misorientation and 
triple junction dynamics) are regarded as
positive constants. Further, we assume in \eqref{eq:2.20}, $\alpha^{(0)}=\alpha^{(3)}$,
$\alpha^{(4)}=\alpha^{(1)}$ and $\vec{b}^{(4)}=\vec{b}^{(1)}$, for
simplicity. We also use notation $|\cdot|$ for a standard
euclidean vector norm. The complete details about model
\eqref{eq:2.20} can be found in the earlier work
\cite[Section 2]{Katya-Chun-Mzn1}.
Next, in \cite{Katya-Chun-Mzn1},  the curvature effect was relaxed,  by taking the limit
$\mu\rightarrow\infty$, and the reduced model without
curvature was derived,
\begin{equation}
 \label{eq:1.1}
 \left\{
  \begin{aligned}
   \frac{d\alpha^{(j)}}{dt}
   &=
   -
   \gamma{
   \Bigl(
   \sigma_\theta(\Delta^{(j+1)}\alpha)|\vec{b}^{(j+1)}|
   -
   \sigma_\theta(\Delta^{(j)}\alpha)|\vec{b}^{(j)}|
   \Bigr)
   }
   ,
   \quad
   j=1,2,3, \\
   \frac{d\vec{a}}{dt}(t)
   &=
   \eta
   \sum_{j=1}^3
   \sigma(\Delta^{(j)}\alpha)
   \frac{\vec{b}^{(j)}}{|\vec{b}^{(j)}|},
   \quad t>0, \\
   \vec{a}(t)+\vec{b}^{(j)}(t)
   &=
   \vec{x}^{(j)},
   \quad
   j=1,2,3.
  \end{aligned}
 \right.
\end{equation}

 In (\ref{eq:1.1}), we consider $\vec{b}^{(j)}(t)$ as a grain
 boundary. Note that, similar to (\ref{eq:2.20}),  the system of equations \eqref{eq:1.1} can also be
 derived from the energetic variational principle for the total grain boundary
 energy (\ref{eq:2.3}) (with $|\Gamma_t^{(j)}|$ replaced by
 $|\vec{b}^{(j)}|$). 
\begin{remark}
\par  {\bf a)} As was discussed in \cite{Katya-Chun-Mzn1},  the reduced
model without curvature effect (\ref{eq:1.1}) is not a
  standard ODE system. This is the ODE system where each variable is
  locally constrained. Moreover, local well-posedness result  (e.g. local
  existence result) for the
  original model (\ref{eq:2.20}) will not imply local well-posedness
  result for the reduced system  (\ref{eq:1.1}). It is not known
  if the reduced model  (\ref{eq:1.1}) is  a small perturbation of
  (\ref{eq:2.20}).
\par {\bf b)}   The reduced model (\ref{eq:1.1}) captures the dynamics of the
   orientations\\/misorientations and the triple junctions. At the same
   time, it was more accessible for the mathematical analysis than the model
   (\ref{eq:2.20}). In addition, the system (\ref{eq:1.1}) is
   a generalization to higher dimension and dynamic misorientations
   of the model from \cite{DK:gbphysrev,MR2772123}. In this paper, we will
   compare and contrast through extensive numerical studies the model
   with the curvature effect
   \eqref{eq:2.20} and the reduced model \eqref{eq:1.1}.
\end{remark}
To establish local well-posedness result for
 model \eqref{eq:1.1} in \cite{Katya-Chun-Mzn1}, the surface tension
 $\sigma$ was assumed to be  $C^3$, positive,
and minimized at $0$, namely,
\begin{equation}
 \label{eq:2.Assumption1}
  \sigma(\theta)\geq \sigma(0)>0,
\end{equation}
for $\theta\in\mathbb{R}$. In addition, it was assumed convexity of $\sigma(\theta)$, for all
$\theta\in\mathbb{R},$
\begin{equation}
 \label{eq:2.Assumption2}
  \sigma_\theta(\theta)\theta\geq0,\qquad
  \text{and}
  \quad
  \sigma_{\theta\theta}(0)>0,
\end{equation}
and
\begin{equation}
 \label{eq:2.Assumption3}
  \sigma_\theta(\theta)=0
  \ \text{if and only if}\
  \theta=0.
\end{equation}

Let us review some of the important theoretical results
established for \eqref{eq:1.1} in previous work
\cite{Katya-Chun-Mzn1,Katya-Chun-Mzn2}. First,  consider the
equilibrium state of the system \eqref{eq:1.1}, namely,
\begin{equation}
 \label{eq:1.3}
  \left\{
   \begin{aligned}
   0
   &
   =
    -
   \left(
   \sigma_\theta(\Delta^{(j+1)}\alpha_{\infty})
   |\vec{b}^{(j+1)}_{\infty}|
   -
   \sigma_\theta(\Delta^{(j)}\alpha_{\infty})
   |\vec{b}^{(j)}_{\infty}|
   \right)
   , \quad j=1, 2, 3,
   \\
   \vec{0}
   &=
   \sum_{j=1}^3
    \sigma(\Delta^{(j)}\alpha_{\infty})
   \frac{\vec{b}^{(j)}_{\infty}}{|\vec{b}^{(j)}_{\infty}|}, \\
   \vec{a}_{\infty}
   &=
   \vec{x}^{(1)}-\vec{b}^{(1)}_{\infty}
   =
   \vec{x}^{(2)}-\vec{b}^{(2)}_{\infty}
   =
   \vec{x}^{(3)}-\vec{b}^{(3)}_{\infty}.
   \end{aligned}
  \right.
\end{equation}
As in \cite{Katya-Chun-Mzn1,Katya-Chun-Mzn2},  assume,  for each $i=1,2,3$,
\begin{equation}
 \label{eq:1.4}
  \left|
   \sum_{j=1, j\neq i}^3
   \frac{\vec{x}^{(j)}-\vec{x}^{(i)}}{|\vec{x}^{(j)}-\vec{x}^{(i)}|}
  \right|
  >1.
\end{equation}
The assumption \eqref{eq:1.4} implies that fixed points $\vec{x}^{(1)}, \vec{x}^{(2)}$ and $\vec{x}^{(3)}$ can not belong to the single line. Furthermore,
\eqref{eq:1.4} is equivalent to the condition that in the triangle with vertices
$\vec{x}^{(1)}\vec{x}^{(2)}\vec{x}^{(3)}$, all three angles are less than
$\frac{2\pi}{3}$. Next, from the assumptions \eqref{eq:1.4},
\eqref{eq:2.Assumption2}-\eqref{eq:2.Assumption3},
associated
equilibrium system \eqref{eq:1.3} becomes,
\begin{equation}
 \label{eq:1.5}
  \left\{
   \begin{aligned}
    \sum_{j=1}^3
    \frac{\vec{b}_\infty^{(j)}}{|\vec{b}_\infty^{(j)}|}&=\vec{0},
    \\
    \vec{a}_\infty+\vec{b}^{(j)}_\infty
    &=
    \vec{x}^{(j)},
    \quad
    j=1,2,3.
   \end{aligned}
  \right.
\end{equation}
In \cite{Katya-Chun-Mzn1}, it was shown that the assumptions \eqref{eq:2.Assumption2}-\eqref{eq:2.Assumption3} imply
$\alpha^{(1)}_\infty=\alpha^{(3)}_\infty=\alpha^{(3)}_\infty$, hence
$\Delta^{(j)}\alpha_\infty=0$ for $j=1,2,3$ for the equilibrium system
\eqref{eq:1.3} (note, that in this case,  for the purpose of
  mathematical modeling, one can still assume a ``fictitious'' grain
  boundary with the same orientation on each side of the grain
  boundary. In addition, in this work we study the grain boundary system
  before it reaches a state of constant orientations, see Section \ref{sec:3}.)
\par We also have energy dissipation principle for the system \eqref{eq:1.1},
\begin{proposition}
 [Energy dissipation{\cite[Proposition 5.1]{Katya-Chun-Mzn1}}]
 \label{prop:3.4}
 Let $(\vec{\alpha},\vec{a})$ be a solution of \eqref{eq:1.1} on $0\leq
 t\leq T$, and let $E(t)$, given by \eqref{eq:2.3}, be the total grain
 boundary energy of the system. Then, for all $0<t\leq T$,
 \begin{equation}
  \label{eq:3.6}
   E(t)
   +
   \frac1\gamma
   \int_0^t
   \left|
   \frac{d\vec{\alpha}}{dt}(\tau)
   \right|^2
   \,d\tau
   +
   \frac1\eta
   \int_0^t
   \left|
    \frac{d\vec{a}}{dt}(\tau)
   \right|^2
   \,d\tau \\
  =
   E(0).
 \end{equation}
\end{proposition}
Next, define, constant as in \cite{Katya-Chun-Mzn2},
\begin{equation}
 \label{eq:4.1}
  \Cl{const:4.1}
  :=
  \inf
  \left\{
   \sum_{j=1}^3|\vec{x}^{(j)}-\vec{a}|
   :
   \text{There exists}\
   j=1,2,3\
    \text{such that}\
    |\vec{a}-\vec{a}_\infty|\geq\frac12|\vec{b}_\infty^{(j)}|
  \right\}.
\end{equation}
Assume also that an initial data $(\vec{\alpha}_0,\vec{a}_0)$ satisfies,
\begin{equation}
  \label{eq:4.4}
  E(0)= \sum_{j=1}^3
  \sigma(\Delta^{(j)}\alpha_0)
   |\vec{a}_0-\vec{x}^{(j)}|
   <
   \sigma(0)
   \Cr{const:4.1}.
 \end{equation}
Then, one can establish the global existence result for the model \eqref{eq:1.1},
\begin{theorem}
 [Global existence{\cite[Theorem 4.1]{Katya-Chun-Mzn2}}]
 \label{thm:4.4}
 Let $\vec{x}^{(1)}$, $\vec{x}^{(2)}$, $\vec{x}^{(3)}\in \mathbb{R}^2$,
 $\vec{a}_0\in\mathbb{R}^2$, and $\vec{\alpha}_0\in\mathbb{R}^3$ be the initial data
 for the system \eqref{eq:1.1}. Assume
 \eqref{eq:1.4},  and let $\vec{a}_\infty$ be a unique
 solution of the equilibrium system \eqref{eq:1.5}. Further, assume
 condition \eqref{eq:4.4}.
 Then there exists a unique global in time solution $(\vec{\alpha},
 \vec{a})$ of \eqref{eq:1.1}.
\end{theorem}
We also have the following large time asymptotic behavior results for the
solution of system \eqref{eq:1.1},
\begin{proposition}
[Large Time Asymptotic{\cite[Proposition 5.1]{Katya-Chun-Mzn2}}]
 \label{prop:5.1}
 Let $\vec{x}^{(1)}$, $\vec{x}^{(2)}$, $\vec{x}^{(3)}\in \mathbb{R}^2$,
 $\vec{a}_0\in\mathbb{R}^2$, and $\vec{\alpha}_0\in\mathbb{R}^3$ be the
 initial data for the system \eqref{eq:1.1}. We assume that the initial
 data satisfy \eqref{eq:4.4}, and we also impose the same assumptions as
 in Theorem \ref{thm:4.4}. Define $\alpha_\infty$ as,
 \begin{equation}
  \label{eq:5.1}
  \alpha_\infty
   :=
   \frac{\alpha_0^{(1)}+\alpha_0^{(2)}+\alpha_0^{(3)}}{3}.
 \end{equation}
Let $\vec{a}_\infty$ be a solution of the equilibrium system
 \eqref{eq:1.5} and $(\vec{\alpha}, \vec{a})$ be a time global solution of
 \eqref{eq:1.1}. Then,
 \begin{equation}
  \label{eq:5.2}
   \vec{\alpha}(t)\rightarrow\alpha_\infty(1,1,1),\quad
   \vec{a}(t)\rightarrow\vec{a}_\infty,
 \end{equation}
 as $t\rightarrow\infty$.
\end{proposition}
\begin{theorem}
[Large Time Asymptotic{\cite[Theorem 5.1]{Katya-Chun-Mzn2}}]
 \label{thm:5.1}
 There is a small constant $\Cl[eps]{eps:5.2}>0$ such that,  if
 $|\vec{\alpha}_0-\vec{\alpha}_\infty|+|\vec{a}_0-\vec{a}_\infty|<\Cr{eps:5.2}$, then the associated
 global solution $(\vec{\alpha},\vec{a})$ of the system \eqref{eq:1.1} satisfies,
 \begin{equation}
  \label{eq:5.24}
  |\vec{\alpha}(t)-\vec{\alpha}_\infty|
   +
   |\vec{a}(t)-\vec{a}_\infty|
   \leq \Cr{const:5.5}e^{-\lambda^{\star} t},
 \end{equation}
 for some positive constants $\Cl{const:5.5}, \lambda^{\star}>0$.
\end{theorem}
\begin{remark}
 The
 decay order $\lambda^{\star}$ in \eqref{eq:5.24} is explicitly estimated as,
 \begin{equation}\label{eq:5.41decay}
  \lambda^{\star} \geq \lambda,
 \end{equation}
 where $\lambda$ depends on $\gamma$, $\eta$, $\sigma_{\theta
   \theta}(0)$, $\sigma(0)$ and on the smallest positive eigenvalues
 of the linearized operators for the equations of the orientation
 ${\vec {\alpha}}$ and of the triple junction ${\bf a}$.
\end{remark}
\begin{corollary}
[Large Time Asymptotic{\cite[Corollary 5.1]{Katya-Chun-Mzn2}}]
 Under the same assumption as in Theorem \ref{thm:5.1}, the associated
 grain boundary energy $E(t)$ satisfies,
 \begin{equation}
  \label{eq:5.32}
  E(t)-E_\infty
   \leq
   \Cr{const:5.7}
   e^{-\lambda^{\star}t},
 \end{equation}
 for some positive constant $\Cl{const:5.7}>0$, where
 \begin{equation*}
   E_\infty
   :=
   \sigma(0)\sum_{j=1}^3|\vec{b}_\infty^{(j)}|.
 \end{equation*}
\end{corollary}

\begin{proof}
 For the reader's convenience, we will review the proof from
 \cite{Katya-Chun-Mzn2}. Since $\alpha^{(1)}_\infty=\alpha^{(2)}_\infty=\alpha^{(3)}_\infty$, we obtain
 \begin{equation}
  \label{eq:5.33}
  \begin{split}
   E(t)-E_\infty
   &=
   \sum_{j=1}^3
   \left(
   \sigma(\Delta^{(j)}\alpha(t))
   |\vec{b}^{(j)}(t)|
   -
   \sigma(0)
   |\vec{b}_\infty^{(j)}|
   \right) \\
   &\leq
   \sum_{j=1}^3
   \left(
   \sigma(0)
   |
   \vec{b}^{(j)}(t)
   -
   \vec{b}_\infty^{(j)}
   |
   +
   \left(
   \sigma(\Delta^{(j)}\alpha(t))
   -
   \sigma(0)
   \right)
   |\vec{b}^{(j)}(t)|
   \right) \\
   &\leq
   \sum_{j=1}^3
   \left(
   \sigma(0)
   |
   \vec{a}^{(j)}(t)
   -
   \vec{a}_\infty
   |
   +
   \left(
   \Cl{const:5.12}
   |\Delta^{(j)}\alpha(t)|
   \right)
   |\vec{b}^{(j)}(t)|
   \right) \\
   &\leq
   \sum_{j=1}^3
   \left(
   \sigma(0)
   |
   \vec{a}^{(j)}(t)
   -
   \vec{a}_\infty
   |
   +
   2\Cr{const:5.12}
   |\vec{b}^{(j)}(t)|
   |\vec{\alpha}(t)-\vec{\alpha}_\infty|
   \right),
  \end{split}
 \end{equation}
 where $\Cr{const:5.12}=\sup_{|\theta|<2\Cr{eps:5.2}}|\sigma_\theta(\theta)|$.
 Using the dissipation estimate \eqref{eq:3.6} and the exponential decay estimate \eqref{eq:5.24}, we obtain
 \eqref{eq:5.32}.
\end{proof}

\begin{remark}
 \label{rm:5.33}
 Note, that the obtained exponential decay to equilibrium,
see estimates (\ref{eq:5.24}) and (\ref{eq:5.32}) was obtained by
considering linearized problem, Lemma 5.1 in
\cite{Katya-Chun-Mzn2}. Consideration of the model with curvature -
with finite $\mu$, (\ref{eq:2.20}) and 
of the nonlinear problem instead of linearized problem could lead to potential power laws estimates
for the decay rates.
See also
discussion  and numerical studies in Section \ref{sec:3}.
\end{remark}

\section{Extension to Grain Boundary Network}
\label{sec:2}

In this section, we review the extension of the results to a grain boundary network
$\{\Gamma_t^{(j)}\}$. As in \cite{Katya-Chun-Mzn1,Katya-Chun-Mzn2},  we define the total grain
boundary energy of the network, like,
\begin{equation} \label{eq:6.4e}
 E(t)
  =
  \sum_{j}
  \sigma
  (
  \Delta^{(j)}\alpha
  )
  |\Gamma_t^{(j)}|,
\end{equation}
where $\Delta^{(j)}\alpha$ is a misorientation, a difference between the
lattice orientation of the two neighboring grains which form the grain
boundary $\Gamma^{(j)}_t$. Then, the energetic variational principle
leads to a full model (network model analog of a single triple
junction system \eqref{eq:2.20}),
\begin{equation}
 \label{eq:6.4}
  \left\{
  \begin{aligned}
   v_n^{(j)}
   &=
   \mu
   \sigma
   (
   \Delta^{(j)}\alpha
   )
   \kappa^{(j)},\quad\text{on}\ \Gamma_t^{(j)},\ t>0, \\
   \frac{d\alpha^{(k)}}{dt}
   &=
   -\gamma
   \frac{\delta E}{\delta \alpha^{(k)}},
   \\
   \frac{d\vec{a}^{(l)}}{dt}
   &=
   \eta
   \sum_{\vec{a}^{(l)}\in\Gamma^{(j)}_t}
   \left(\sigma(\Delta^{(j)}\alpha)\frac{\vec{b}^{(j)}}{|\vec{b}^{(j)}|}\right),
   \quad t>0.
  \end{aligned}
  \right.
\end{equation}

As in \cite{Katya-Chun-Mzn1}, we consider the relaxation parameters,
$\mu\rightarrow\infty$, and we further assume that the energy density
$\sigma(\theta)$ is an even function with respect to
the misorientation $\theta= \Delta^{(j)}\alpha$, that is, the
misorientation effects are symmetric with respect to the difference
between the lattice orientations. Then, the problem \eqref{eq:6.4}
reduces to (network model analog of a single triple
junction system \eqref{eq:1.1} ),
\begin{equation}
 \label{eq:6.7} \left\{
  \begin{aligned}
   \Gamma_t^{(j)} &\ \text{is a line segment between some}\
   \vec{a}^{(l_{j,1})}\ \text{and}\ \vec{a}^{(l_{j,2})}, \\
   \frac{d\alpha^{(k)}}{dt}
   &=
   -
   \gamma
   \sum_{\substack{\text{ grain with } \alpha^{(k')}\ \text{is the
         neighbor of the grain with }\ \alpha^{(k)} \\
   \Gamma^{(j)}_t\ \text{is formed by the two grains with }\ \alpha^{(k)}\ \text{and}\ \alpha^{(k')}}}
   |\Gamma^{(j)}_t|
   \sigma_{\theta}(\alpha^{(k)}-\alpha^{(k')}), \\
   \frac{d\vec{a}^{(l)}}{dt}
   &=
   \eta
   \sum_{\vec{a}^{(l)}\in\Gamma^{(j)}_t}
   \left(\sigma(\Delta^{(j)}\alpha)\frac{\vec{b}^{(j)}}{|\vec{b}^{(j)}|}\right).
  \end{aligned}
  \right.
\end{equation}
To obtain the global solution of the system \eqref{eq:6.7} in
\cite{Katya-Chun-Mzn2}, we studied the system before the critical events, and we first
 considered an associated energy minimizing state,
$(\alpha_\infty^{(k)},\vec{a}_\infty^{(l)})$ of \eqref{eq:6.7}.  The
critical events are the disappearance events, such as, e.g., disappearance of
the grains and/or grain boundaries during coarsening of the system,
facet interchange and splitting of unstable junctions.
Then,
$(\alpha_\infty^{(k)},\vec{a}_\infty^{(l)})$ satisfies,
\begin{equation}
 \label{eq:6.8}
 \left\{
  \begin{aligned}
   \Gamma_\infty^{(j)}&\ \text{is a line segment between some}\
   \vec{a}^{(l_{j,1})}_\infty\ \text{and}\ \vec{a}_\infty^{(l_{j,2})}, \\
   0
   &=
   -
   \gamma
   \sum_{\substack{\text{ grain with } \alpha^{(k')}\ \text{is the
         neighbor of the grain with }\ \alpha^{(k)} \\
   \Gamma^{(j)}_t\ \text{is formed by the two grains with }\ \alpha^{(k)}\ \text{and}\ \alpha^{(k')}}}
   |\Gamma_\infty^{(j)}|
   \sigma_{\theta}(\alpha_\infty^{(k)}-\alpha_\infty^{(k')}), \\
   \vec{0}
   &=
   \eta
   \sum_{\vec{a}_\infty^{(l)}\in\Gamma^{(j)}_\infty}
   \left(\sigma(\Delta^{(j)}\alpha_\infty)\frac{\vec{b}_\infty^{(j)}}{|\vec{b}_\infty^{(j)}|}\right).
  \end{aligned}
  \right.
\end{equation}
Hence, the total energy $E_\infty$ of the grain boundary network
(\ref{eq:6.8}) is
\begin{equation}
 \label{eq:6.9}
 \begin{split}
  E_\infty
  &=
  \sum_{j}
  \sigma(\Delta^{(j)}\alpha_\infty)
  |\vec{b}_\infty^{(j)}|
  =
  \inf
  \biggl\{
  \sum_{j}
  \sigma(\Delta^{(j)}\alpha)
  |\vec{b}^{(j)}|
  \biggr\}.
 \end{split}
\end{equation}
\begin{remark}
Note, we assumed in (\ref{eq:6.7})-(\ref{eq:6.8}) that the total number of grains, grain boundaries and triple
    junctions are the same as in the initial configuration (assumption
    of no critical events in the network).
\end{remark}
Further, if there is a neighborhood $U^{(l)}\subset\mathbb{R}^2$ of $\vec{a}_\infty^{(l)}$
such that
\begin{equation}
 \label{eq:6.10}
  E_\infty
  <
  \sum_{j}
  |\vec{b}^{(j)}|
\end{equation}
for all $\vec{a}^{(l)}\in U^{(l)}$, one can obtain a priori estimate for
the triple junctions, and, hence,  obtain the time global solution of
\eqref{eq:6.7}. Note that, the assumption \eqref{eq:6.10} is related to
the boundary condition of the line segments
$\Gamma_t^{(j)}$. Further, if the energy minimizing state is unique,
then we can proceed with the same argument as in Lemma 4.1 in \cite{Katya-Chun-Mzn2}, and
obtain the global solution \eqref{eq:6.7} near the energy minimizing
state.

\begin{remark}
 Note that, the solution of \eqref{eq:6.8} may not be unique even though the grain
 orientations are constant (misorientation is zero) \cite{Katya-Chun-Mzn2}.

 \end{remark}

The asymptotics of the grain boundary networks are rather
nontrivial. Our arguments in \cite{Katya-Chun-Mzn2} were based on the uniqueness of the equilibrium
state \eqref{eq:1.5}. However,  we do not know the uniqueness of solutions of
the equilibrium state for the grain boundary network
\eqref{eq:6.8}. Thus,  in general we cannot take a full limit for the
large time asymptotic behavior of the solution of the network model
\eqref{eq:6.7}. But, one can show, the following result instead,
\begin{corollary}
[{\cite[Corollary 6.1]{Katya-Chun-Mzn2}}]
 In a grain boundary network \eqref{eq:6.7}, assume that the initial
 configuration is sufficiently close to an associated energy minimizing
 state \eqref{eq:6.8}.
 Then,  there
 is a global solution $(\alpha^{(k)}, \vec{a}^{(l)})$ of
 \eqref{eq:6.7}. Furthermore, there exists a time sequence
 $t_n\rightarrow\infty$ such that $(\alpha^{(k)}(t_n),
 \vec{a}^{(l)}(t_n))$ converges to an associated equilibrium
 configuration \eqref{eq:6.8}.
\end{corollary}

\section{Experiments and Numerical Simulations}\label{sec:3}
In this section we present results of some experiments in thin films and numerical study of
the grain growth using models of planar grain boundary network from
Section \ref{sec:2}.  The energetics and connectivity
of the grain boundary network play a crucial role in determining the properties of a material
across multiple scales, see also Sections
\ref{sec:1}-\ref{sec:2}. Therefore, our main focus here is
to develop a better understanding of the energetic properties of the
experimental and computational microstructures.
\subsection{Experimental Results: Grain Boundary Character Distribution}\label{sec:3a}
To more fully characterize a microstructure, it is necessary to
consider the types and energies of the constituent grain
boundaries, in addition to geometric features such as grain size.  Indeed, experiments and simulations over the past 30+
years have led to the discovery and notion of the Grain Boundary
Character Distribution (GBCD) \cite
{ISI:000071740700001,Adams1999,rohrer_distribution_2004,
  ISI:000225119800166,rohrer_influence_2005,ISI:000248758300020}. {\it
  The GBCD, denoted by $\rho$, is an empirical
  distribution of the relative area (in 3D) or relative length (in 2D)
  of interface/grain boundaries with a given misorientation and
  boundary normal.} The GBCD can be viewed as a leading statistical descriptor to characterize the texture of the grain boundary network (see, e.g., \cite{ISI:000071740700001,Adams1999,ISI:000225119800166,rohrer_influence_2005,ISI:000248758300020, DK:gbphysrev}).
\begin{figure}
\centerline{\includegraphics[width=4.0in, angle=0]{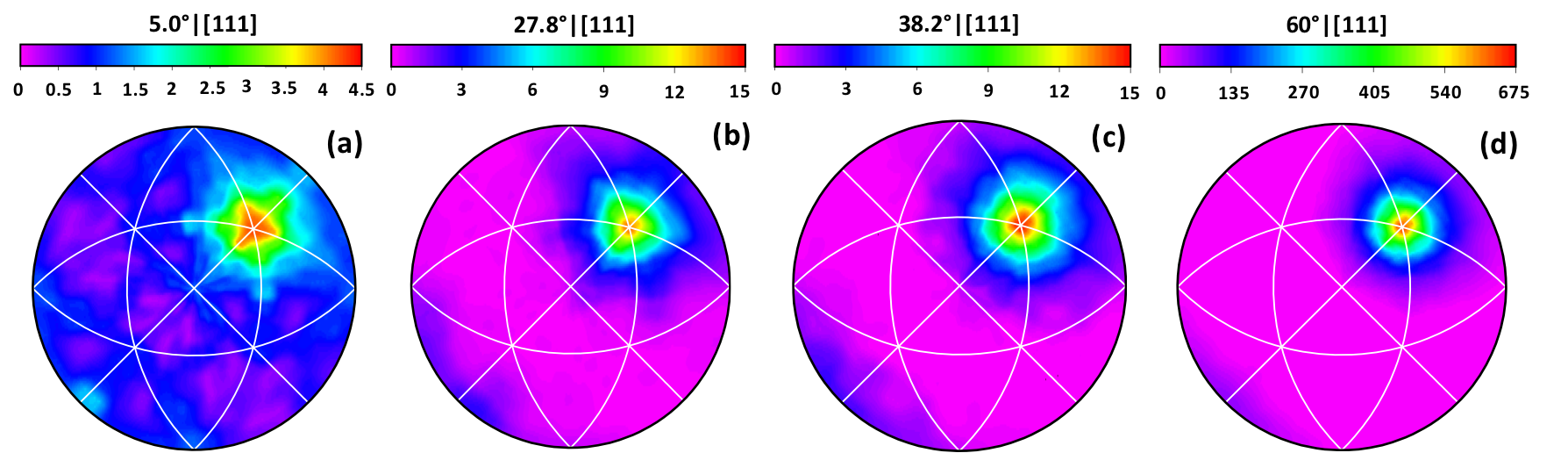}}
\caption{Experiments: (a-d) Grain boundary character distribution of 100 nm-thick, as-deposited Al film with a mean grain size of approximately 100 nm for four given misorientations. Misorientations are specified as angle-axis pairs. Pseudosymmetry cleanup of the crystal orientation maps was used in generating the figures. The scale is multiples of random distribution.}
\label{fig:exp_GBCD}
\end{figure}
\par Figure~\ref{fig:exp_GBCD} presents the GBCD for four different misorientations for an as-deposited aluminum film with near random orientation distribution.  The details of film deposition, sample preparation and precession electron diffraction crystal orientation mapping in the transmission electron microscope are given in \cite{Rohrer_AlGBCD_2017}. However, in contrast to \cite{Rohrer_AlGBCD_2017}, the orientation data were subjected to the same cleanup procedure as for the grain size distribution, namely the pseudosymmetry cleanup procedure detailed in \cite{Liu_AlGrain_2014} with the exception of the 60\degree|[111] boundaries, which are clearly abundant and should not be removed. The minimum grain size of the dilation cleanup step was 20 pixels.
\par Given that grain boundaries have five crystallographic degrees of freedom - three to specify the misorientation across the grain boundary, and two to define the normal to the boundary, the two-dimensional graphical presentation of the GBCD as in Figure~\ref{fig:exp_GBCD} is achieved in the following manner. To begin, a given misorientation is selected, for example, 5\degree|[111]. The rotation axis, here [111], is given by the Miller indices of the crystallographic direction that is common to both grains on either side of the boundary. The misorientation angle is usually, but not always, chosen to be within the fundamental zone of misorientations, which for cubic crystals has a minimum of zero and a maximum of 62.8\degree. Common choices of angles are either those of low angle boundaries, with rotation angles of less than 15 degrees, or those of coincident site lattice (CSL) type. In Figure~\ref{fig:exp_GBCD}, the selected rotation angles about the [111] axis of 27.8\degree|[111], 38.2\degree|[111] and 60\degree|[111] correspond to CSL designations $\Sigma13b$, $\Sigma7$ and $\Sigma3$, respectively. The numerical value in the $\Sigma$ designation is the reciprocal of the number of atomic sites that are coincident in the crystallographic plane perpendicular to rotation axis.  For face centered cubic crystals, the Miller indices of this plane are the same as the Miller indices of the misorientation axis, e.g., the (111) plane for the [111] rotation axis. The letters a or b in the $\Sigma$ designation then indicate different angle-axis pairs with the same number of coincident sites. Note that the CSL designation does not specify the grain boundary plane that is present in the sample; rather it specifies only a given misorientation.
\par Next, the grain boundary planes present in the experimental sample for the given misorientation are represented by the crystallographic directions normal to the planes in standard stereographic projections, such as those in Figure~\ref{fig:exp_GBCD}. The use of stereographic projection rather than other types of projections in single crystal or bicrystal crystallography of materials has been common practice. Its choice is based on the fact that it is an angle-preserving projection that does not depend on the size of the crystal (from nano to macro).  For cubic crystals, the standard projection has the [001] cubic crystal axis pointing out of the page thereby projecting onto the page as the origin of the plot at the center of the (projected equatorial) circle. In Figure~\ref{fig:exp_GBCD}, the [100] crystallographic axis points to the right, and the [010] crystallographic axis points up, thereby defining a right-handed axis set.
\par The stereographic projections of the boundary plane normals such
as those of Figure~\ref{fig:exp_GBCD} then show the abundance of
grain boundary plane normals in multiples of random distribution
(MRD) on the thermal scale. The MRD is similar to a probability
density plot, but its integrated value is 2, rather than 1, since
every grain boundary segment is counted twice, once for the grain on
one side of the boundary and once for the grain on the other side of
the boundary.  When the direction normal to the boundary plane and the
misorientation axis are the same, the grain boundary is termed a
twist boundary, since the axis of rotation is normal to the observed
boundary plane.  In Figure~\ref{fig:exp_GBCD}, a high relative
intensity is seen at the position of the [111] twist boundaries for
all four selected misorientations. If, on the other hand, the high
intensities were seen as bands along a great circle ninety degrees
away from the chosen misorientation axis, then the boundaries would
have been designated as tilt boundaries, with the misorientation axis
in the plane of the grain boundary. In effect, GBCD plots such as
those of Figure~\ref{fig:exp_GBCD} make manifest texture formation in the grain boundary network, see also
numerical experiments Section~\ref{sec:3b}. 
\par The most striking feature of Figure~\ref{fig:exp_GBCD} is the
very high abundance of 60\degree|[111] boundaries, which show a
population of several hundred times MRD.  Given that the majority of
the boundary planes were also found to be (111), this sample is said
to have a large population of coherent $\Sigma3$, or so-called
coherent twin boundaries.  $\Sigma3$ boundaries constitute
approximately one quarter of all the boundaries in this sample. In
contrast, for a "bulk" aluminum sample, i.e., in an aluminum sample
with mean grain size of 23 $\mu$m, the population of $\Sigma3$
boundaries is more than ten times lower \cite{Rohrer_AlGBCD_2017}.
The very high population of $\Sigma3$ boundaries in the thin film
sample of Figure~\ref{fig:exp_GBCD} is likely a result of the
structure forming processes that take place during film deposition,
rather than a result of normal grain growth. The evolution of the
grain boundary network and the GBCD of this sample towards equilibrium
or steady-state  will be determined by the dynamics of the grain
boundaries and the relaxation time scales for the boundary curvature,
misorientation and triple junctions, for which models and simulations
are presented in the current work.  We note that in experimental samples where GBCD has reached steady-state, the GBCD averaged over its five crystallographic parameters is
  inversely related to the grain boundary energy density similar to
  the GBCD extracted from grain growth models, Section~\ref{sec:3b}.
Laboratory-based experimental quantification of grain boundary
dynamics via in-situ annealing experiments similar to the experiment
in Figure~\ref{fig:lattice1}, together with intermittent mapping of
crystal orientations for determination of the evolving GBCD will be the key to
connecting more closely experimental findings to mathematical and computational models of
grain growth. These experiments are the subject of the ongoing research.
\subsection{Numerical Experiments}\label{sec:3b}
Here, we present several numerical experiments to illustrate the
effects of different time scales, such as 
the dynamic orientations/misorientations (grains ``rotations'')  and mobility of the
triple junctions, as well as we compare the grain growth model with
curvature \eqref{eq:6.4} and model without curvature \eqref{eq:6.7},
as described in Sections \ref{sec:1}-\ref{sec:2}.  
\begin{figure}[hbtp]
\centering
\vspace{-1.0cm}
\includegraphics[width=2.1in]{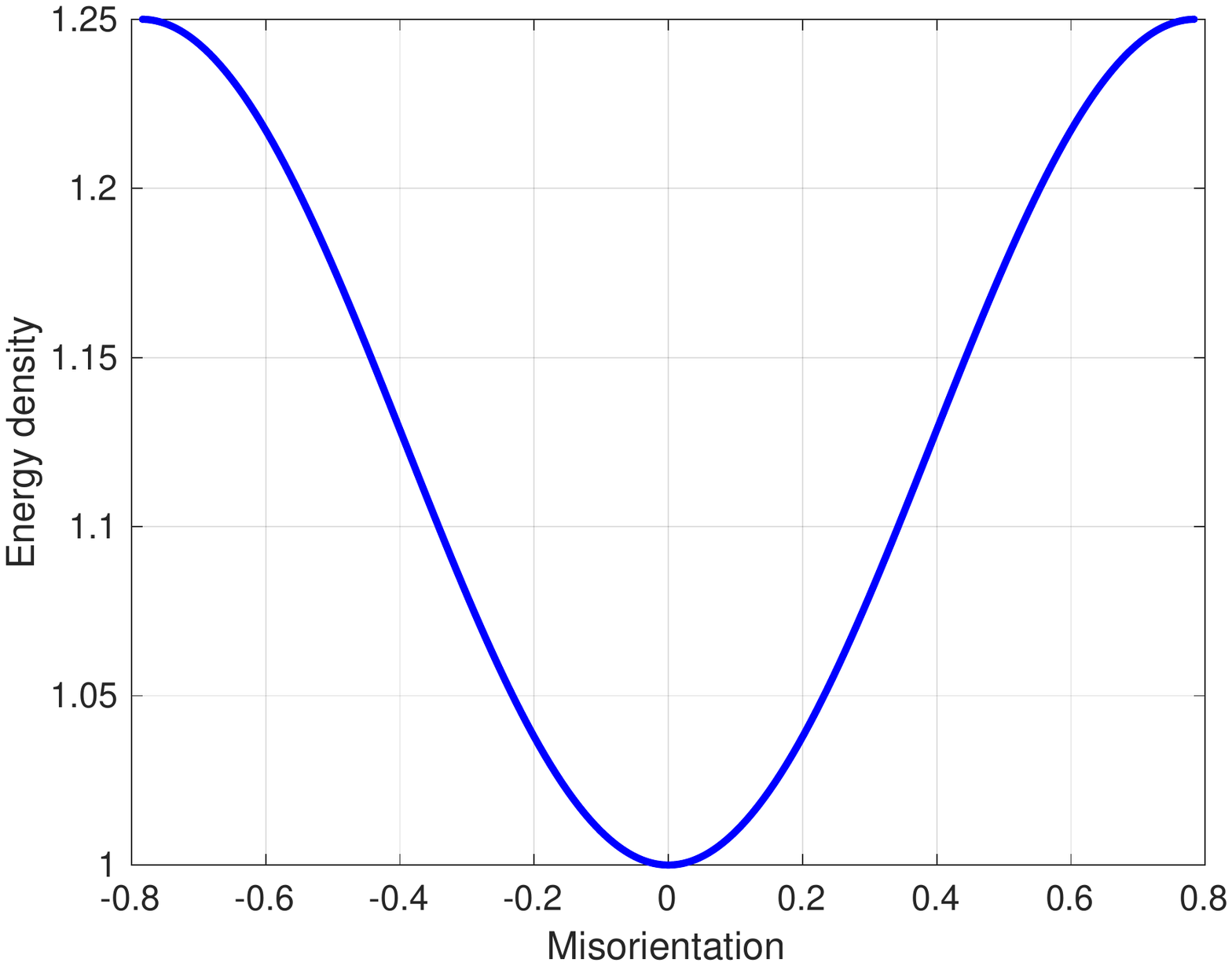}
\includegraphics[width=2.1in]{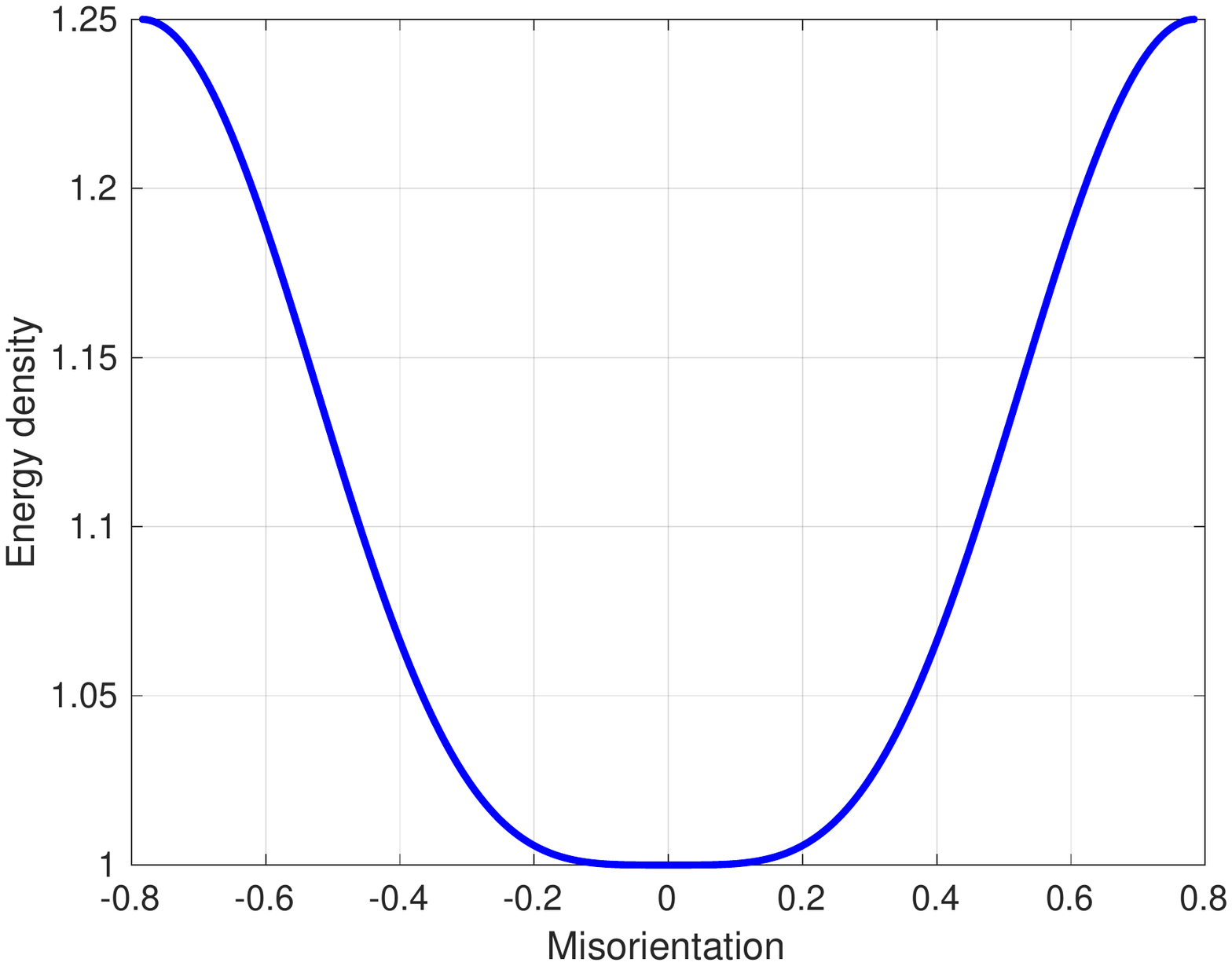}
\vspace{-1.0cm}
\caption{\footnotesize Grain boundary energy density function
  $\sigma(\Delta \alpha)$: {\it (a) Left plot,} $\sigma=1+0.25\sin^2(2\Delta \alpha)$ and {\it (b) Right
    plot,} $\sigma=1+0.25\sin^4(2\Delta \alpha).$}\label{gbend}
\end{figure}

\begin{figure}[hbtp]
\centering
\vspace{-1.8cm}
\includegraphics[width=2.1in]{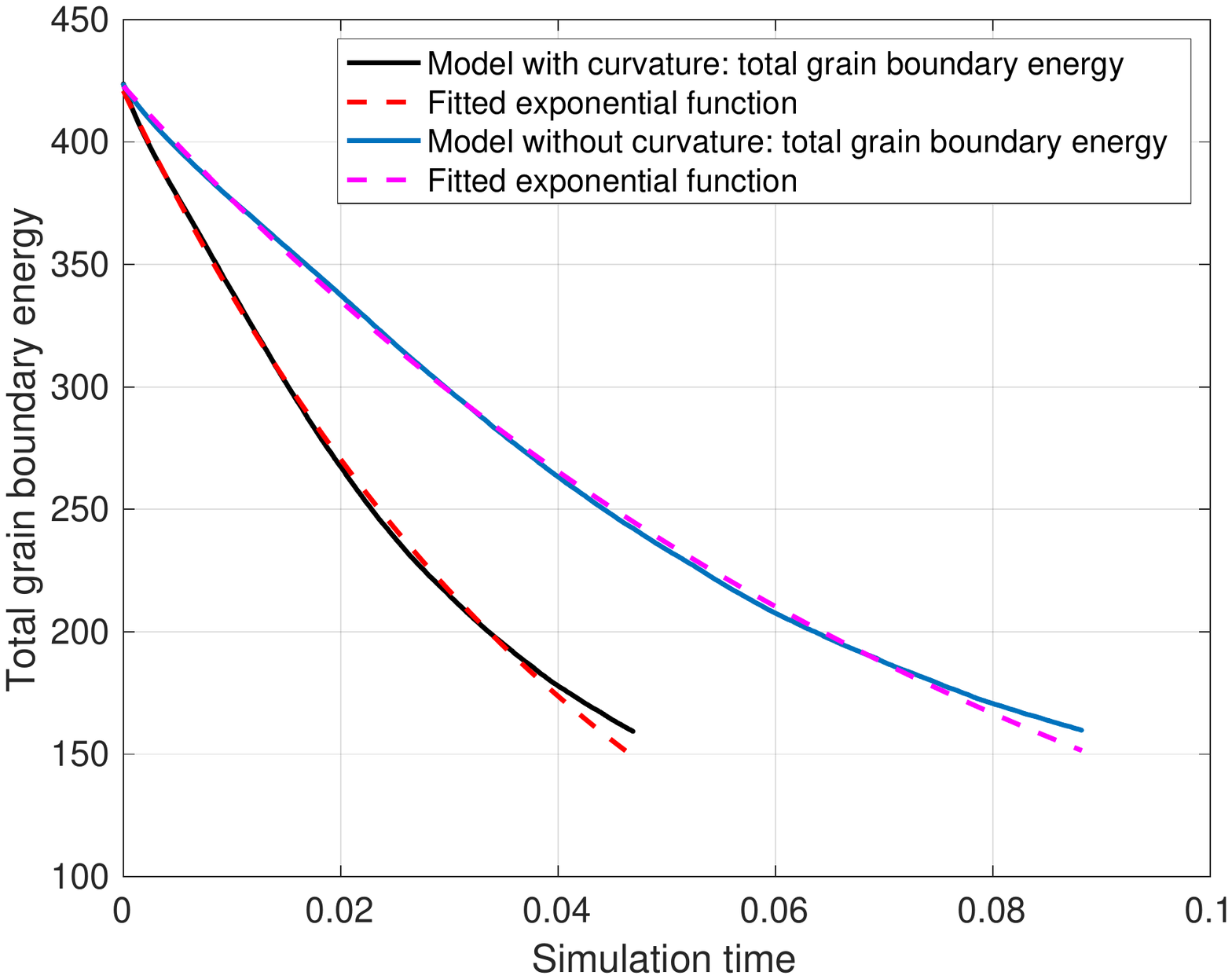}
\includegraphics[width=2.1in]{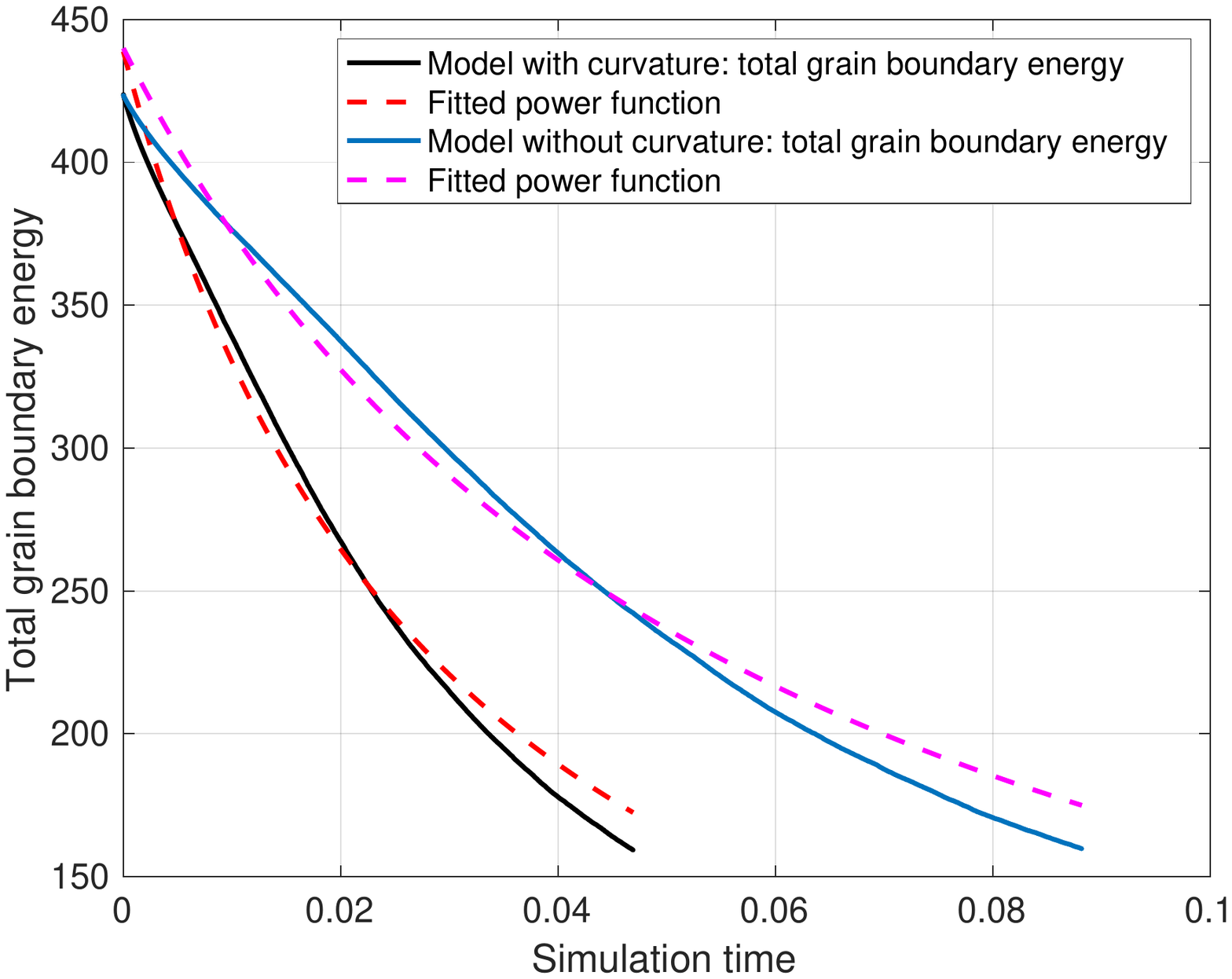}
\vspace{-1.8cm}
\caption{\footnotesize One run of $2$D trial with $10000$ initial
  grains: {\it (a) Left plot,} Total grain boundary energy plot, model
  with curvature  (solid black) versus fitted  exponential decaying function
  $y(t)=420.9\exp(-22.13t)$ (dashed red). Total grain boundary energy plot, model
  without curvature  (solid blue) versus fitted  exponential decaying function
  $y(t)=422.8\exp(-11.64t)$ (dashed magenta); {\it (b) Right
    plot},  Total grain boundary energy plot, model
  with curvature  (solid black) versus fitted  power function
   $y_1(t)=438.8797(1.0+32.9489t)^{-1}$ (dashed red). Total grain boundary energy plot, model
  without curvature  (solid blue) versus fitted power function
$y_1(t)=439.9588(1.0+17.1792t)^{-1}$ (dashed magenta). 
 Mobility of the triple
junctions is
$\eta=10$ and the misorientation parameter $\gamma=1$.  Grain boundary
energy density $\sigma=1+0.25\sin^2(2\Delta \alpha).$}\label{fig7}
\end{figure}

\begin{figure}[hbtp]
\centering
\vspace{-1.8cm}
\includegraphics[width=2.1in]{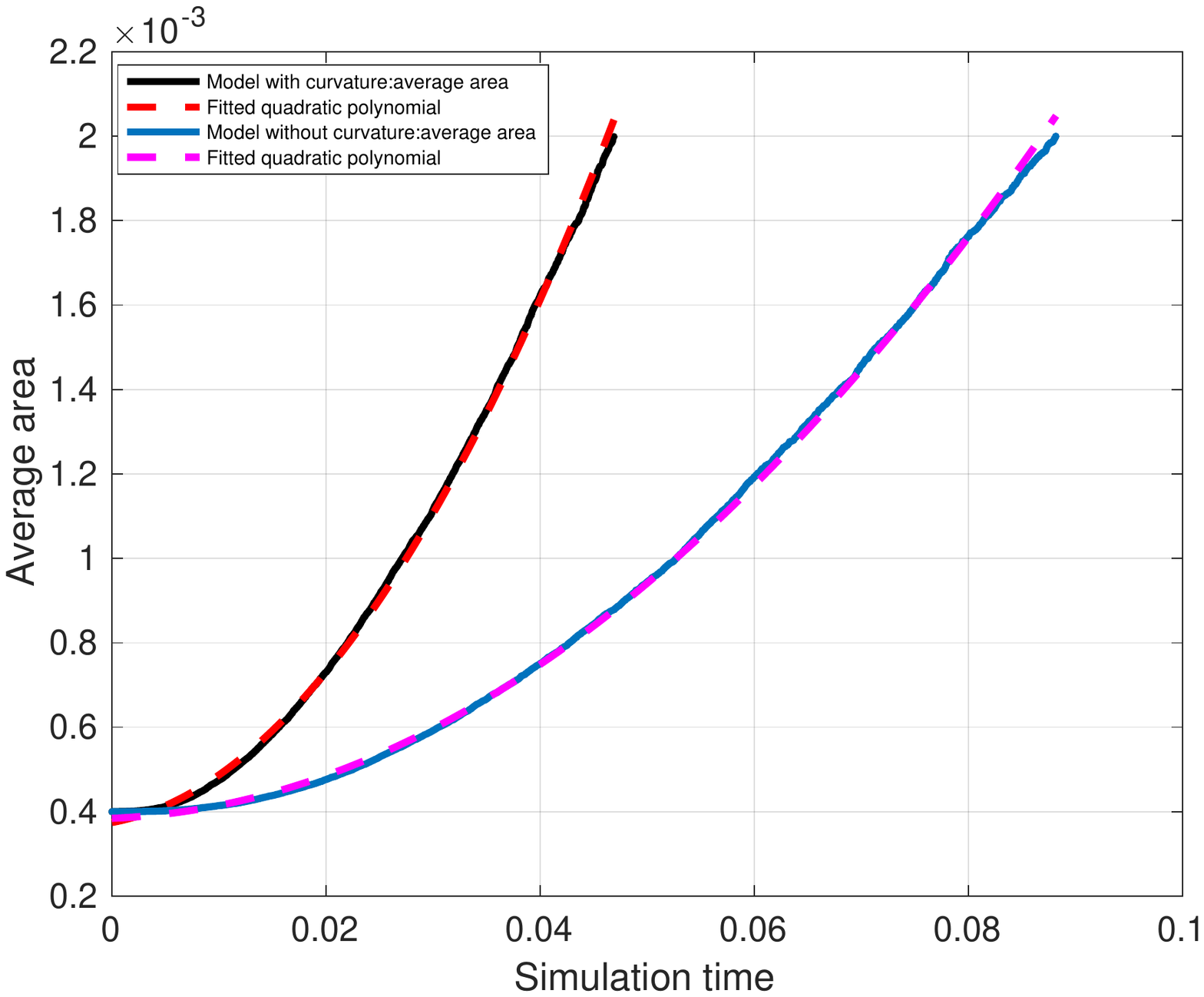}
\includegraphics[width=2.1in]{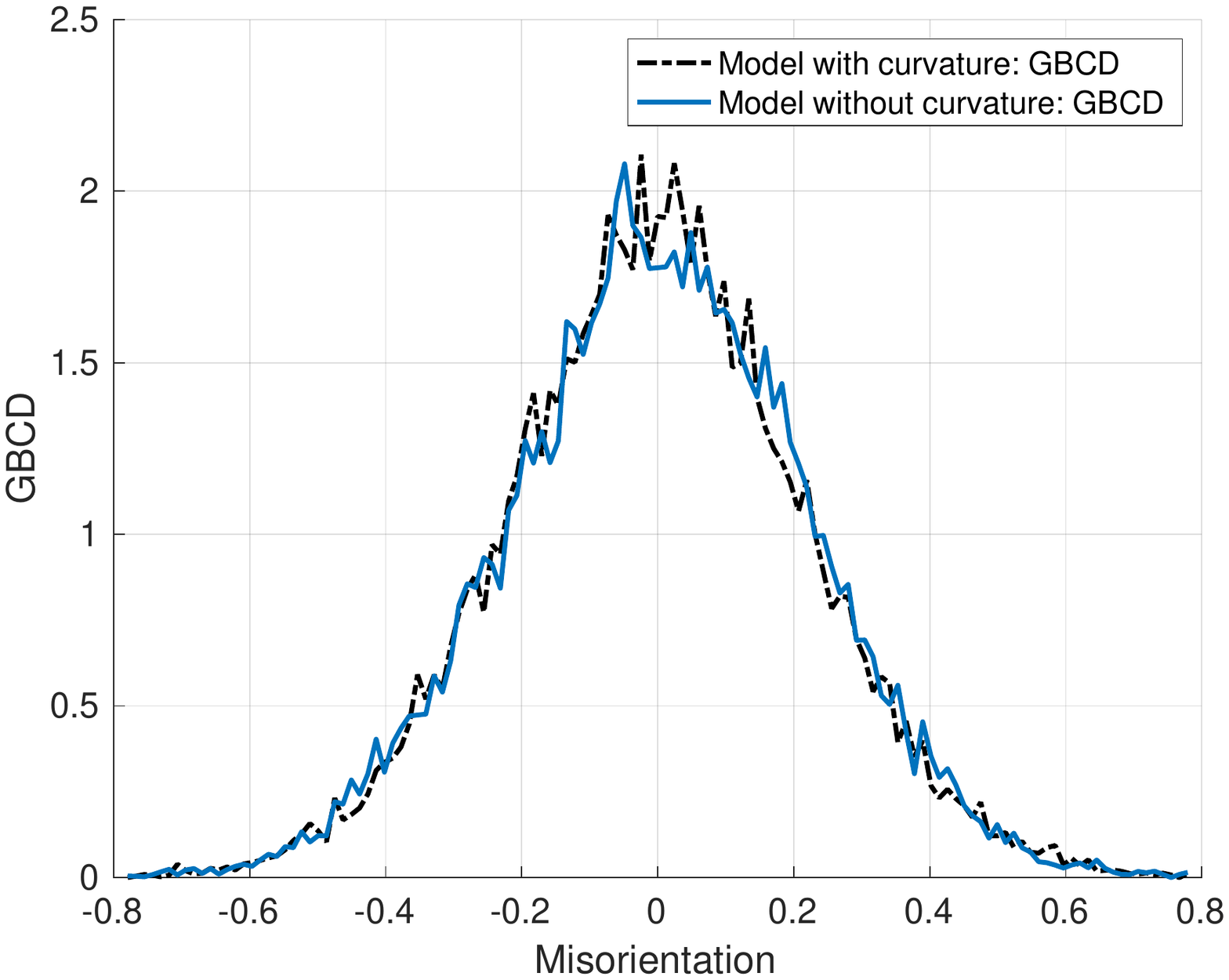}
\vspace{-1.8cm}
\caption{\footnotesize {\it (a) Left plot,} One run of $2$D trial with $10000$ initial
  grains: Growth of the average area of the
  grains, model with curvature (solid black) versus fitted  quadratic
  polynomial function  $y(t)=0.6575t^2+0.004668t+0.0003745$ (dashed
  red). Growth of the average area of the
  grains, model without curvature (solid blue) versus fitted  quadratic
  polynomial function  $y(t)=0.2025t^2+0.001016t+0.0003844$ (dashed magenta);
 {\it (b) Right plot}, GBCD (black curve,  model with curvature) and  GBCD (blue curve,  model without curvature) at $T_{\infty}$ averaged over 3 runs of $2$D trials with $10000$ initial
  grains. Mobility of the triple junctions is $\eta=10$ and the
  misorientation parameter $\gamma=1$.  Grain boundary
energy density $\sigma=1+0.25\sin^2(2\Delta \alpha).$}\label{fig8}
\end{figure}

\begin{figure}[hbtp]
\centering
\vspace{-1.8cm}
\includegraphics[width=2.1in]{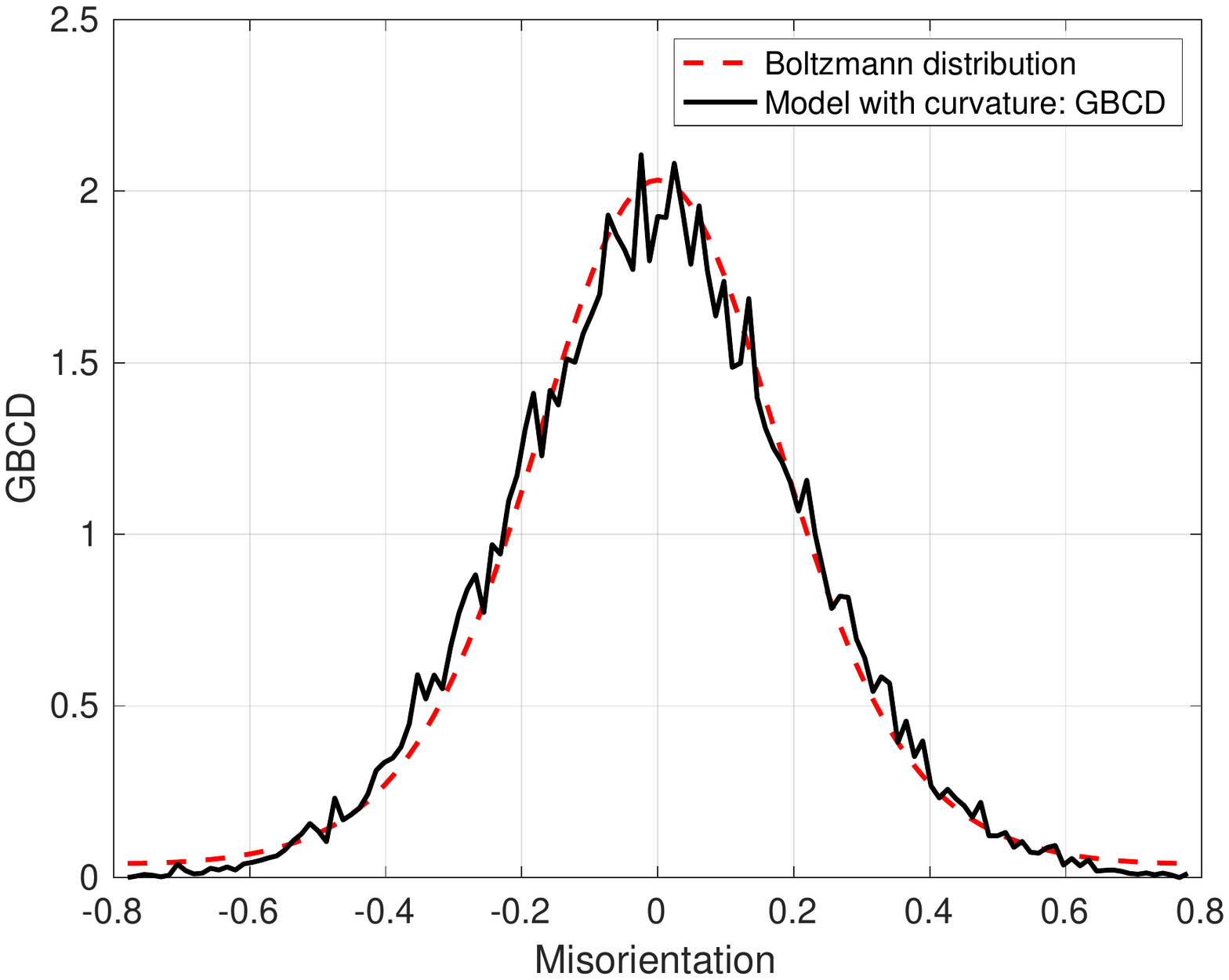}
\includegraphics[width=2.1in]{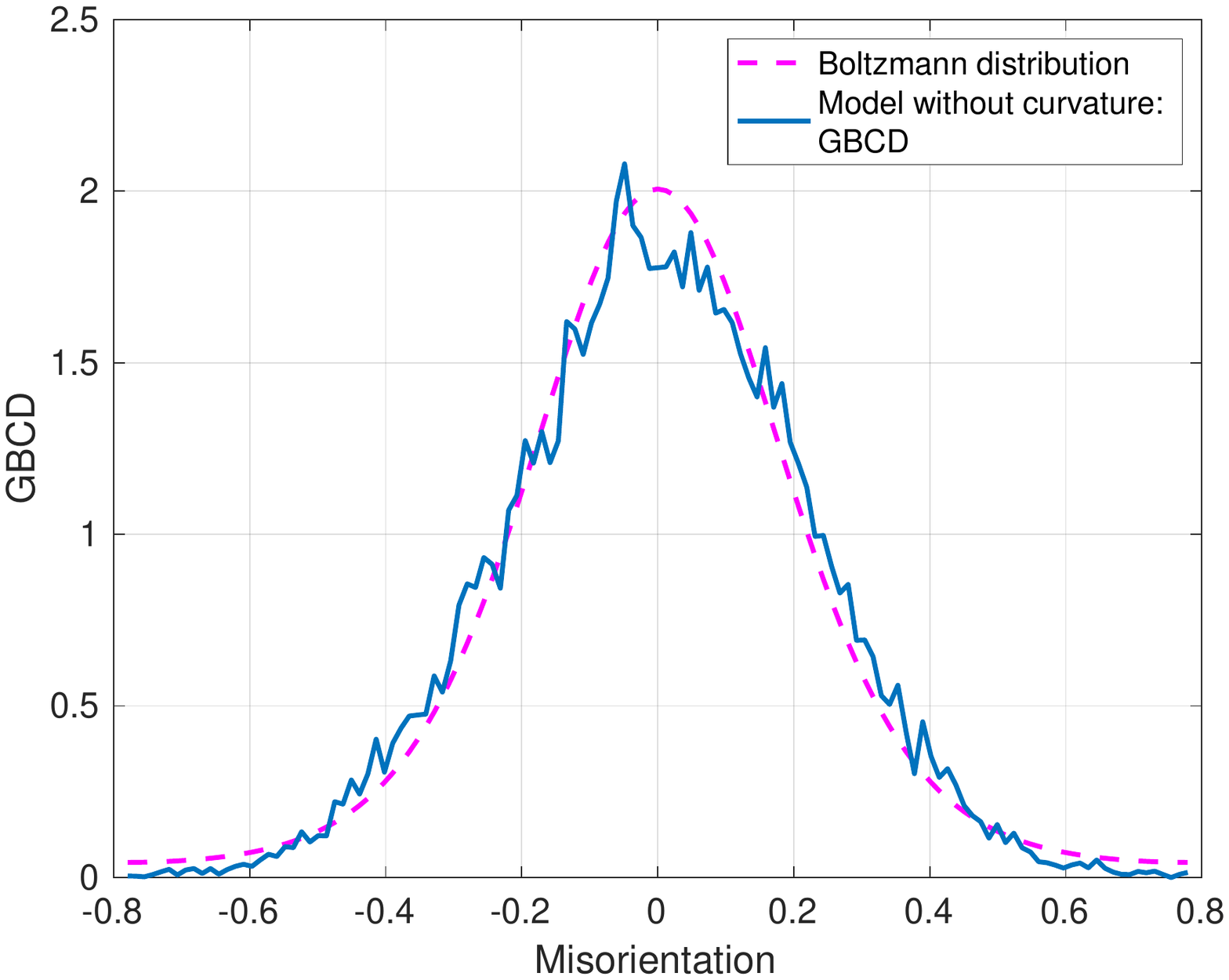}
\vspace{-1.8cm}
\caption{\footnotesize {\it (a) Left plot,} Model with curvature,  GBCD (black curve) at $T_{\infty}$ averaged over 3 runs of $2$D trials with $10000$ initial
  grains versus Boltzmann distribution with ``temperature''-
$D\approx 0.0641$ (dashed red curve). 
 {\it (b) Right plot}, Model without curvature,  GBCD (blue curve) at $T_{\infty}$ averaged over 3 runs of $2$D trials with $10000$ initial
  grains versus Boltzmann distribution with ``temperature''-
$D\approx 0.0655$ (dashed magenta curve). Mobility of the triple
junctions is $\eta=10$ and the misorientation parameter $\gamma=1$.  Grain boundary
energy density $\sigma=1+0.25\sin^2(2\Delta \alpha).$}\label{fig9}
\end{figure}

\begin{figure}[hbtp]
\centering
\vspace{-1.8cm}
\includegraphics[width=2.1in]{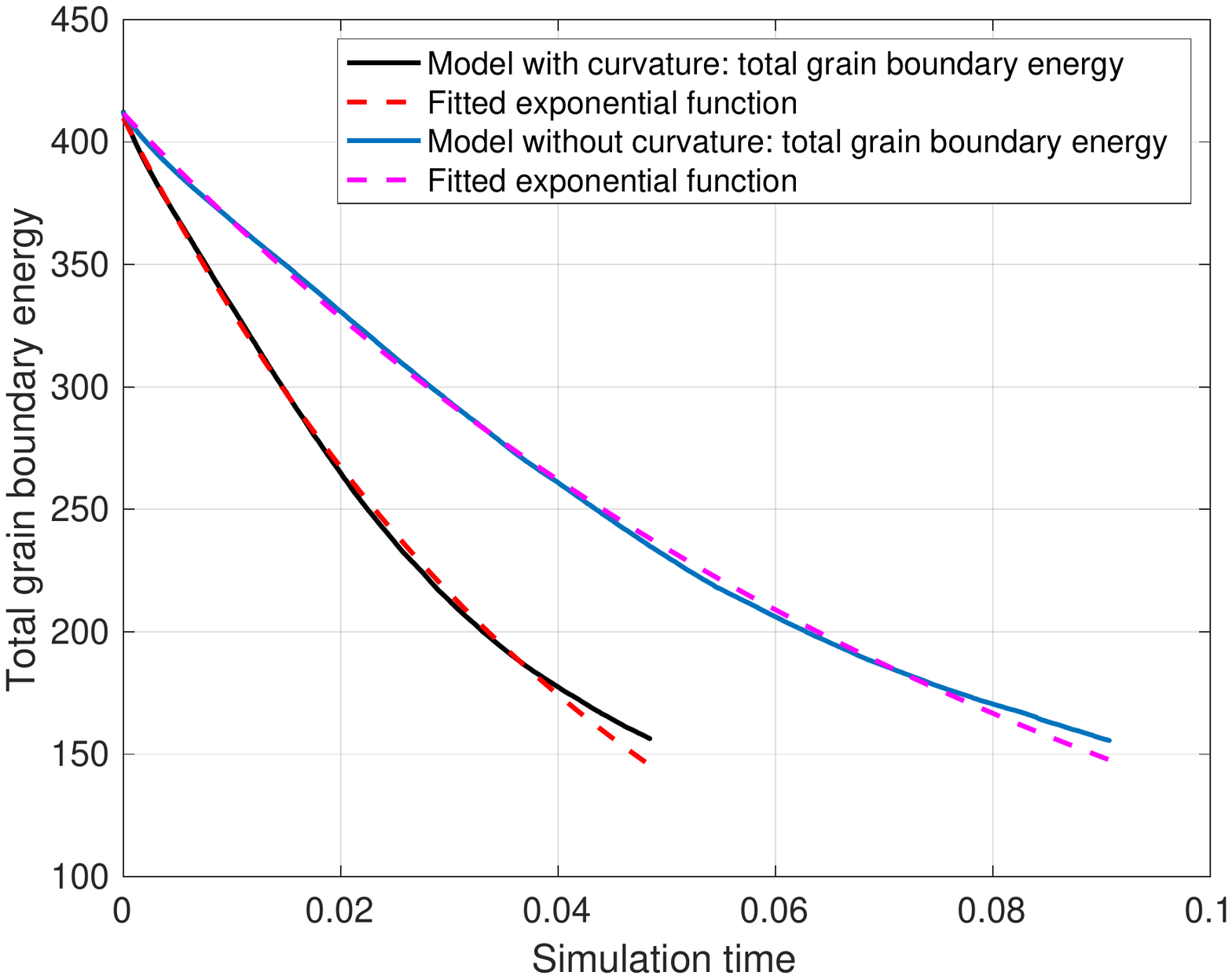}
\includegraphics[width=2.1in]{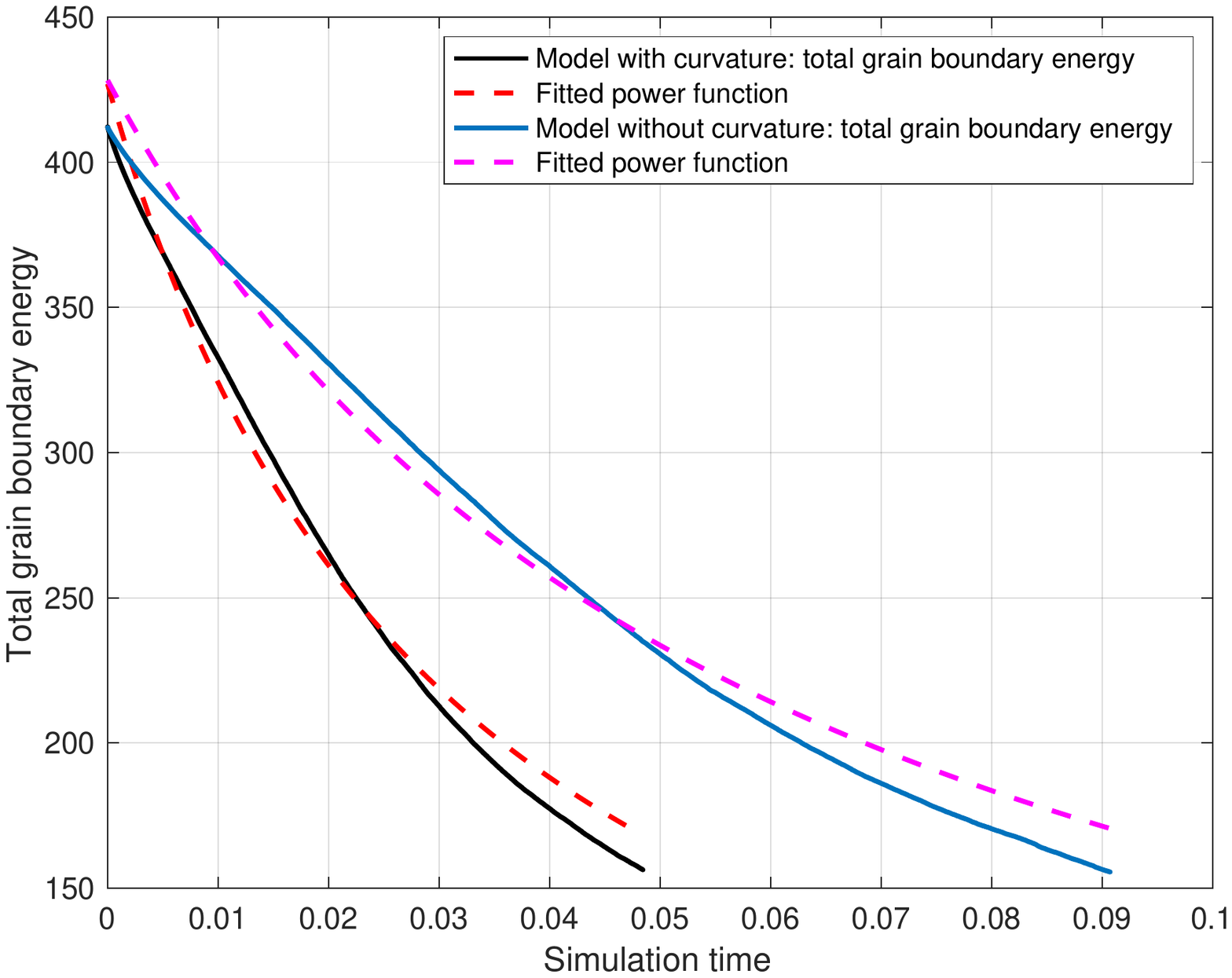}
\vspace{-1.8cm}
\caption{\footnotesize One run of $2$D trial with $10000$ initial
  grains: {\it (a) Left plot,} Total grain boundary energy plot, model
  with curvature  (solid black) versus fitted  exponential decaying function
  $y(t)=409.8\exp(-21.38t)$ (dashed red). Total grain boundary energy plot, model
  without curvature  (solid blue) versus fitted  exponential decaying function
  $y(t)=411.6\exp(-11.3t)$ (dashed magenta); {\it (b) Right
    plot},  Total grain boundary energy plot, model
  with curvature  (solid black) versus fitted  power function
  $y_1(t)=426.9841(1.0+31.746t)^{-1}$ (dashed red). Total grain boundary energy plot, model
  without curvature  (solid blue) versus fitted power function
  $y_1(t)=428.2145(1.0+16.6556t)^{-1}$ (dashed magenta). Mobility of the triple
junctions is
$\eta=10$ and the misorientation parameter $\gamma=1$.  Grain boundary
energy density $\sigma=1+0.25\sin^4(2\Delta \alpha).$}\label{fig10}
\end{figure}

\begin{figure}[hbtp]
\centering
\vspace{-1.8cm}
\includegraphics[width=2.1in]{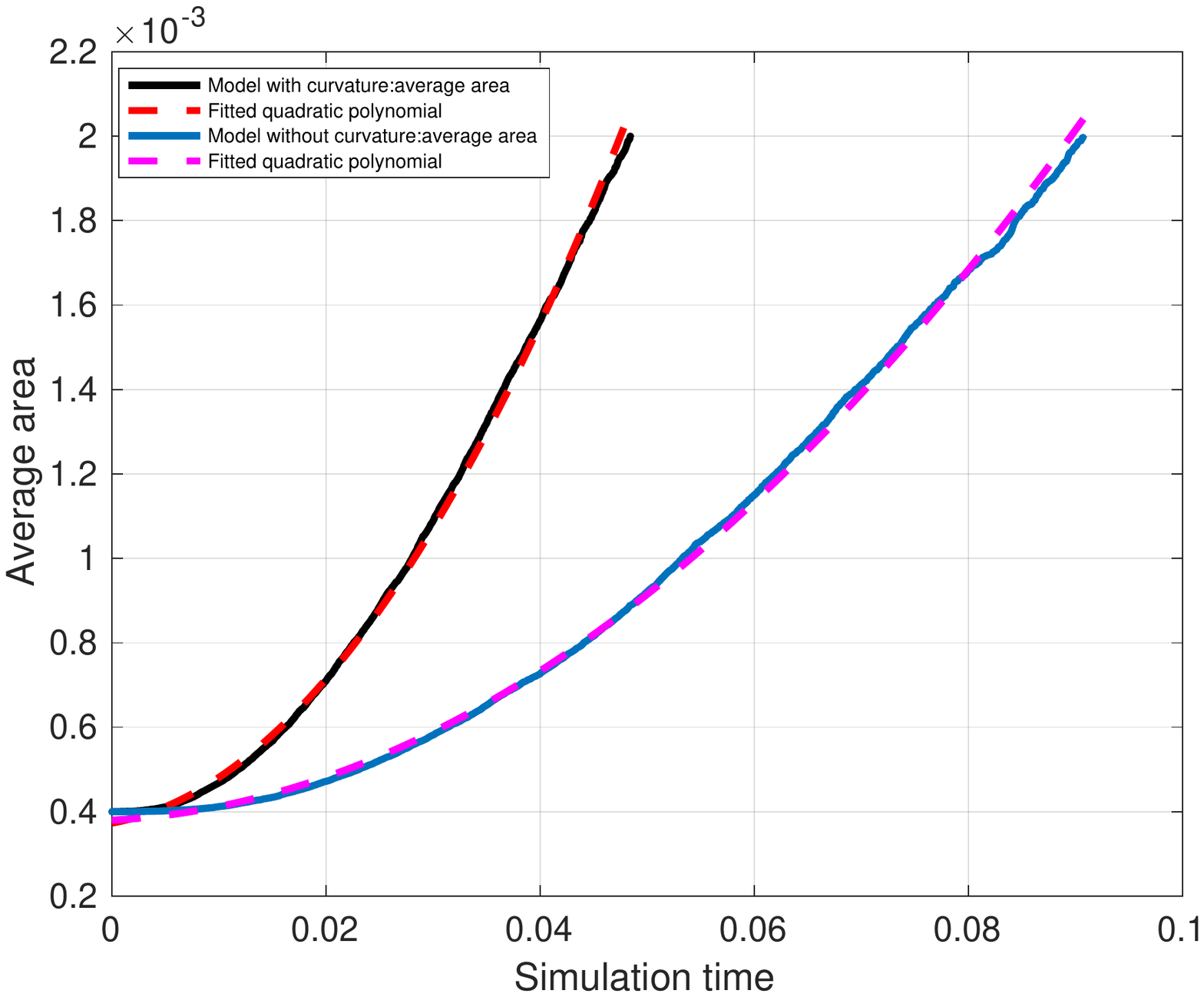}
\includegraphics[width=2.1in]{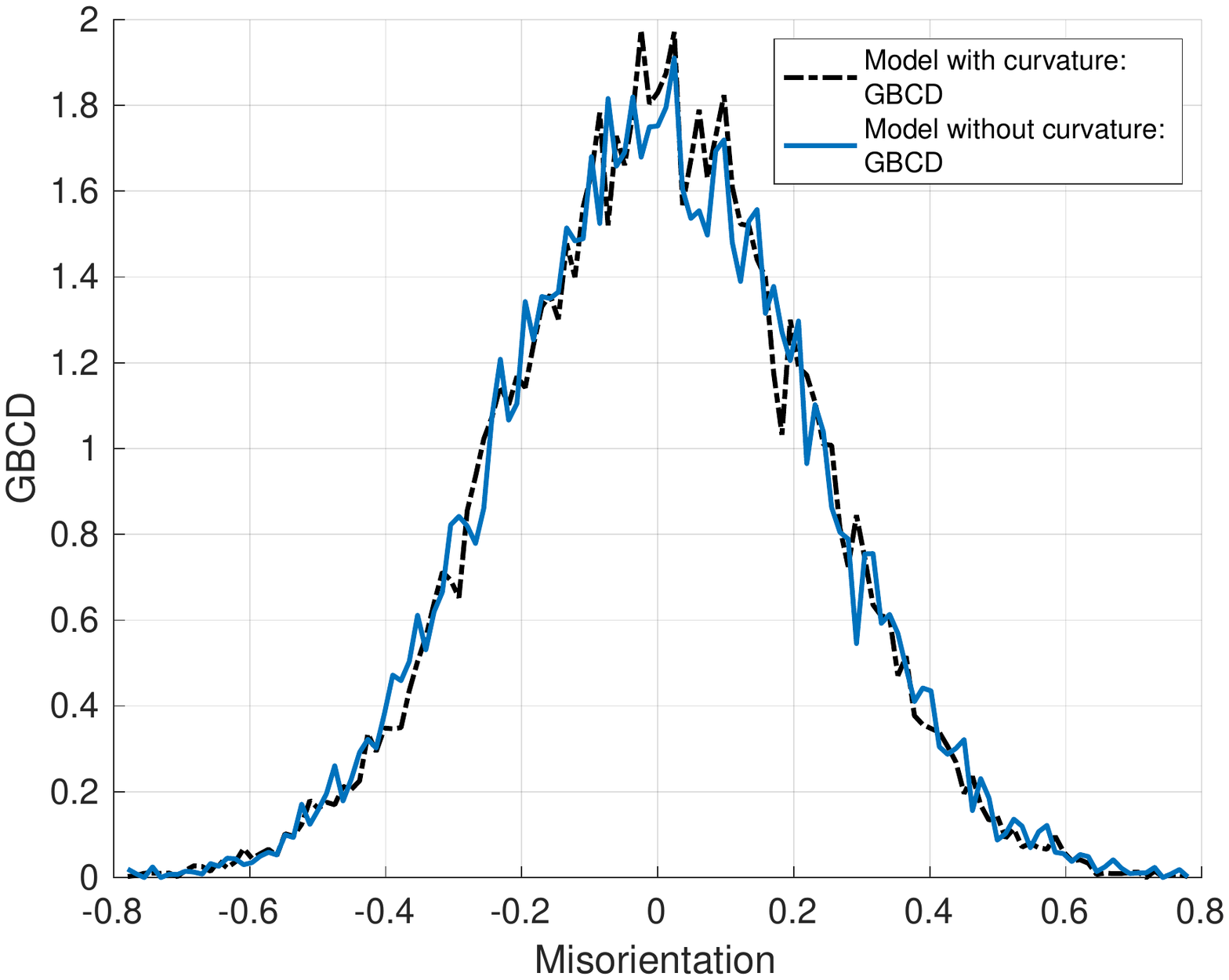}
\vspace{-1.8cm}
\caption{\footnotesize {\it (a) Left plot,} One run of $2$D trial with $10000$ initial
  grains: Growth of the average area of the
  grains, model with curvature (solid black) versus fitted  quadratic
  polynomial function  $y(t)=0.6258t^2+0.004538t+0.0003732$ (dashed
  red). Growth of the average area of the
  grains, model without curvature (solid blue) versus fitted  quadratic
  polynomial function  $y(t)=0.1866t^2+0.001377t+0.0003799$ (dashed magenta);
 {\it (b) Right plot}, GBCD (black curve,  model with curvature) and  GBCD (blue curve,  model without curvature) at $T_{\infty}$ averaged over 3 runs of $2$D trials with $10000$ initial
  grains. Mobility of the triple junctions is $\eta=10$ and the
  misorientation parameter $\gamma=1$.  Grain boundary
energy density $\sigma=1+0.25\sin^4(2\Delta \alpha).$}\label{fig11}
\end{figure}

\begin{figure}[hbtp]
\centering
\vspace{-1.8cm}
\includegraphics[width=2.1in]{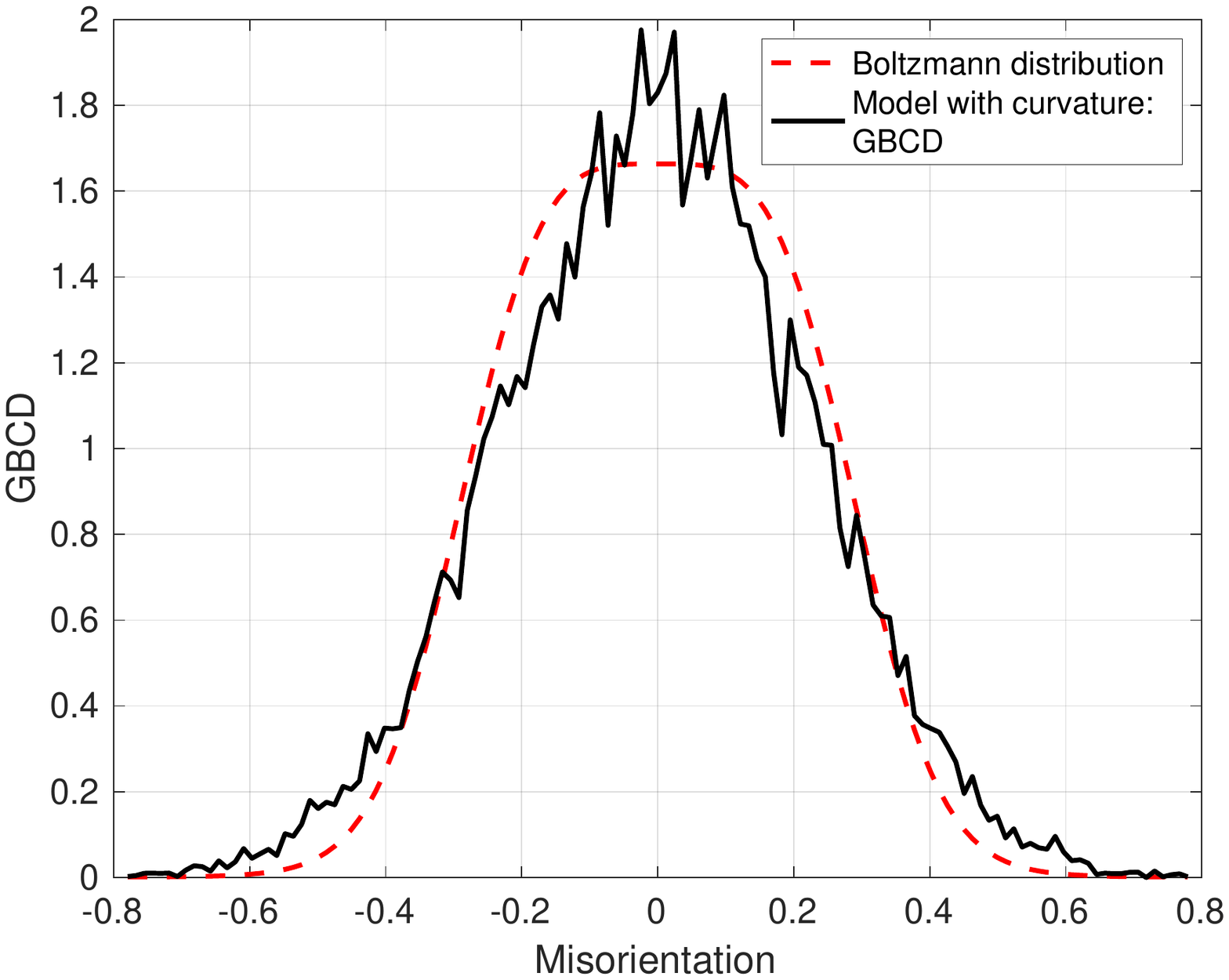}
\includegraphics[width=2.1in]{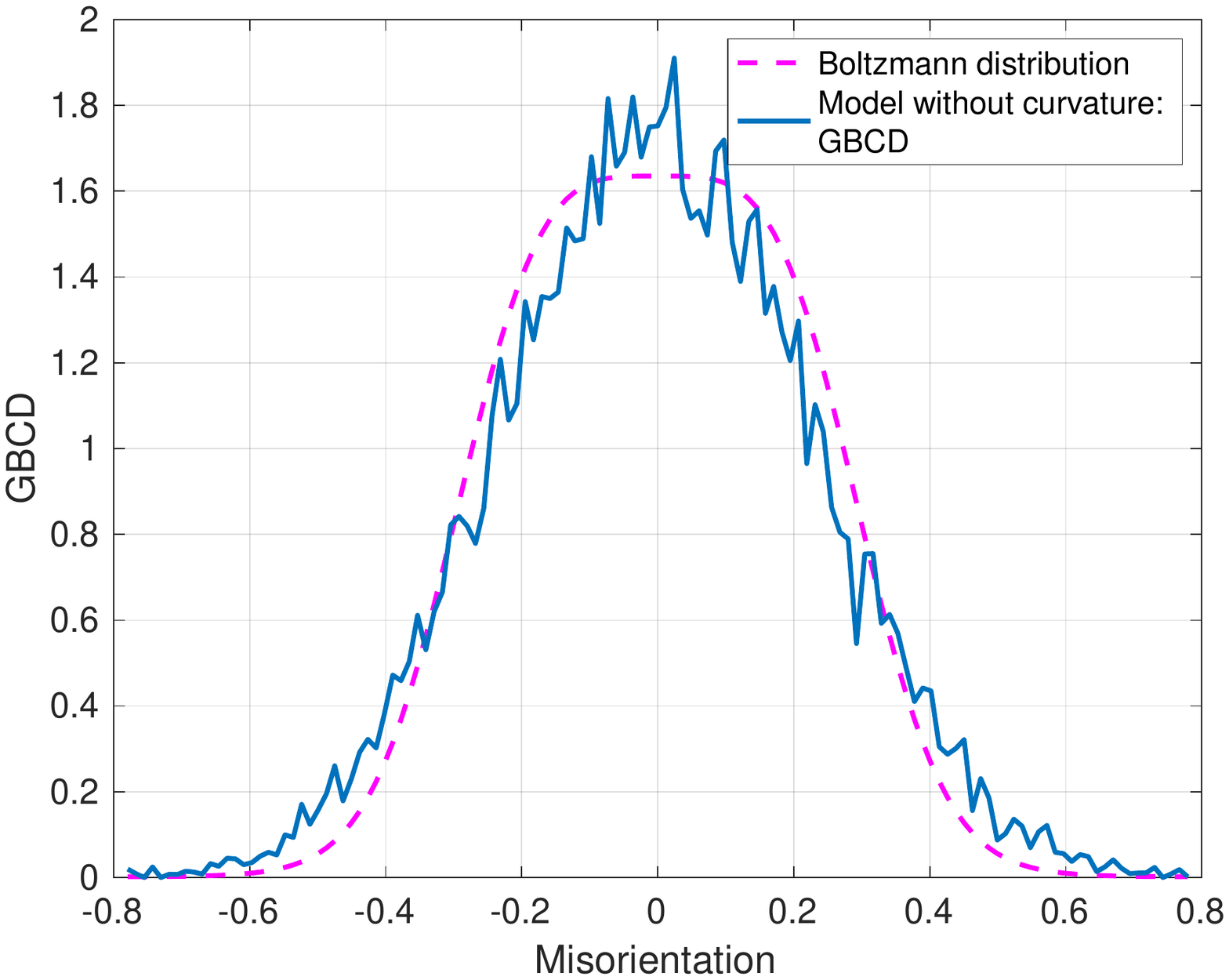}
\vspace{-1.8cm}
\caption{\footnotesize {\it (a) Left plot,} Model with curvature, GBCD (black curve) at $T_{\infty}$ averaged over 3 runs of $2$D trials with $10000$ initial
  grains versus Boltzmann distribution with ``temperature''-
$D\approx 0.035$ (dashed red curve). 
 {\it (b) Right plot}, Model without curvature, GBCD (blue curve) at $T_{\infty}$ averaged over 3 runs of $2$D trials with $10000$ initial
  grains versus Boltzmann distribution with ``temperature''-
$D\approx 0.035$ (dashed magenta curve). Mobility of the triple
junctions is $\eta=10$ and the misorientation parameter $\gamma=1$.  Grain boundary
energy density $\sigma=1+0.25\sin^4(2\Delta \alpha).$}\label{fig12}
\end{figure}

\begin{figure}[hbtp]
\centering
\vspace{-1.8cm}
\includegraphics[width=2.1in]{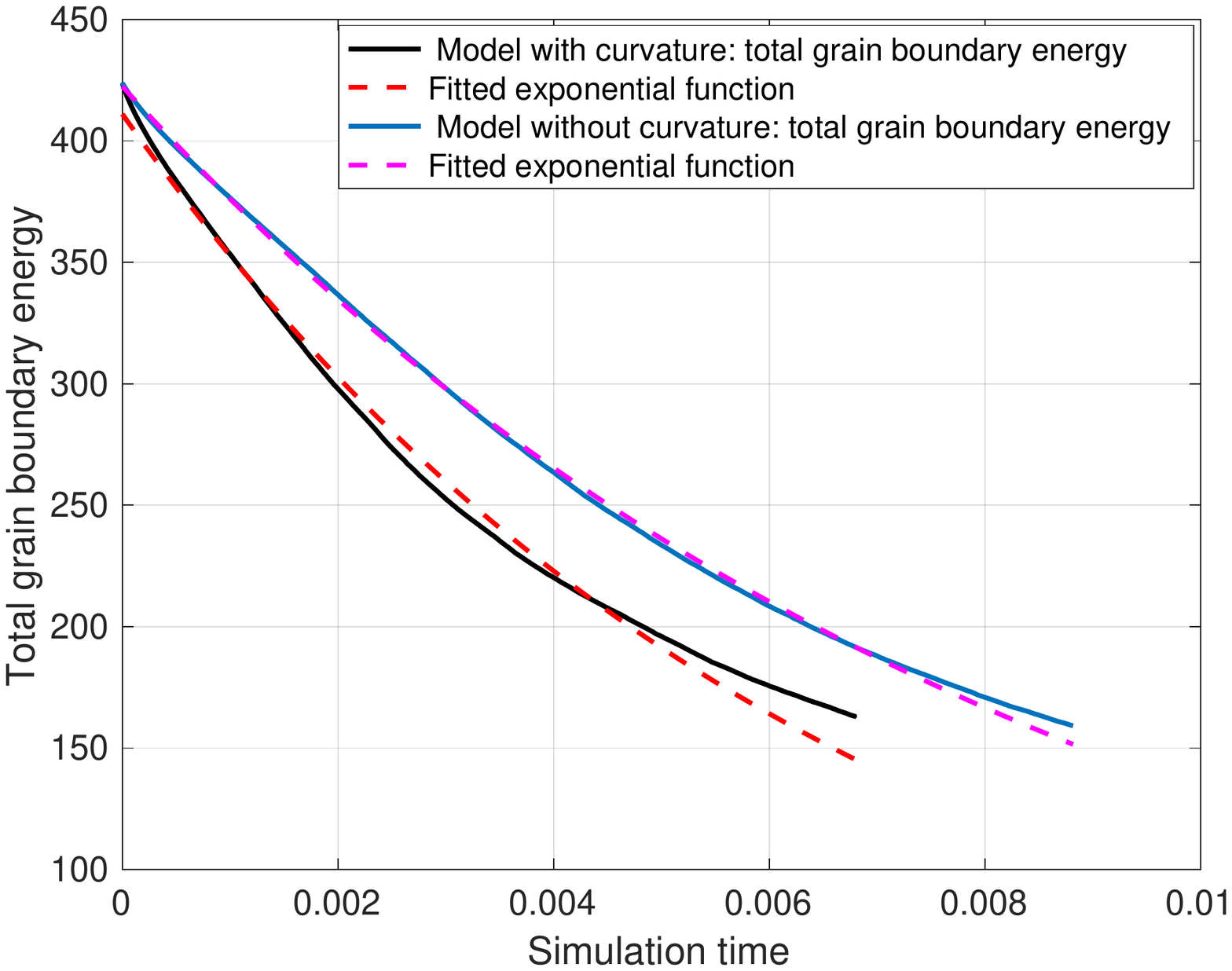}
\includegraphics[width=2.1in]{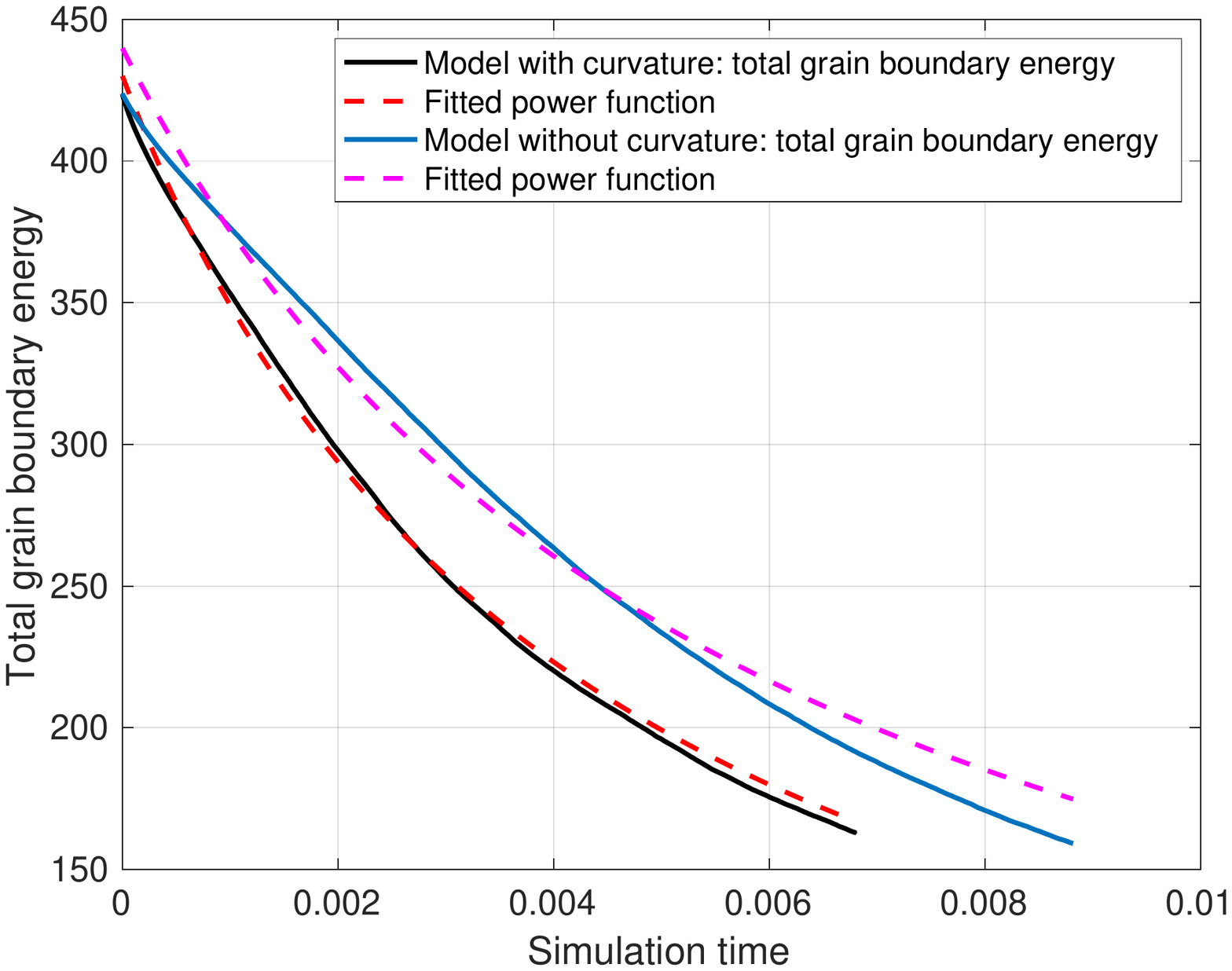}
\vspace{-1.8cm}
\caption{\footnotesize One run of $2$D trial with $10000$ initial
  grains: {\it (a) Left plot,} Total grain boundary energy plot, model
  with curvature  (solid black) versus fitted  exponential decaying function
  $y(t)=411\exp(-153t)$ (dashed red). Total grain boundary energy plot, model
  without curvature  (solid blue) versus fitted  exponential decaying function
  $y(t)=422.4\exp(-116.3t)$ (dashed magenta); {\it (b) Right
    plot},  Total grain boundary energy plot, model
  with curvature  (solid black) versus fitted  power function
  $y_1(t)=430.0278(1.0+231.6960t)^{-1}$ (dashed red). Total grain boundary energy plot, model
  without curvature  (solid blue) versus fitted power function
  $y_1(t)=439.8212(1.0+171.9395t)^{-1}$ (dashed magenta).
 Mobility of the triple
junctions is
$\eta=100$ and the misorientation parameter $\gamma=1$. Grain boundary
energy density $\sigma=1+0.25\sin^2(2\Delta \alpha).$}\label{fig1}
\end{figure}
In
particular, the main goal of our numerical experiments is to
illustrate the time scales effect of curvature - through grain
boundary mobility $\mu$, mobility of the triple junctions $\eta$,
and misorientation parameter $\gamma$ on how the grain boundary
system decays energy and coarsens with time. For that we will numerically study
evolution of the total grain boundary energy,
\begin{equation}\label{eq:7.1}
 E(t)
  =
  \sum_{j}
  \sigma
  (
  \Delta^{(j)}\alpha
  )
  |\Gamma_t^{(j)}|,
\end{equation}
where as before, $\Delta^{(j)}\alpha$ is a misorientation of the grain
boundary $\Gamma^{(j)}_t$, and  $|\Gamma^{(j)}_t|$ is the length of the
grain boundary. We will also consider the growth of the average area, defined as,
\begin{equation}\label{eq:7.2}
A(t)=\frac{4}{N(t)},
\end{equation}
here $4$ is the total area of the sample, and $N(t)$ is the total
number of grains at time $t$. The growth of the average area is
closely related to the coarsening rate of the grain system that
undergoes critical/disappearance events. However, it is important to
note that critical events include not only grain disappearance, but
also facet/grain boundary disappearance, facet interchange,
and splitting of unstable junctions, for more details about
  numerical modeling of critical events in 2D, see e.g. \cite{MR2272185,MR2772123}. Further,
we will investigate the distribution of the grain boundary character
distribution (GBCD) $\rho(\Delta ^{(j)} \alpha)$ at a final time of
the simulations $T_{\infty}$ (defined below) under a simplified
assumption on a grain boundary energy density, namely that $\sigma(
\Delta^{(j)}\alpha)$ is only a function of the misorientation, see
also Sections \ref{sec:1}-\ref{sec:2}.
The GBCD (in this context) is an empirical
statistical measure of the relative length (in 2D) of the grain
boundary interface with a given lattice misorientation,
\begin{eqnarray}
&\rho(\Delta ^{(j)} \alpha, t) =\mbox{ relative length of interface of
  lattice misorientation } \Delta ^{(j)} \alpha \mbox{ at time }t,\nonumber\\
&\mbox{ normalized so that } \int_{\Omega_{\Delta ^{(j)} \alpha}} \rho
  d \Delta ^{(j)} \alpha=1, \label{eq:7.3}
\end{eqnarray}
where we consider $\Omega_{\Delta ^{(j)} \alpha}=[-\frac{\pi}{4},
\frac{\pi}{4}]$ in the numerical experiments below (for planar grain
boundary network, it is reasonable to consider such range for the
misorientations). For more details, see for example \cite{DK:gbphysrev, Katya-Chun-Mzn2}. In all our tests below,
we compare the GBCD at $T_{\infty}$ to the stationary solution of the
Fokker-Planck equation, the  Boltzmann distribution for the grain
boundary energy density $\sigma
  (
  \Delta^{(j)}\alpha
  )$,
\begin{equation} \begin{aligned} \label{eq:7.4}
& \rho_D(\Delta^{(j)}\alpha) =
\frac{1}{Z_D}e^{-\frac{\sigma(\Delta^{(j)}\alpha)}{D}}, \\
& \textrm{with partition function, i.e.,normalization factor} \\
&Z_D = \int_{\Omega_{\Delta^{(j)}\alpha}}
e^{-\frac{\sigma(\Delta^{(j)}\alpha)}{D}}d\Delta^{(j)} \alpha,
\end{aligned}\end{equation}
\cite{DK:BEEEKT,DK:gbphysrev,MR2772123,MR3729587}. We employ the
Kullback-Leibler relative entropy test to obtain a unique
``temperature-like'' parameter $D$ and to construct the corresponding
Boltzmann distribution for the GBCD at $T_{\infty}$ as it was originally done in \cite{DK:BEEEKT,DK:gbphysrev,MR2772123,MR3729587}.
Note, as we also discussed in Section \ref{sec:3a},  GBCD is a primary candidate to characterize texture of the grain boundary
network, and is inversely related to the grain boundary energy density
as discovered in experiments and simulations.
The
reader can consult, for example, \cite{DK:BEEEKT,DK:gbphysrev,
  MR2772123, MR3729587} for more details about GBCD and the theory of
the GBCD.
 In the numerical experiments in this paper, we consider two choices
 for the
grain boundary energy density as plotted in Figure \ref{gbend} and given below,
\[\sigma(\Delta^{(j)} \alpha)=1+0.25\sin^2(2\Delta^{(j)} \alpha)
  \mbox{ and } \sigma(\Delta^{(j)} \alpha)=1+0.25\sin^4(2\Delta^{(j)} \alpha).\]
\par  We consider simulation of 2D grain
boundary network using the algorithm based on the sharp interface
approach  \cite{Katya-Chun-Mzn2} with dynamic misorientation and finite mobility of the triple junctions which we also extended to a
model without curvature \eqref{eq:6.7}.  Note, that the algorithm
\cite{Katya-Chun-Mzn2} is a further extension of the algorithm from
\cite{MR2772123,MR3729587}. We recall that in the numerical scheme
we work with a variational principle. The cornerstone of the algorithm, which assures its
stability, is the discrete dissipation inequality for the total grain boundary energy that holds when
either the discrete Herring boundary condition ($\eta \to \infty$) or  discrete ``dynamic boundary
condition'' (finite mobility $\eta$ of the triple junctions,  third
equation of (\ref{eq:6.4}) or of (\ref{eq:6.7})) is
satisfied at the triple junctions. We also recall that in the
numerical algorithm for model (\ref{eq:6.4}) we impose
the Mullins' theory (first equation of (\ref{eq:6.4})) as the local
evolution law for the grain boundaries (and the time scale $\mu$ is
kept finite). For model (\ref{eq:6.7}), $\mu\to \infty$, hence the
dynamics of the grain boundaries are defined by the evolution of the
triple junctions (the third equation of (\ref{eq:6.7})) and by the
grains rotation  (the second equation of (\ref{eq:6.7})).
The reader can consult
\cite{MR2772123,MR3729587, Katya-Chun-Mzn2} for more details about
numerical algorithm based on the sharp interface approach.
\par In all the numerical tests below we initialized our system with
$10^4$ cells/grains with normally distributed misorientation angles at
initial time $t=0$. We also assume that the final time of the
simulations $T_{\infty}$ is the time when
approximately $80\%$ of grains disappeared from the system, namely the
time when only about $2000$ cells/grains remain. The final time is
selected based on the system \eqref{eq:6.4} with no dynamic misorientations
($\gamma=0$) and with the Herring condition at the triple junctions ($\eta
\to \infty$)
and, it is selected to ensure that statistically significant number of grains still remain
in the system and that
the system reached its statistical stead-state. Therefore, all the
numerical results which are presented below are for the grain boundary
system that undergoes critical/disappearance events.
\par First,  we study the effect of dynamics of triple
junctions on the dissipation and coarsening of the system,
see Figures \ref{fig7}-\ref{fig18} (we consider different values of
misorientation parameter $\gamma$ for these tests). We observe that for smaller
values of the mobility of the triple junctions $\eta$, the energy
decay $E(t)$ is well-approximated by an exponential function for both
models, for the model
with curvature \eqref{eq:6.4} and for the model without curvature \eqref{eq:6.7}, see Figures
\ref{fig7}  and \ref{fig10} (left plots). This is consistent with the results of our
theory, see Sections \ref{sec:1}-\ref{sec:2} and \cite{Katya-Chun-Mzn1, Katya-Chun-Mzn2}, even
though, the theoretical results are obtained under assumption of no
critical events and $\mu\to \infty$ (for grain growth model without
curvature).  This result indicates that for lower mobility of the triple
junctions $\eta$,  the dynamics of triple
junctions have a dominant effect on the grain growth, see model \eqref{eq:6.7}. This explains the similarity in the energy decay for
grain growth model with curvature  \eqref{eq:6.4} and  without curvature \eqref{eq:6.7} when $\eta=10$, Figures
\ref{fig7}  and \ref{fig10}.
In comparison, we also present fit to a power law decaying
function, see Figures \ref{fig7} and \ref{fig10} (right plots). The power law function
does not seem to give as good approximation in this case. 
\begin{figure}[hbtp]
\centering
\vspace{-1.8cm}
\includegraphics[width=2.1in]{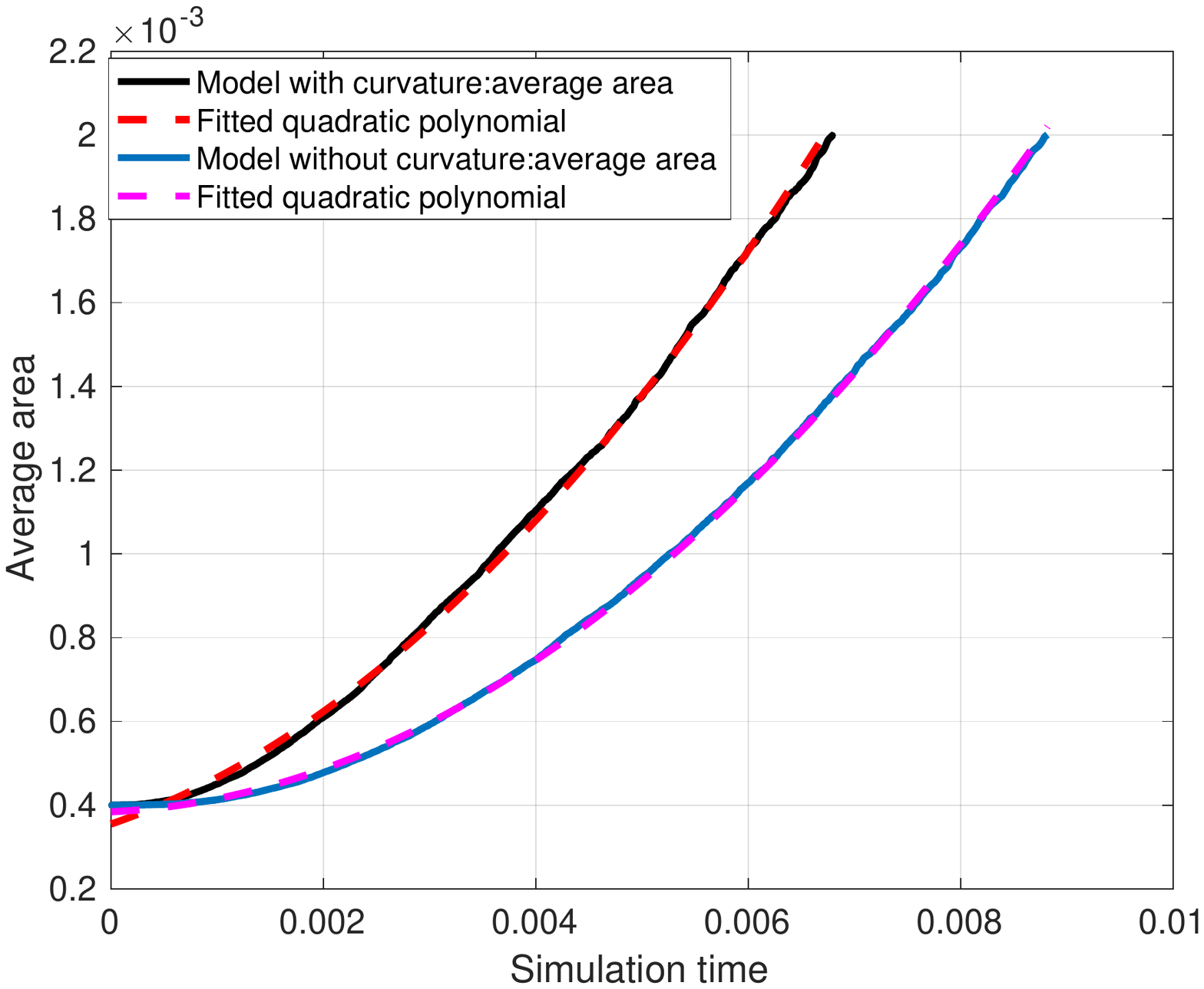}
\includegraphics[width=2.1in]{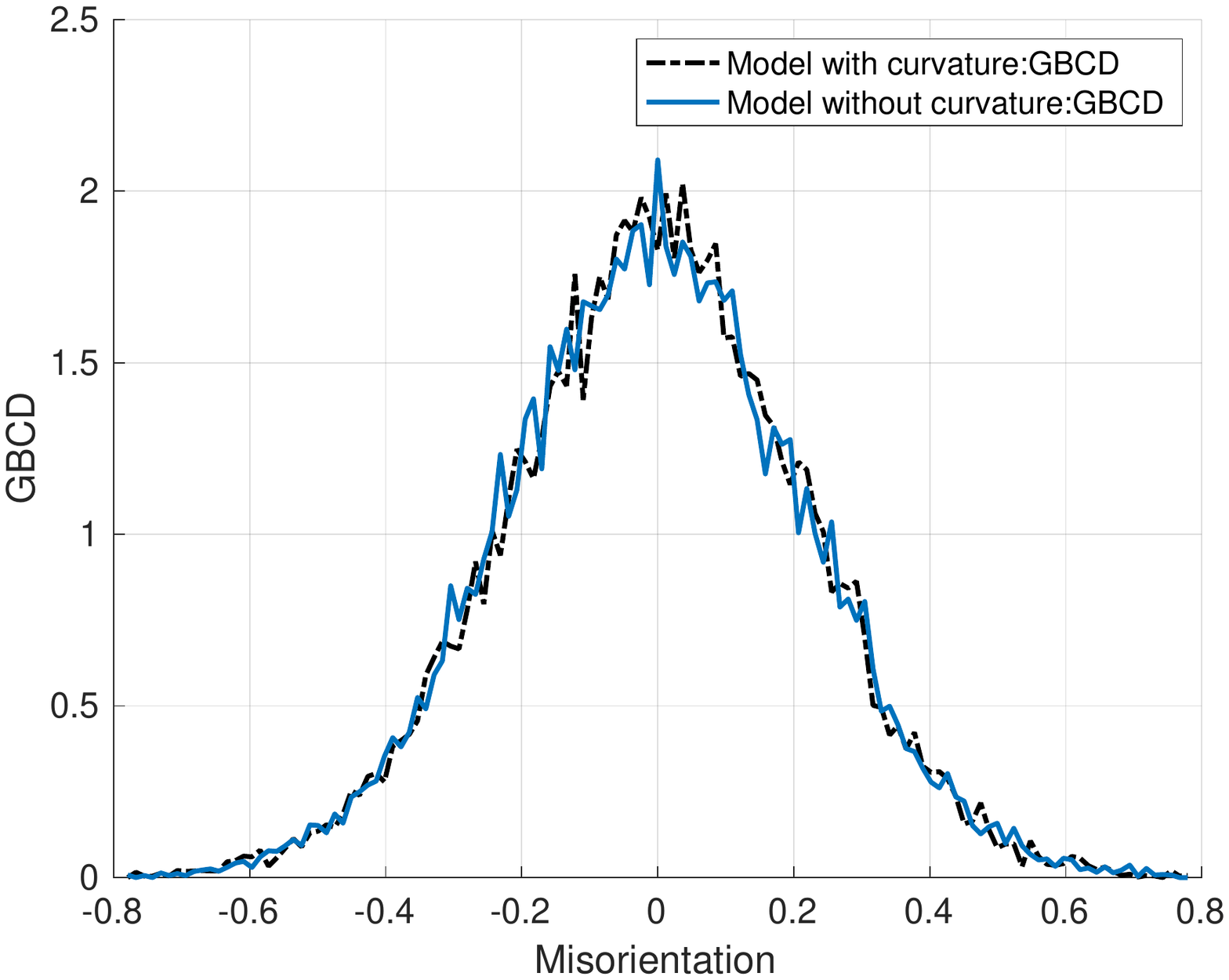}
\vspace{-1.8cm}
\caption{\footnotesize {\it (a) Left plot,} One run of $2$D trial with $10000$ initial
  grains: Growth of the average area of the
  grains, model with curvature (solid black) versus fitted  quadratic
  polynomial function  $y(t)=23.47t^2+0.08748t+0.0003549$ (dashed
  red). Growth of the average area of the
  grains, model without curvature (solid blue) versus fitted  quadratic
  polynomial function  $y(t)=19.74t^2+0.01157t+0.0003843$ (dashed magenta);
 {\it (b) Right plot}, GBCD (black curve,  model with curvature) and
  GBCD (blue curve,  model without curvature) at $T_{\infty}$ averaged over 3 runs of $2$D trials with $10000$ initial
  grains. Mobility of the triple junctions is $\eta=100$ and the
  misorientation parameter $\gamma=1$.  Grain boundary
energy density $\sigma=1+0.25\sin^2(2\Delta \alpha).$}\label{fig2}
\end{figure}

\begin{figure}[hbtp]
\centering
\vspace{-1.8cm}
\includegraphics[width=2.1in]{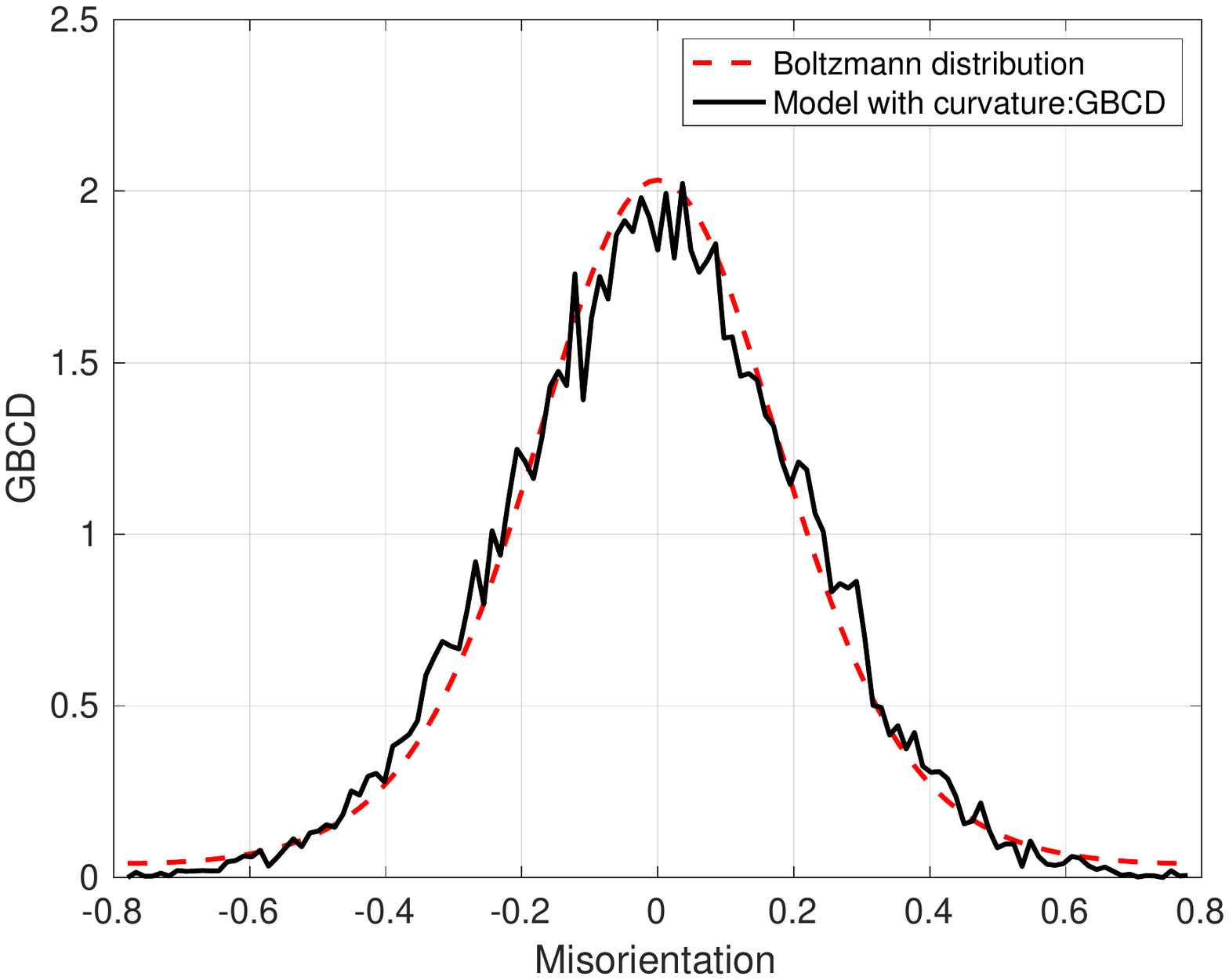}
\includegraphics[width=2.1in]{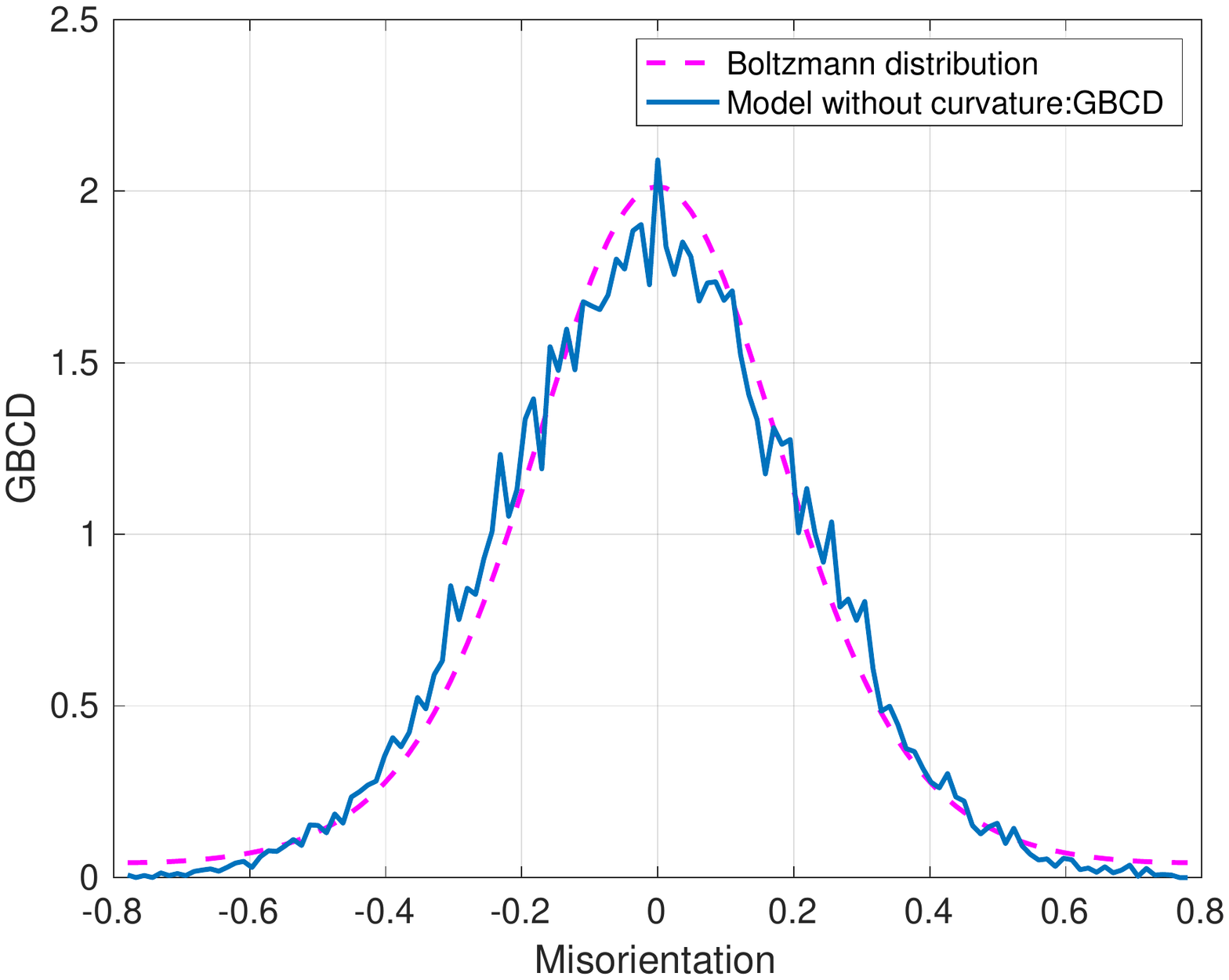}
\vspace{-1.8cm}
\caption{\footnotesize {\it (a) Left plot,} Model with curvature, GBCD (black curve) at $T_{\infty}$ averaged over 3 runs of $2$D trials with $10000$ initial
  grains versus Boltzmann distribution with ``temperature''-
$D\approx 0.0641$ (dashed red curve). 
 {\it (b) Right plot}, Model without curvature, GBCD (blue curve) at $T_{\infty}$ averaged over 3 runs of $2$D trials with $10000$ initial
  grains versus Boltzmann distribution with ``temperature''-
$D\approx 0.0651$ (dashed magenta curve). Mobility of the triple
junctions is $\eta=100$ and the misorientation parameter
$\gamma=1$.  Grain boundary
energy density $\sigma=1+0.25\sin^2(2\Delta \alpha).$}\label{fig3}
\end{figure}
However, for a
larger value of $\eta=100$, Figures \ref{fig1}, \ref{fig4},
\ref{fig13} and \ref{fig16},  we obtain that the total grain boundary energy does not
follow exponential decay anymore for the model with curvature
\eqref{eq:6.4}, but rather the energy decay is closer to a power
law. Thus, the curvature time scale-the grain boundary evolution has a dominant effect for large $\eta$. However, for the
model without curvature \eqref{eq:6.7}, the energy decay is still well
approximated by the exponential function which is consistent with the
theory, Sections \ref{sec:1}-\ref{sec:2}. Note
also, that the numerically observed energy decay rates increase with the
mobility $\eta$  of the triple junctions which is also consistent with
the developed theory \cite{Katya-Chun-Mzn2}. In
addition, we observe that the average area grows as a quadratic function in
time for the finite mobility $\eta$ of the triple junctions, Figures
 \ref{fig8}, \ref{fig11}, \ref{fig2}, \ref{fig5}, \ref{fig14} and
\ref{fig17} (left plots) and see also our earlier work
\cite{Katya-Chun-Mzn2}.  We also observe that the coarsening rate of
grain growth slows down with the
smaller $\eta$.  In addition, we note that the energy
 decay in our numerical tests is consistent with the growth of the average area.
Moreover, we observe that neither dynamics of the triple junctions nor curvature
show as much of an effect on the GBCD,
see Figures  \ref{fig8} (right plot)-\ref{fig9}, \ref{fig11}
(right plot)-\ref{fig12},
\ref{fig2} (right plot)-\ref{fig3}, \ref{fig5} (right
plot)-\ref{fig6}, \ref{fig14} (right plot)-\ref{fig15} and \ref{fig17} (right plot)-\ref{fig18} (note,  the ``temperature'' like
parameter $D$  also accounts for various critical events--grains
disappearance, facet/grain boundary disappearance, facet interchange,
splitting of unstable junctions. It will be part of our
  future study to understand how $D$ depends on the critical events.)
\par For the other series of tests, we vary the misorientation
parameter $\gamma$, second equation of (\ref{eq:6.4}) or of
(\ref{eq:6.7}) (and we set the mobility of the triple junctions
$\eta=100$,  third equation of (\ref{eq:6.4}) or of
(\ref{eq:6.7})). We do not
observe as much of an effect on the energy decay or average area growth in this
case, but we observe the significant effect on the GBCD
and the diffusion coefficient/``temperature''-like parameter $D$,
see Figures \ref{fig7}-\ref{fig6} (with the  misorientation
parameter $\gamma=1$) and  Figures \ref{fig13}-\ref{fig18} (with
larger values of the  misorientation
parameter $\gamma$).  As concluded from our numerical
results, larger values of $\gamma$ give smaller diffusion
coefficient/''temperature''-like parameter $D$, and hence higher GBCD
peak near misorientation $0$. This is consistent with our theory that
basically, larger misorientation parameter $\gamma$ produces direct motion of misorientations
towards equilibrium state of zero misorientations, see Section
\ref{sec:1} and also \cite{Katya-Chun-Mzn2}. Furthermore, from all of
our numerical experiments with dynamic
misorientations and with different triple junction mobilities, we
observe that the GBCD at time $T_{\infty}$ is well-approximated by the Boltzmann
distribution for the grain boundary energy density see Figures  \ref{fig9}, \ref{fig12},
\ref{fig3}, \ref{fig6}, \ref{fig15} and
\ref{fig18}, as well as consistent with experimental findings as
discussed in Section \ref{sec:3a}, which is similar to the
work in \cite{DK:BEEEKT,DK:gbphysrev,
  MR2772123, MR3729587}, but more detailed analysis
needs to be done for a system that undergoes critical events to
understand the relation between GBCD, ``temperature''-like/diffusion
parameter $D$, and different
relaxation time scales, as well as the effect of the time scales
on the dissipation mechanism and certain coarsening rates.
\begin{remark}
Note that, we performed $3$ runs for each numerical test
presented in this work. We report results of a single run for the energy decay
and growth of the average area (the results from the other two
runs for each test were very similar to the presented ones),  and we
illustrate  averaged
over the $3$ runs the GBCD statistics. The curve-fitting
 for the energy and the average area plots was done using Matlab (\cite{Matlab}) toolbox cftool.
\end{remark}

\begin{figure}[hbtp]
\centering
\vspace{-1.8cm}
\includegraphics[width=2.1in]{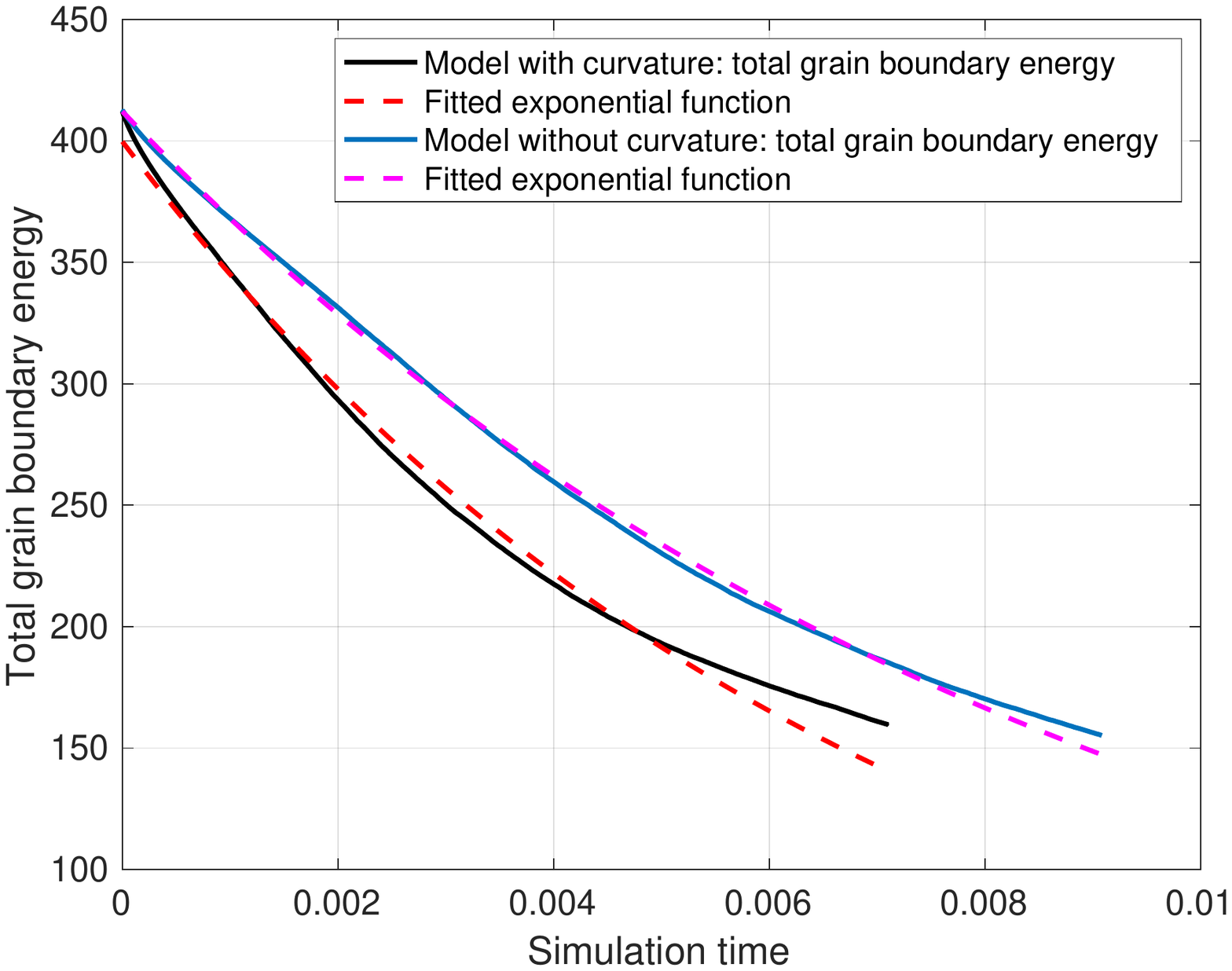}
\includegraphics[width=2.1in]{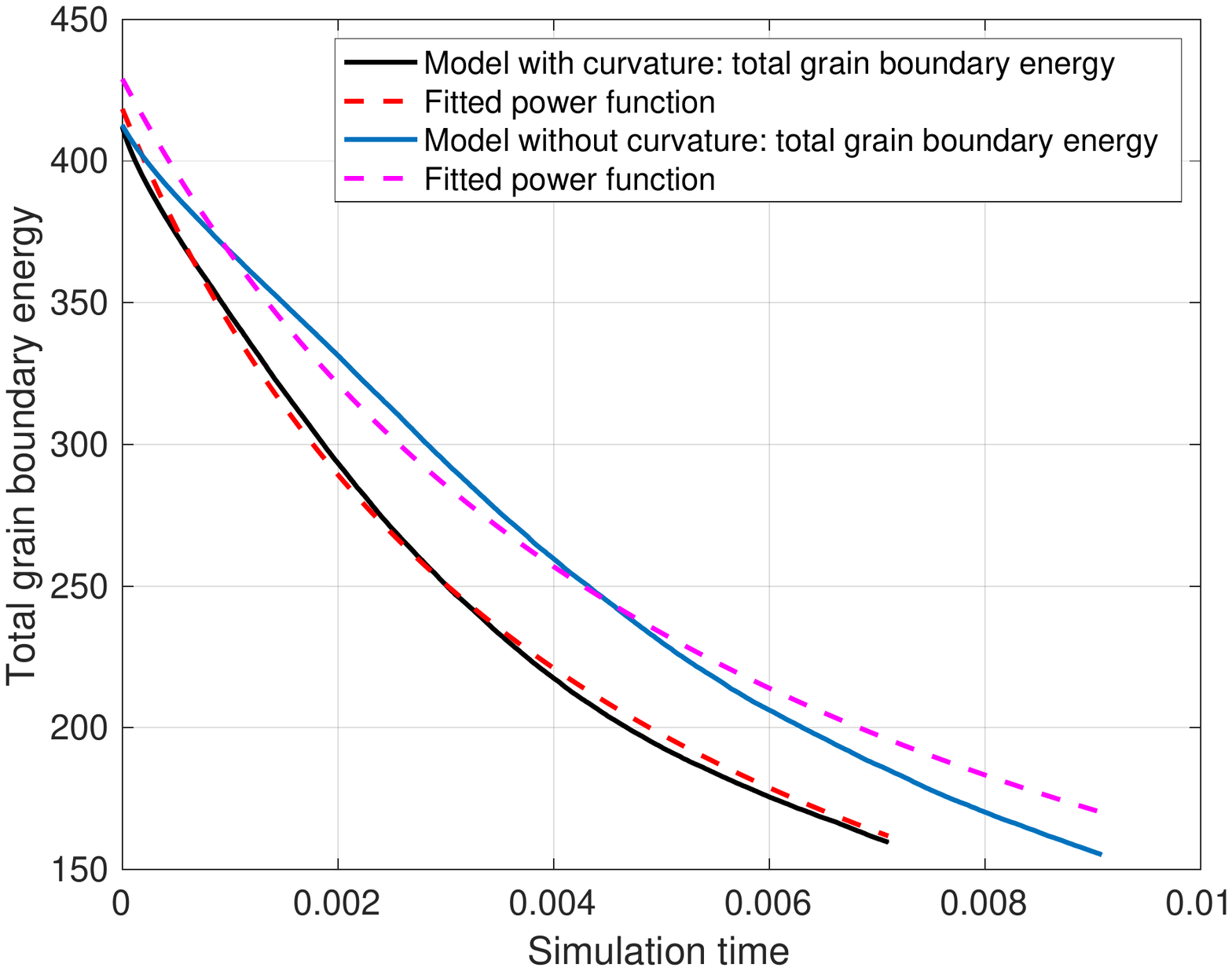}
\vspace{-1.8cm}
\caption{\footnotesize One run of $2$D trial with $10000$ initial
  grains: {\it (a) Left plot,} Total grain boundary energy plot, model
  with curvature  (solid black) versus fitted  exponential decaying function
  $y(t)=399.7\exp(-147.2t)$ (dashed red). Total grain boundary energy plot, model
  without curvature  (solid blue) versus fitted  exponential decaying function
  $y(t)=412.1\exp(-113.3t)$ (dashed magenta); {\it (b) Right
    plot},  Total grain boundary energy plot, model
  with curvature  (solid black) versus fitted  power function
  $y_1(t)=418.3970(1.0+223.2641t)^{-1}$ (dashed red). Total grain boundary energy plot, model
  without curvature  (solid blue) versus fitted power function
  $y_1(t)=428.9782(1.0+167.5042t)^{-1}$ (dashed magenta).
 Mobility of the triple
junctions is
$\eta=100$ and the misorientation parameter $\gamma=1$.  Grain boundary
energy density $\sigma=1+0.25\sin^4(2\Delta \alpha).$}\label{fig4}
\end{figure}

\begin{figure}[hbtp]
\centering
\vspace{-1.8cm}
\includegraphics[width=2.1in]{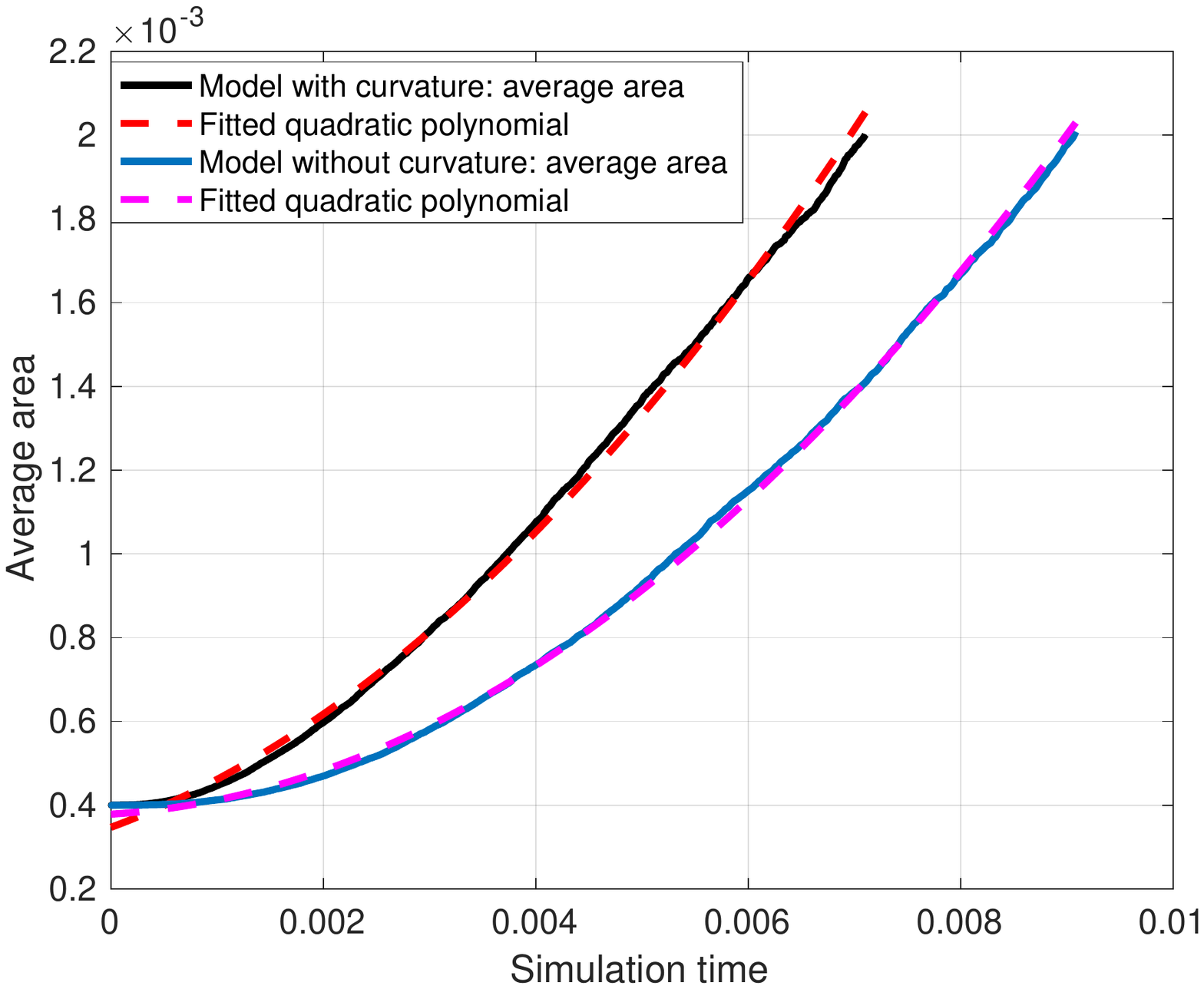}
\includegraphics[width=2.1in]{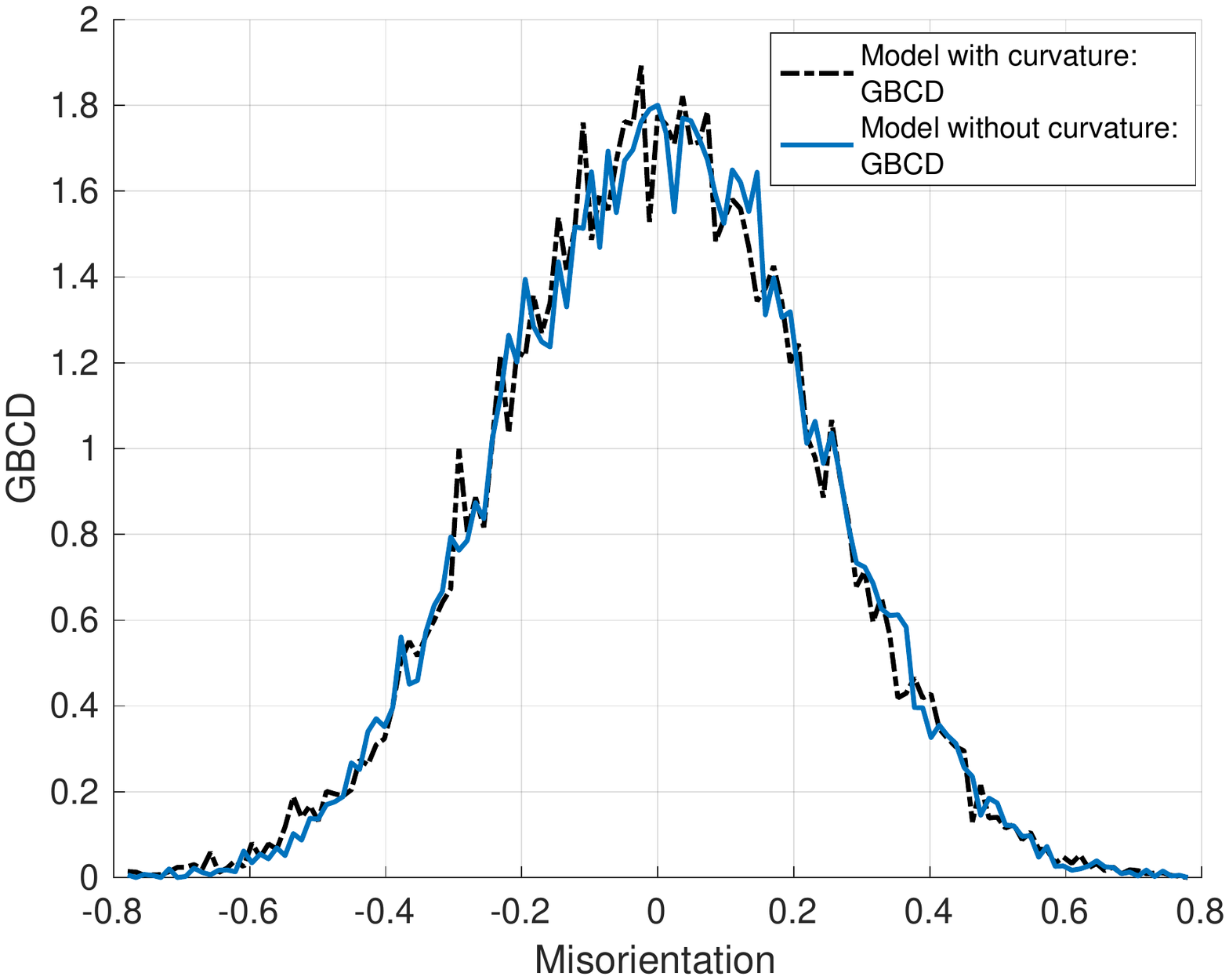}
\vspace{-1.8cm}
\caption{\footnotesize {\it (a) Left plot,} One run of $2$D trial with $10000$ initial
  grains: Growth of the average area of the
  grains, model with curvature (solid black) versus fitted  quadratic
  polynomial function  $y(t)=20.63t^2+0.09393t+0.0003472$ (dashed
  red). Growth of the average area of the
  grains, model without curvature (solid blue) versus fitted  quadratic
  polynomial function  $y(t)=18.31t^2+0.01553t+0.0003786$ (dashed magenta);
 {\it (b) Right plot},  GBCD (black curve,  model with curvature) and GBCD (blue curve,  model without curvature) at $T_{\infty}$ averaged over 3 runs of $2$D trials with $10000$ initial
  grains. Mobility of the triple junctions is $\eta=100$ and the
  misorientation parameter $\gamma=1$.  Grain boundary
energy density $\sigma=1+0.25\sin^4(2\Delta \alpha).$}\label{fig5}
\end{figure}

\begin{figure}[hbtp]
\centering
\vspace{-1.8cm}
\includegraphics[width=2.1in]{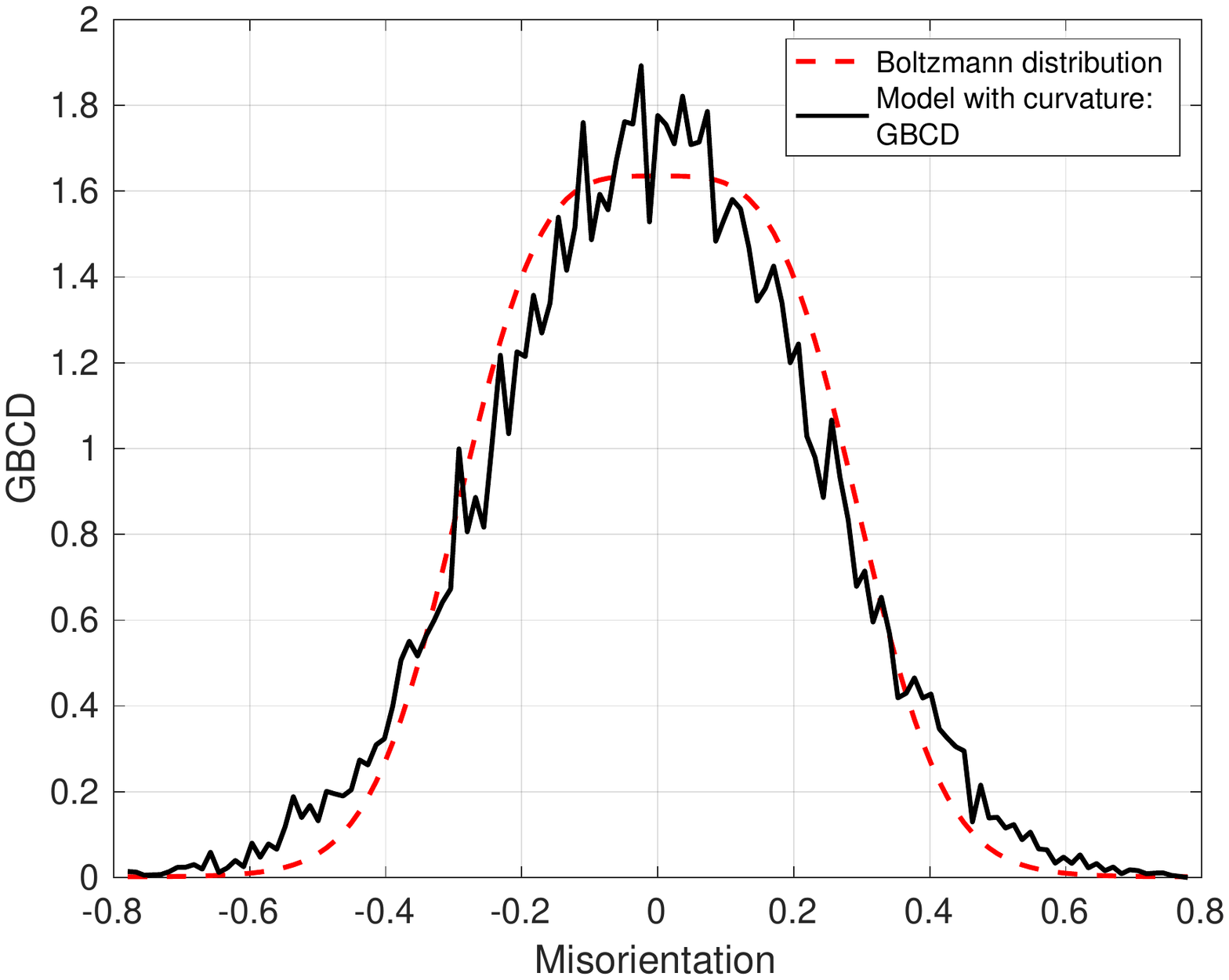}
\includegraphics[width=2.1in]{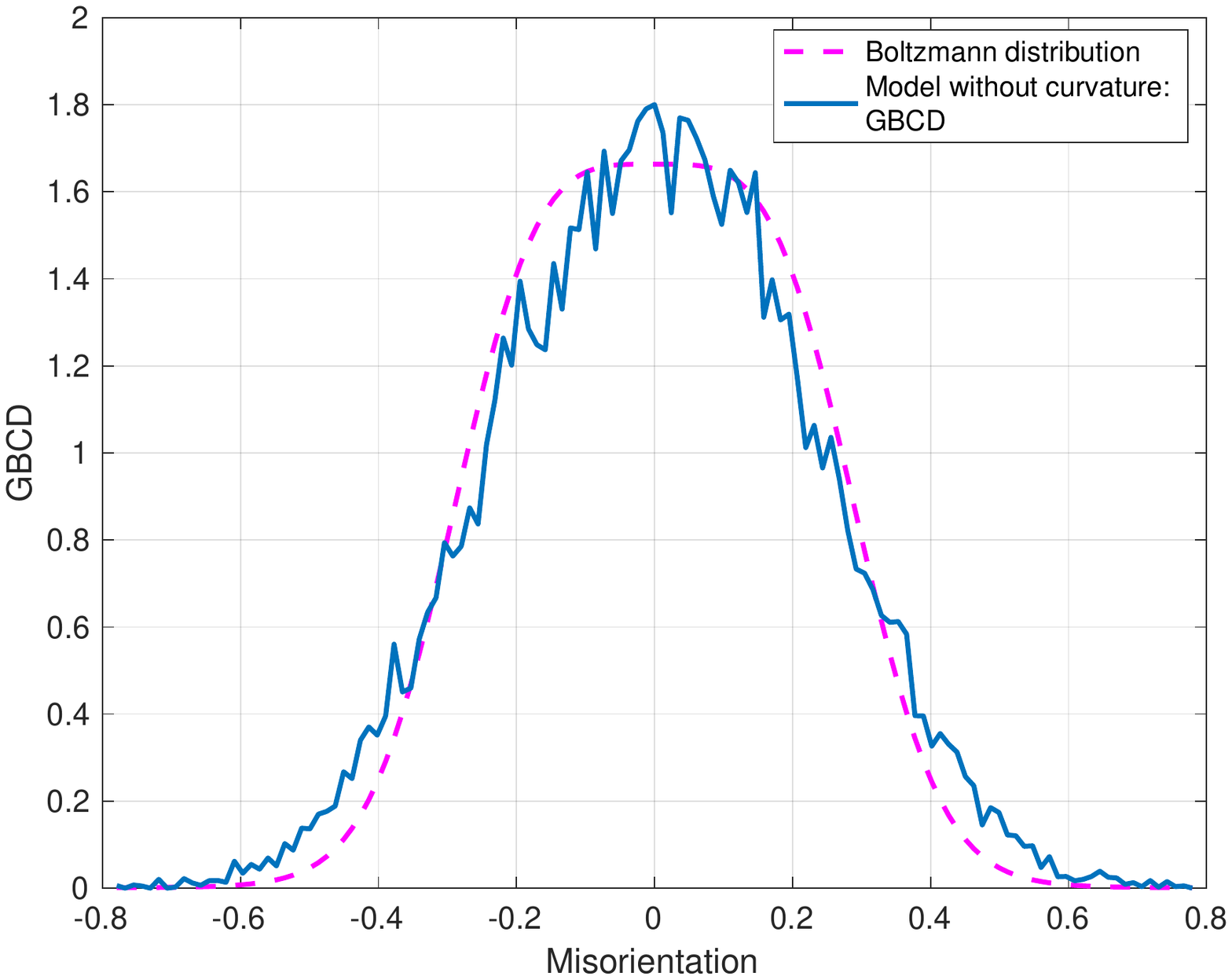}
\vspace{-1.8cm}
\caption{\footnotesize {\it (a) Left plot,} Model with curvature,  GBCD (black curve) at $T_{\infty}$ averaged over 3 runs of $2$D trials with $10000$ initial
  grains versus Boltzmann distribution with ``temperature''-
$D\approx 0.037$ (dashed red curve). 
 {\it (b) Right plot}, Model without curvature, GBCD (blue curve) at $T_{\infty}$ averaged over 3 runs of $2$D trials with $10000$ initial
  grains versus Boltzmann distribution with ``temperature''-
$D\approx 0.035$ (dashed magenta curve). Mobility of the triple
junctions is $\eta=100$ and the misorientation parameter $\gamma=1$.  Grain boundary
energy density $\sigma=1+0.25\sin^4(2\Delta \alpha).$}\label{fig6}
\end{figure}

\begin{figure}[hbtp]
\centering
\vspace{-1.8cm}
\includegraphics[width=2.1in]{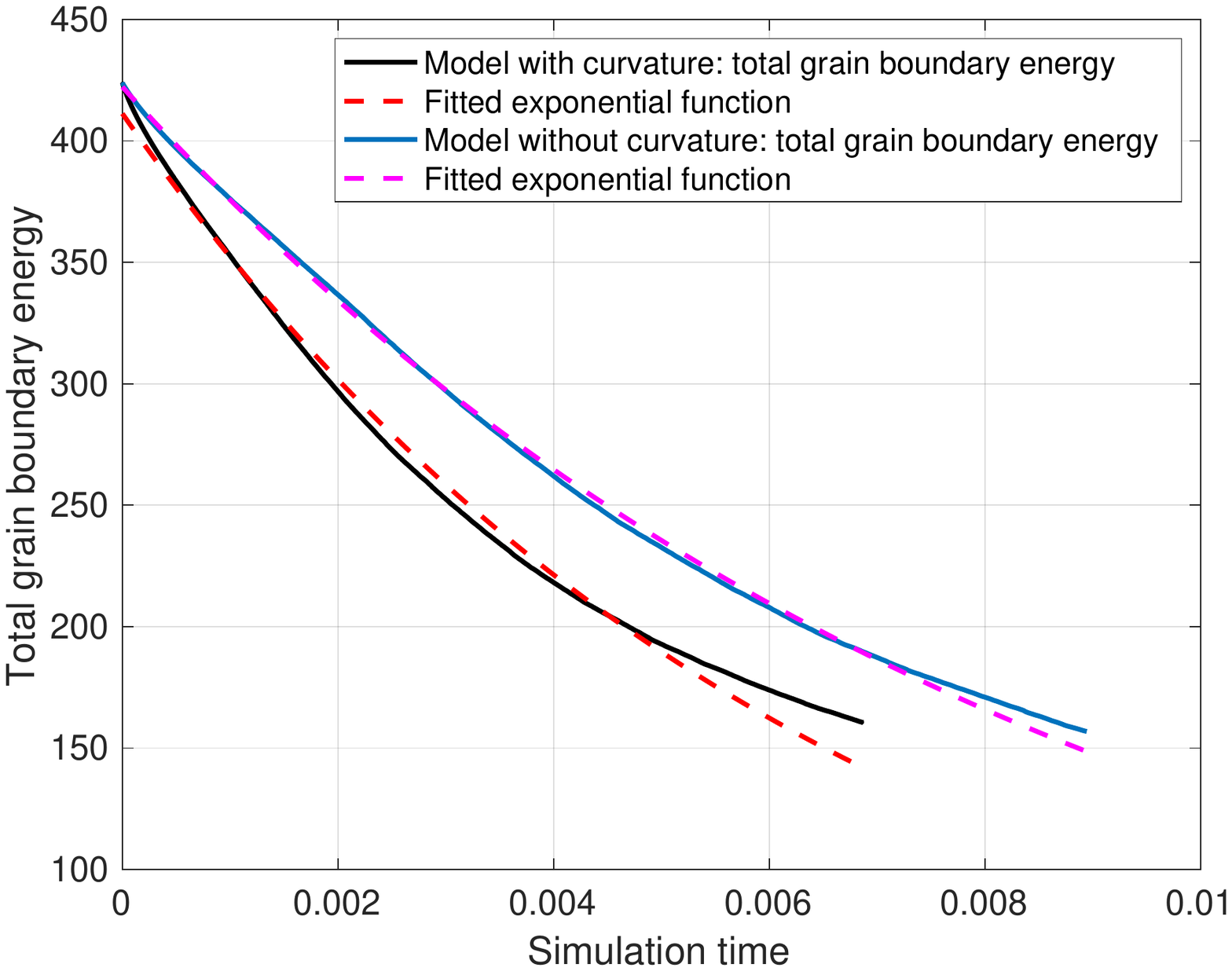}
\includegraphics[width=2.1in]{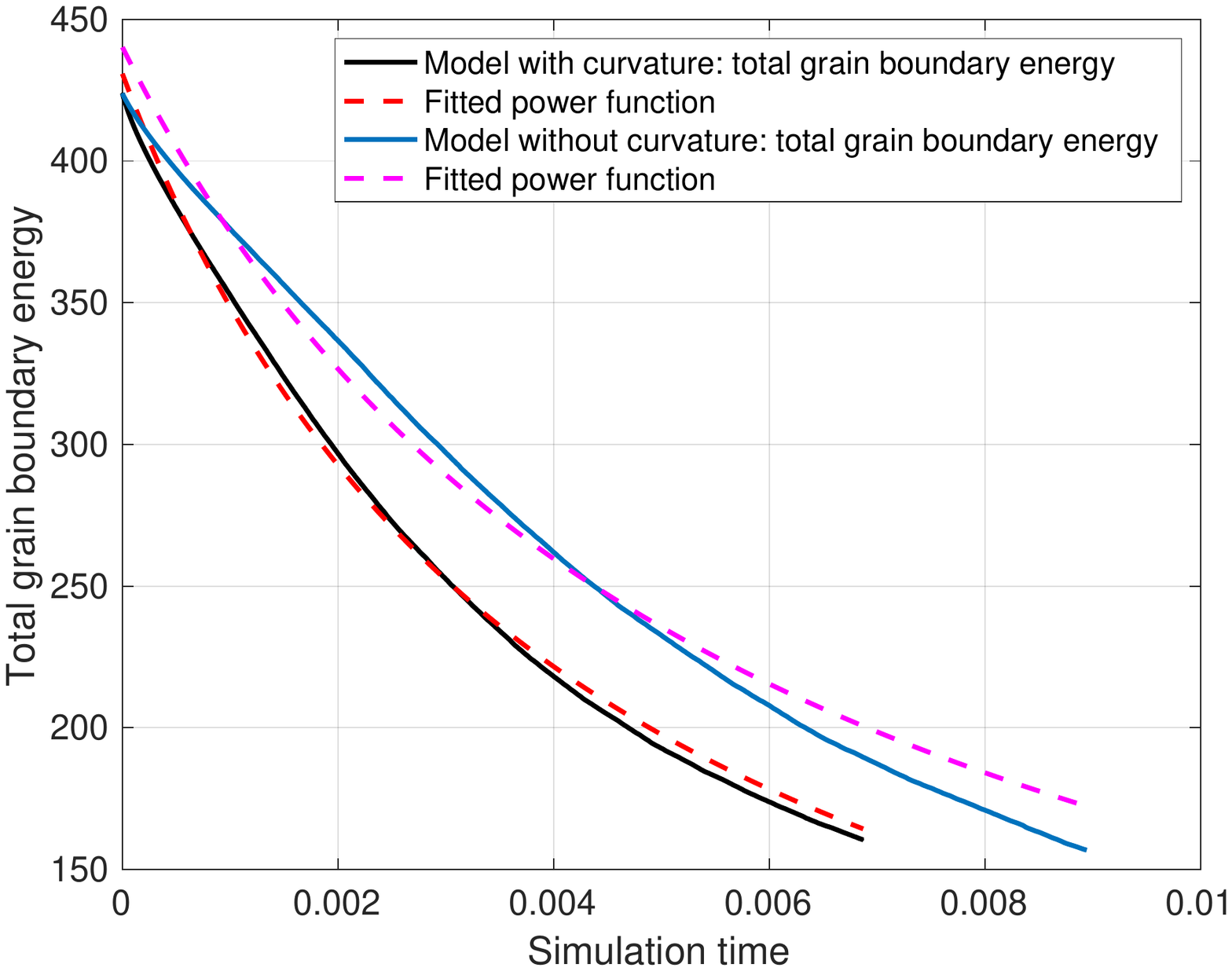}
\vspace{-1.8cm}
\caption{\footnotesize One run of $2$D trial with $10000$ initial
  grains: {\it (a) Left plot,} Total grain boundary energy plot, model
  with curvature  (solid black) versus fitted  exponential decaying function
  $y(t)=411.2\exp(-154.9t)$ (dashed red). Total grain boundary energy plot, model
  without curvature  (solid blue) versus fitted  exponential decaying function
  $y(t)=422.2\exp(-116.8t)$ (dashed magenta); {\it (b) Right
    plot},  Total grain boundary energy plot, model
  with curvature  (solid black) versus fitted  power function
  $y_1(t)=430.8310(1.0+236.0718t)^{-1}$ (dashed red). Total grain boundary energy plot, model
  without curvature  (solid blue) versus fitted power function
  $y_1(t)=440.1947(1.0+173.8526t)^{-1}$ (dashed magenta).
 Mobility of the triple
junctions is
$\eta=100$, the misorientation parameter $\gamma=250$ (curvature
model) and $\gamma=300$ (vertex model).  Grain boundary
energy density $\sigma=1+0.25\sin^2(2\Delta \alpha).$}\label{fig13}
\end{figure}

\begin{figure}[hbtp]
\centering
\vspace{-1.8cm}
\includegraphics[width=2.1in]{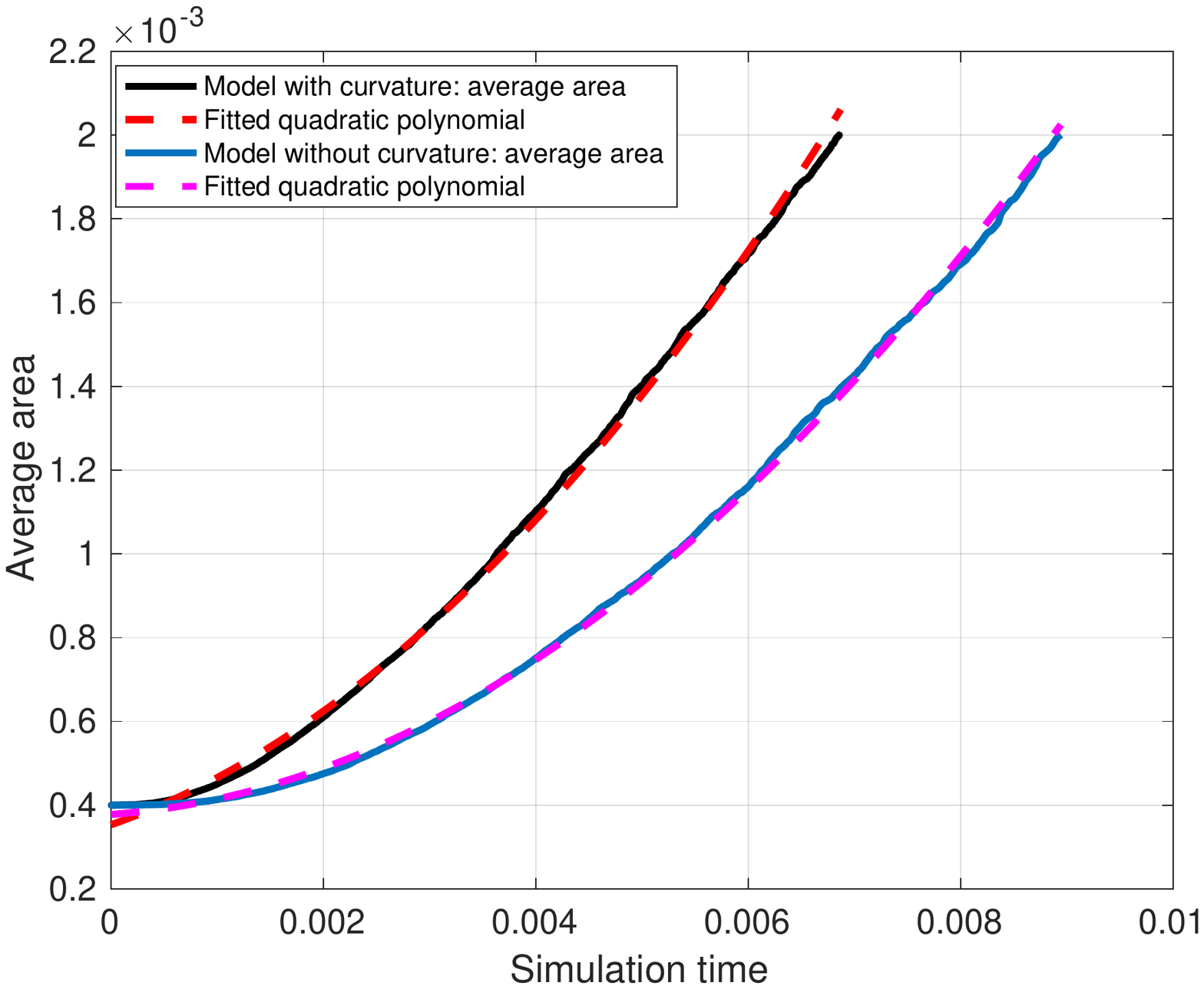}
\includegraphics[width=2.1in]{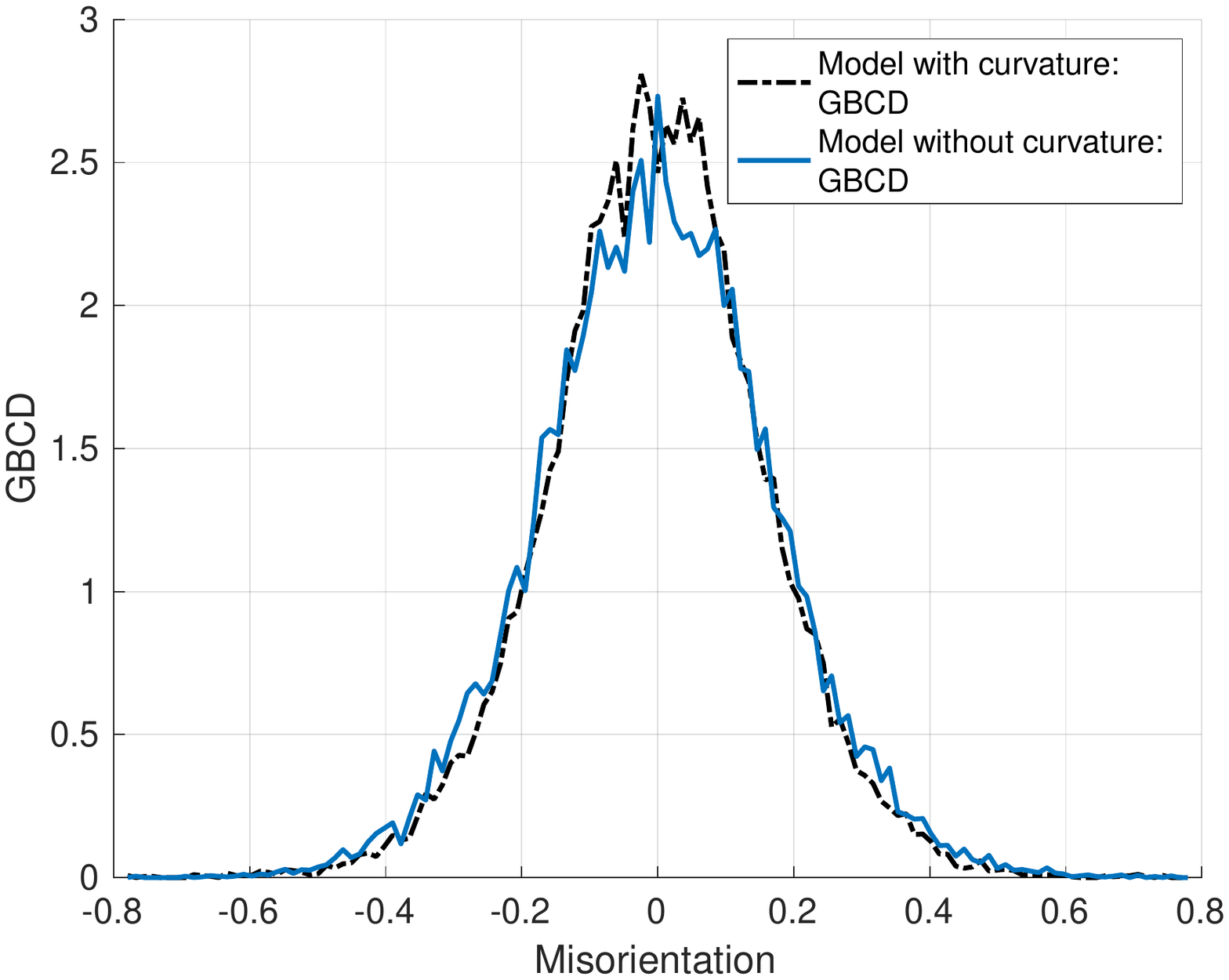}
\vspace{-1.8cm}
\caption{\footnotesize {\it (a) Left plot,} One run of $2$D trial with $10000$ initial
  grains: Growth of the average area of the
  grains, model with curvature (solid black) versus fitted  quadratic
  polynomial function  $y(t)=23.18t^2+0.08941t+0.0003532$ (dashed
  red). Growth of the average area of the
  grains, model without curvature (solid blue) versus fitted  quadratic
  polynomial function  $y(t)=18.56t^2+0.01824t+0.000378$ (dashed magenta);
 {\it (b) Right plot},  GBCD (black curve,  model with curvature) and  GBCD (blue curve,  model without curvature) at $T_{\infty}$ averaged over 3 runs of $2$D trials with $10000$ initial
  grains. Mobility of the triple junctions is $\eta=100$, the
  misorientation parameter $\gamma=250$ (curvature model) and
  $\gamma=300$ (vertex model).  Grain boundary
energy density $\sigma=1+0.25\sin^2(2\Delta \alpha).$}\label{fig14}
\end{figure}

\begin{figure}[hbtp]
\centering
\vspace{-1.8cm}
\includegraphics[width=2.1in]{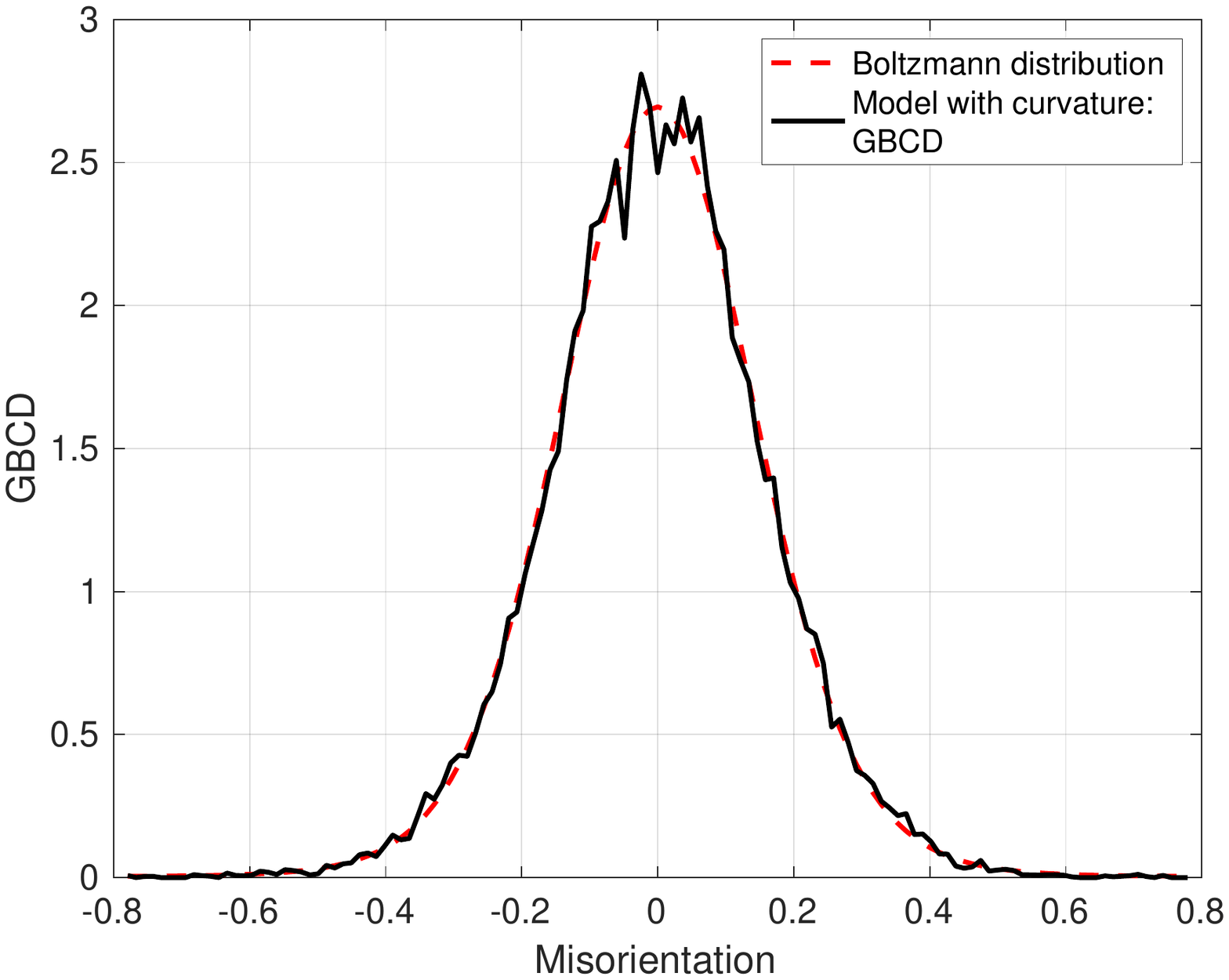}
\includegraphics[width=2.1in]{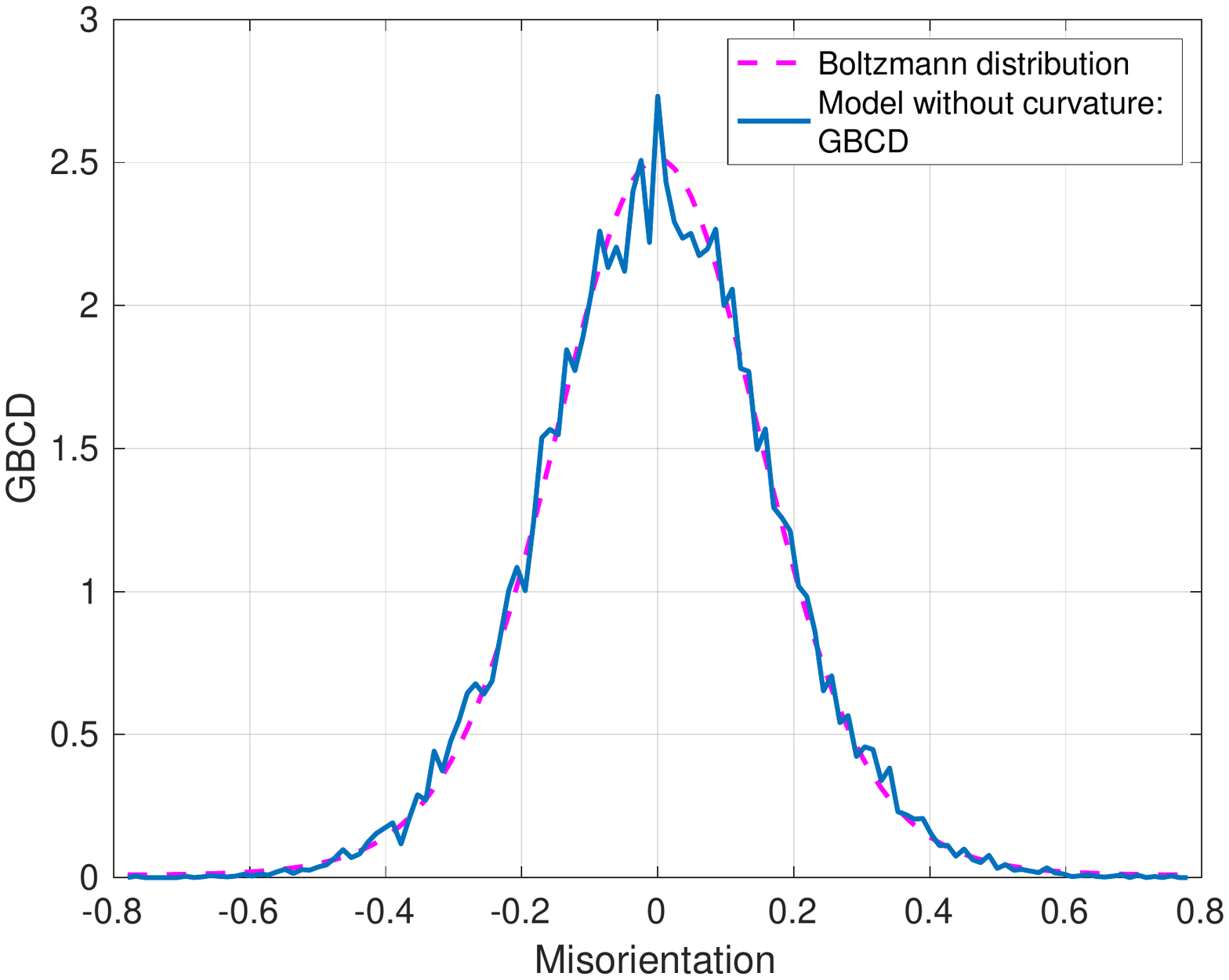}
\vspace{-1.8cm}
\caption{\footnotesize {\it (a) Left plot,} Model with curvature,  GBCD (black curve) at $T_{\infty}$ averaged over 3 runs of $2$D trials with $10000$ initial
  grains versus Boltzmann distribution with ``temperature''-
$D\approx 0.0397$ (dashed red curve). 
 {\it (b) Right plot}, Model without curvature, GBCD (blue curve) at $T_{\infty}$ averaged over 3 runs of $2$D trials with $10000$ initial
  grains versus Boltzmann distribution with ``temperature''-
$D\approx 0.0448$ (dashed magenta curve). Mobility of the triple
junctions is $\eta=100$, the misorientation parameter $\gamma=250$
(model with curvature) and $\gamma=300$ (model without curvature).  Grain boundary
energy density $\sigma=1+0.25\sin^2(2\Delta \alpha).$}\label{fig15}
\end{figure}

\begin{figure}[hbtp]
\centering
\vspace{-1.8cm}
\includegraphics[width=2.1in]{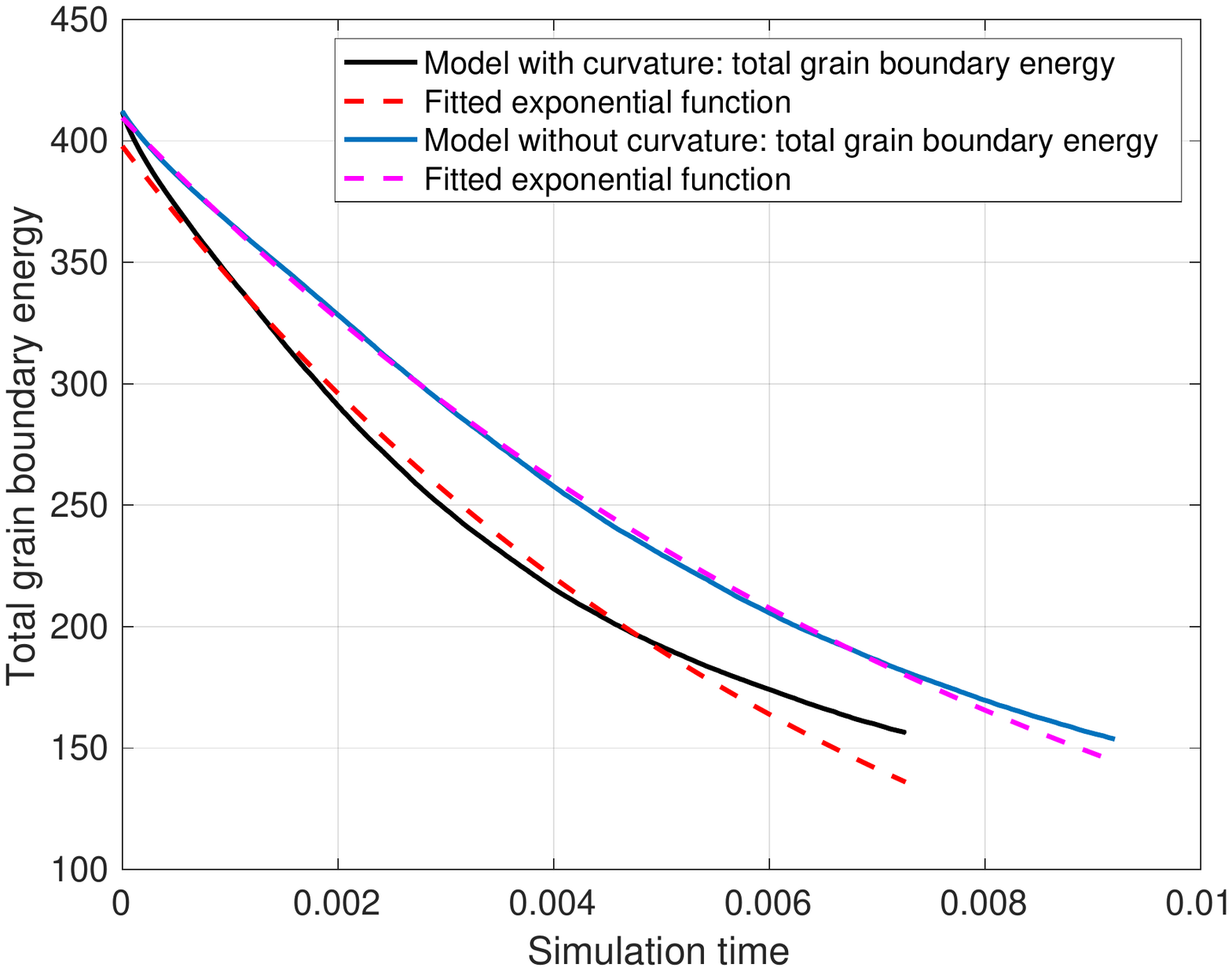}
\includegraphics[width=2.1in]{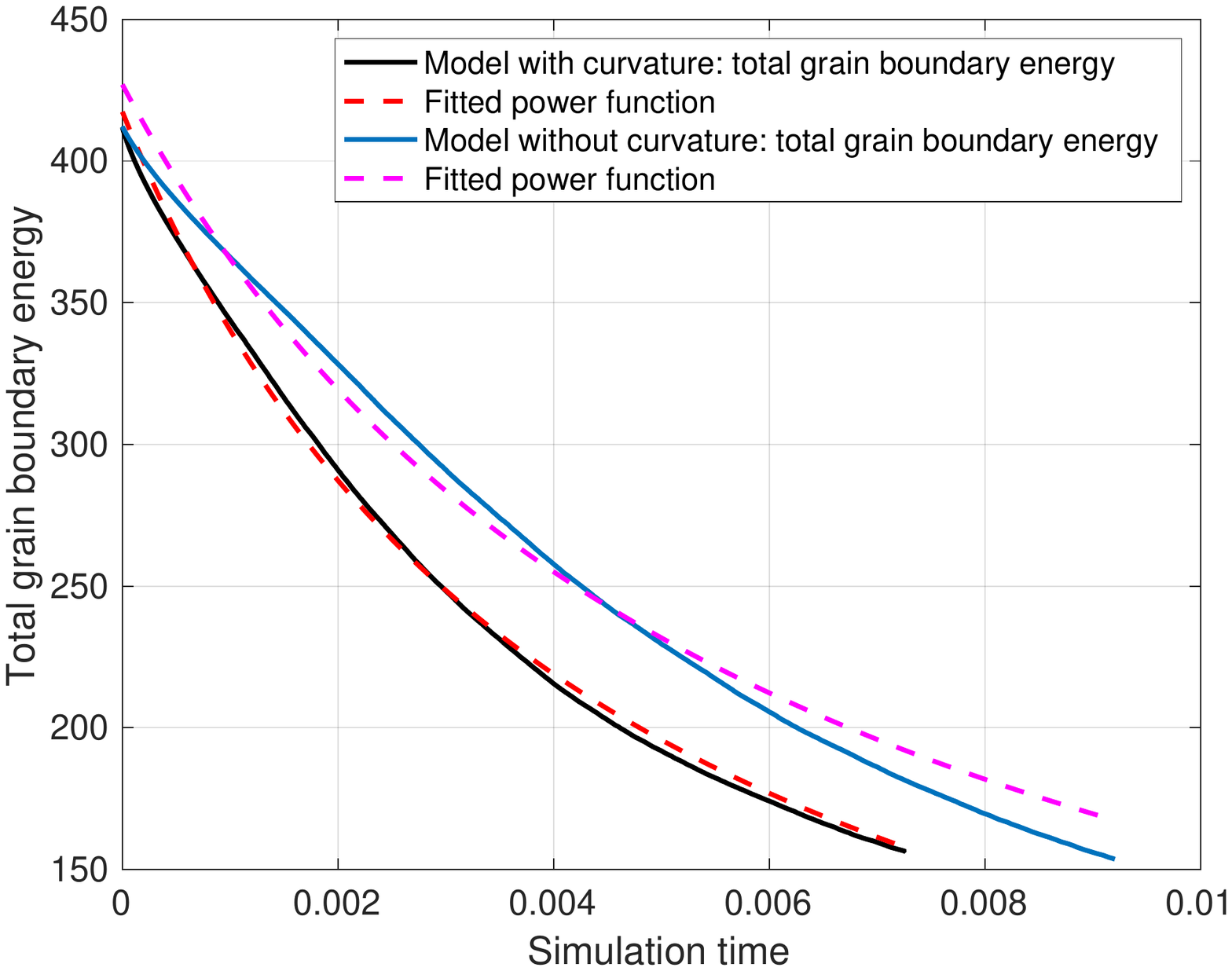}
\vspace{-1.8cm}
\caption{\footnotesize One run of $2$D trial with $10000$ initial
  grains: {\it (a) Left plot,} Total grain boundary energy plot, model
  with curvature  (solid black) versus fitted  exponential decaying function
  $y(t)=397.8\exp(-147.8t)$ (dashed red). Total grain boundary energy plot, model
  without curvature  (solid blue) versus fitted  exponential decaying function
  $y(t)=409.5\exp(-113.2t)$ (dashed magenta); {\it (b) Right
    plot},  Total grain boundary energy plot, model
  with curvature  (solid black) versus fitted  power function
  $y_1(t)=417.4031(1.0+226.6032t)^{-1}$ (dashed red). Total grain boundary energy plot, model
  without curvature  (solid blue) versus fitted power function
  $y_1(t)=427.0061(1.0+168.5772t)^{-1}$ (dashed magenta)
 Mobility of the triple
junctions is
$\eta=100$, the misorientation parameter $\gamma=1000$ (curvature
model) and $\gamma=1500$ (vertex model).  Grain boundary
energy density $\sigma=1+0.25\sin^4(2\Delta \alpha).$}\label{fig16}
\end{figure}

\begin{figure}[hbtp]
\centering
\vspace{-1.8cm}
\includegraphics[width=2.1in]{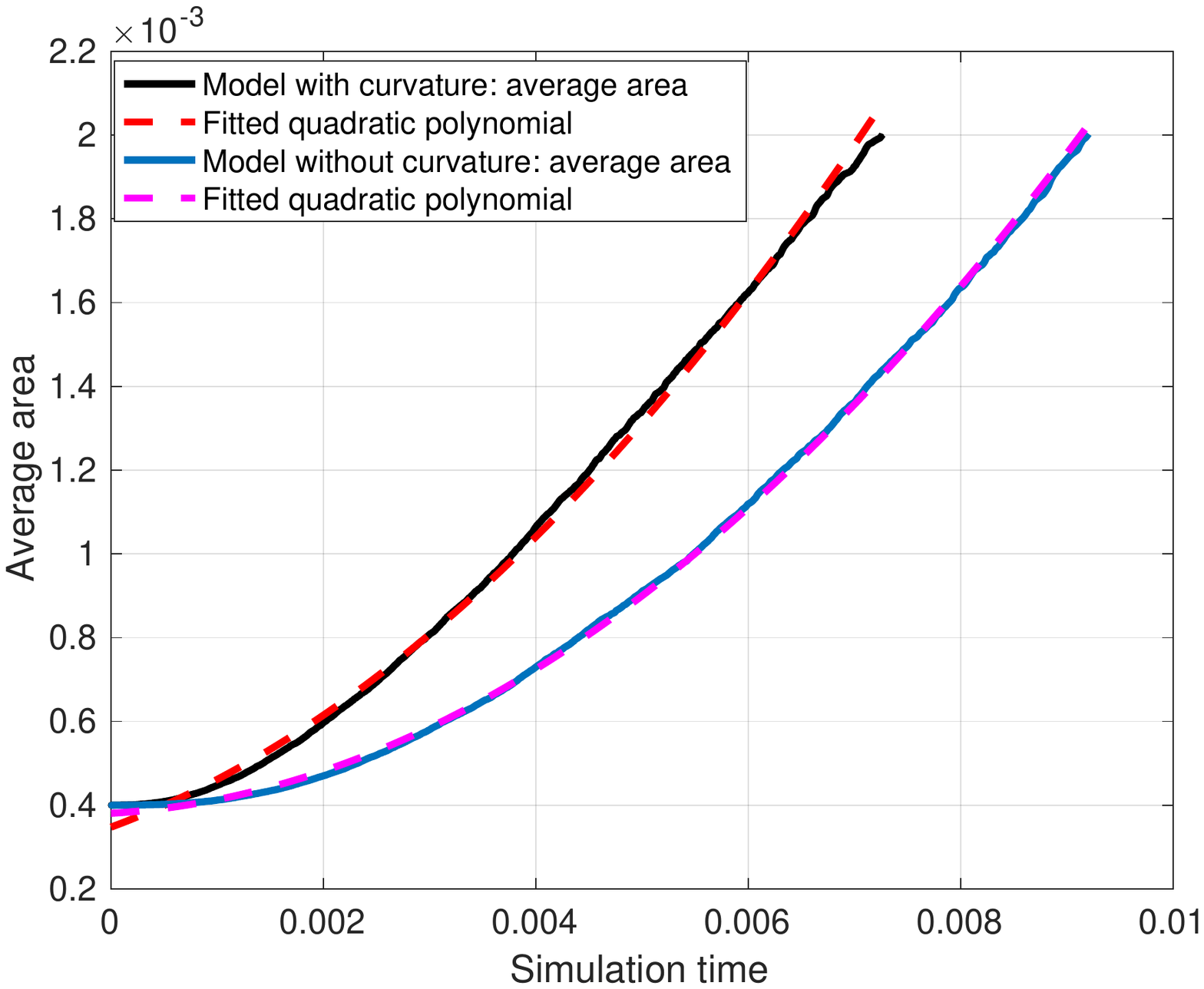}
\includegraphics[width=2.1in]{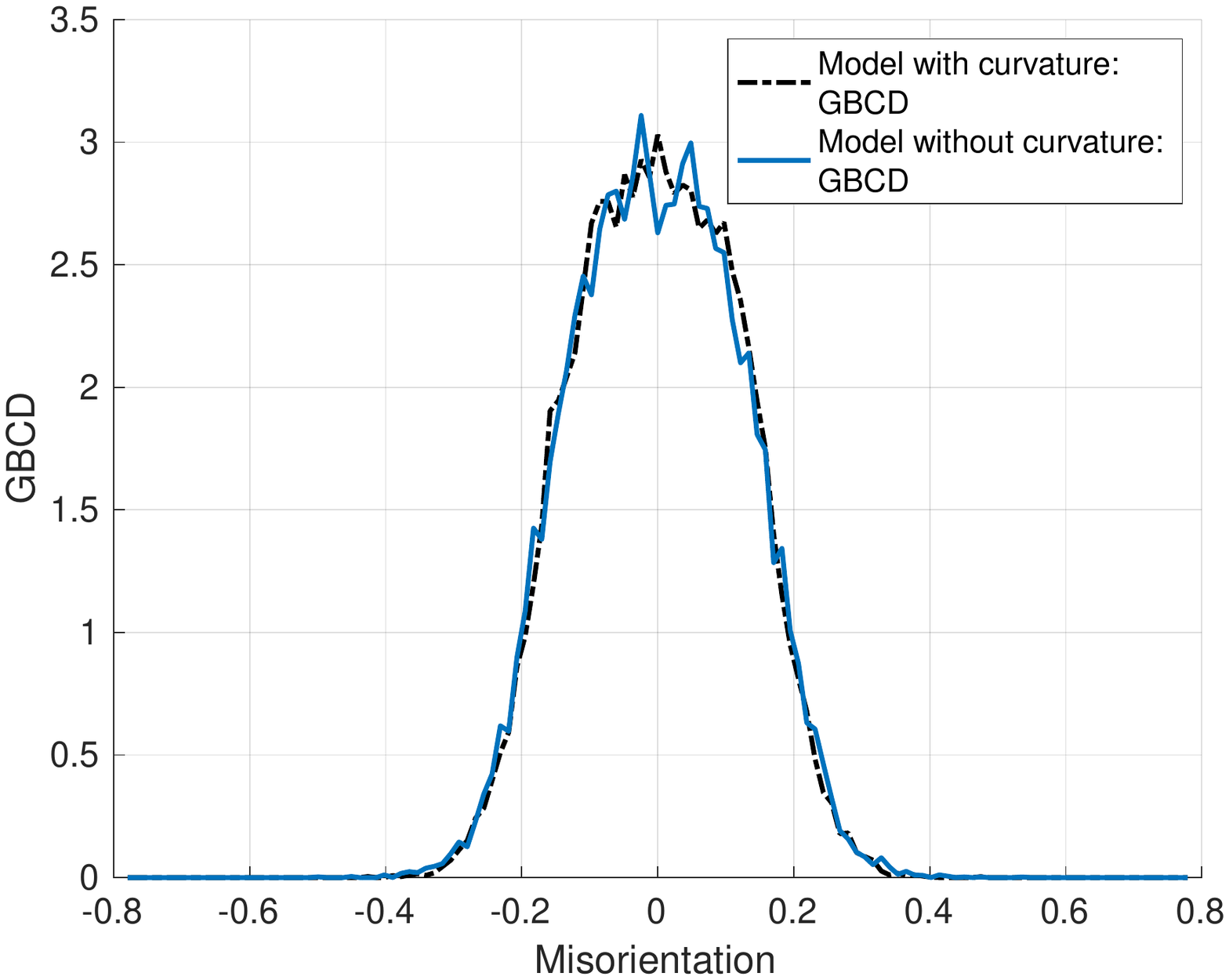}
\vspace{-1.8cm}
\caption{\footnotesize {\it (a) Left plot,} One run of $2$D trial with $10000$ initial
  grains: Growth of the average area of the
  grains, model with curvature (solid black) versus fitted  quadratic
  polynomial function  $y(t)=19.81t^2+0.09408t+0.0003476$ (dashed
  red). Growth of the average area of the
  grains, model without curvature (solid blue) versus fitted  quadratic
  polynomial function  $y(t)=17.81t^2+0.01484t+0.0003807$ (dashed magenta);
 {\it (b) Right plot}, GBCD (black curve,  model with curvature) and  GBCD (blue curve,  model without curvature) at $T_{\infty}$ averaged over 3 runs of $2$D trials with $10000$ initial
  grains. Mobility of triple junctions is $\eta=100$,  the
  misorientation parameter $\gamma=1000$ (curvature model) and
  $\gamma=1500$ (vertex model).  Grain boundary
energy density $\sigma=1+0.25\sin^4(2\Delta \alpha).$}\label{fig17}
\end{figure}

\begin{figure}[hbtp]
\centering
\vspace{-1.8cm}
\includegraphics[width=2.1in]{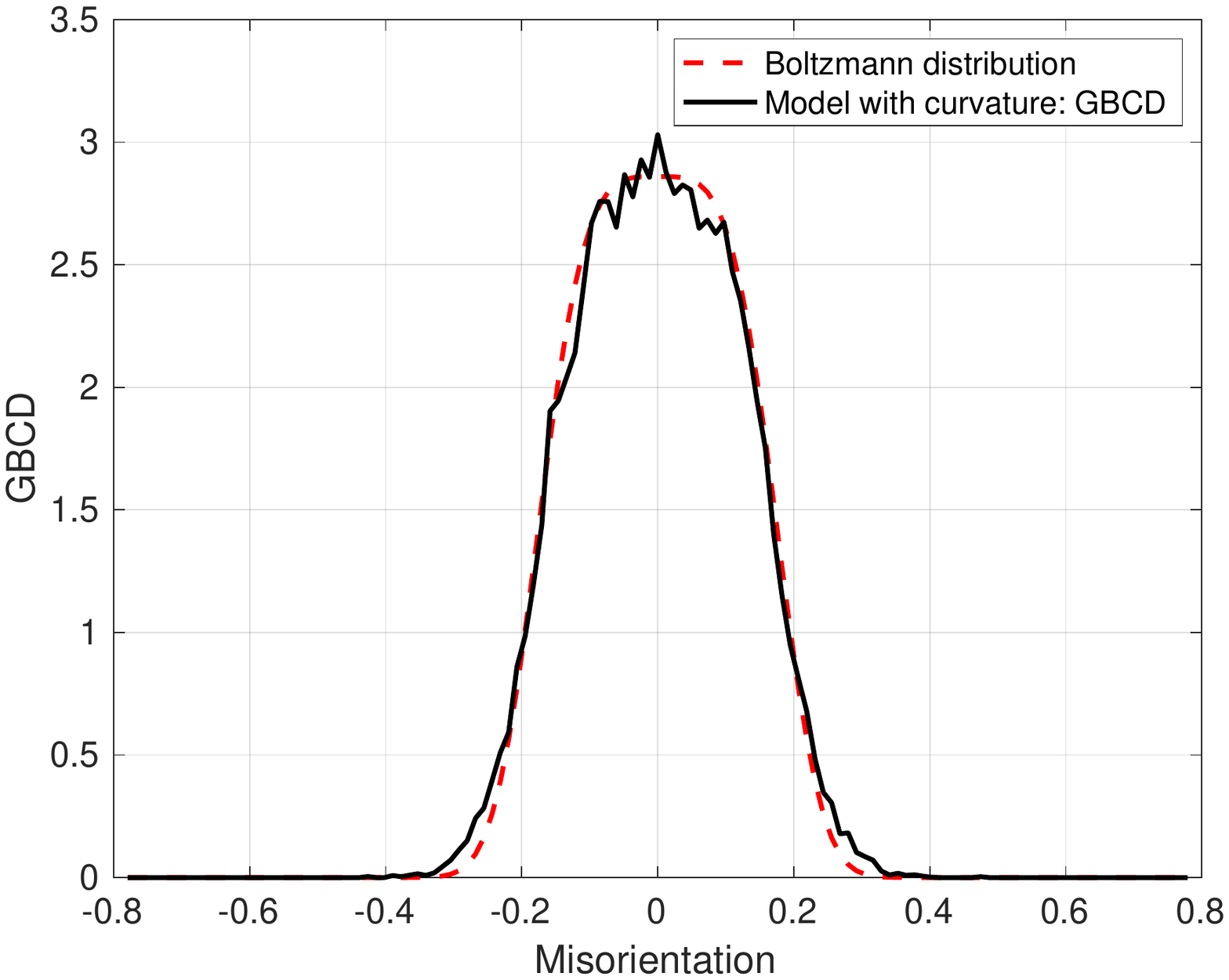}
\includegraphics[width=2.1in]{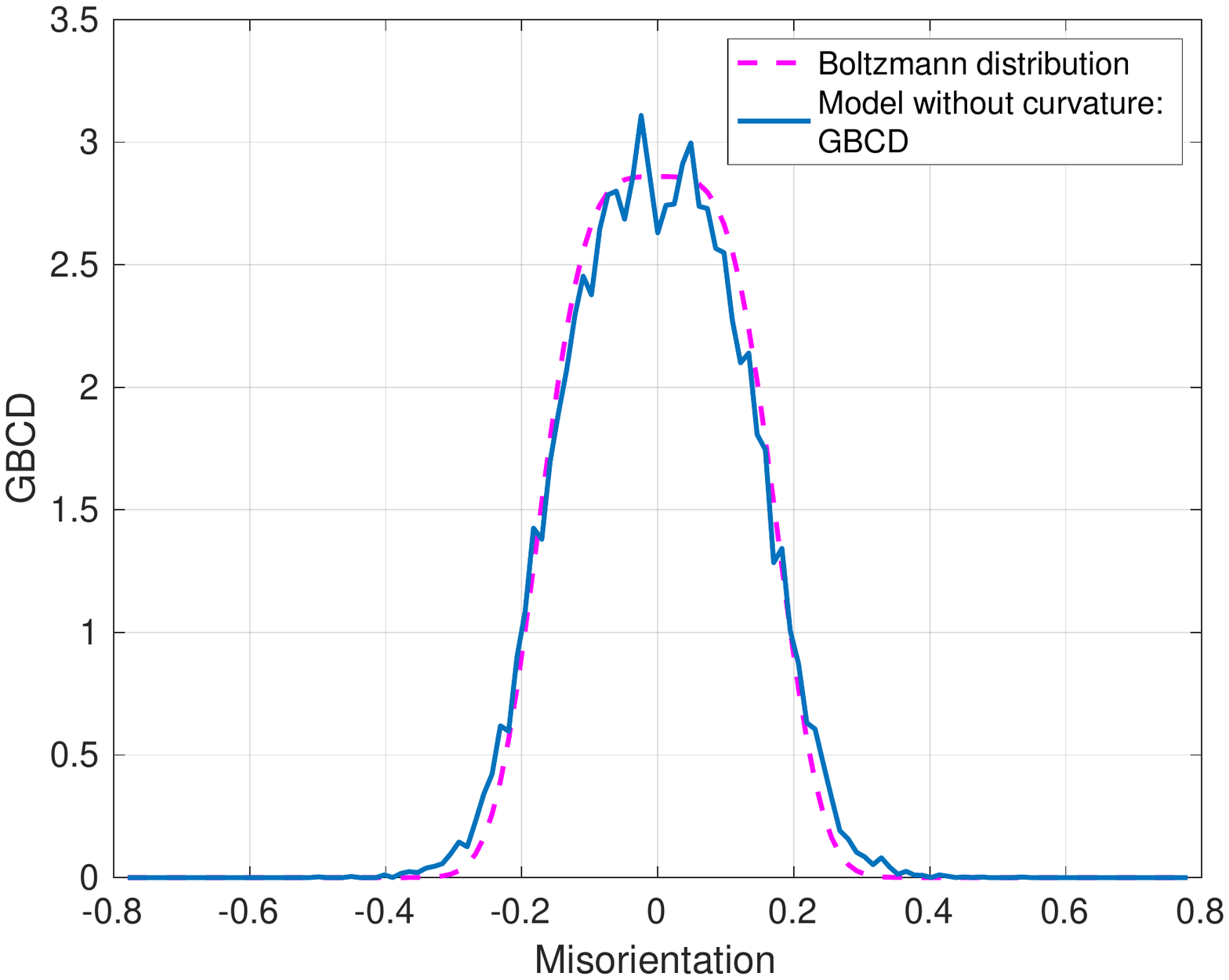}
\vspace{-1.8cm}
\caption{\footnotesize {\it (a) Left plot,} Model with curvature,  GBCD (black curve) at $T_{\infty}$ averaged over 3 runs of $2$D trials with $10000$ initial
  grains versus Boltzmann distribution with ``temperature''-
$D\approx 0.005$ (dashed red curve). 
 {\it (b) Right plot}, Model without curvature,  GBCD (blue curve) at $T_{\infty}$ averaged over 3 runs of $2$D trials with $10000$ initial
  grains versus Boltzmann distribution with ``temperature''-
$D\approx 0.005$ (dashed magenta curve). Mobility of the triple
junctions is $\eta=100$,  the misorientation parameter $\gamma=1000$
(model with curvature) and $\gamma=1500$
(model without curvature).  Grain boundary
energy density $\sigma=1+0.25\sin^4(2\Delta \alpha).$}\label{fig18}
\end{figure}

\section*{Conclusion}\label{sec:Con}
\par In this work, we conducted extensive numerical studies of the two
models developed in \cite{Katya-Chun-Mzn1, Katya-Chun-Mzn2}: a model with curved grain boundaries and a model without curvature/''vertex
 model'' of planar grain boundaries network with the dynamic lattice misorientations
 and with the drag of triple junctions.  The goal of our study was to further
 understand the effect of relaxation time scales, e.g. of the curvature of grain boundaries, mobility of triple junctions, and dynamics of misorientations on how the grain boundary system decays energy and coarsens with time. We also presented and discussed relevant experimental results of grain growth in thin films.
\section*{Acknowledgments}\label{sec:Ack}

The authors are grateful to David Kinderlehrer for the fruitful
discussions, inspiration and motivation of the work. Matthew Patrick and Amirali Zangiabadi are thanked for assistance with the experimental work.  Katayun Barmak acknowledges
partial support of NSF DMS-1905492, Yekaterina Epshteyn acknowledges
partial support of NSF DMS-1905463,  Chun Liu acknowledges partial
support of NSF DMS-1759535
and of NSF DMS-1950868, and Masashi Mizuno
acknowledges partial support of JSPS KAKENHI Grant No. 18K13446.

\bibliographystyle{plain}
\bibliography{references}
\end{document}